\documentclass[11pt]{article}
\usepackage[latin1]{inputenc}
\usepackage{graphicx}
\usepackage{epsfig}
\usepackage{color}
\usepackage{fullpage}
\usepackage{lscape}
\usepackage{verbatim}
\usepackage{amsthm, amssymb}
\usepackage{amsmath}
\newtheorem{Lem}{Lemma}
\newtheorem{theorem}{Theorem}
\newtheorem{Cor}{Corollary}

\newtheorem{Inv}{Invariant}

\newcommand{\eps}{\epsilon}
\newcommand{\MSF}[1]{\ensuremath{\mbox{MSF}({#1})}}
\newcommand{\MST}[1]{\ensuremath{\mbox{MST}({#1})}}
\newcommand{\FFE}[1]{\ensuremath{\mbox{FFE}({#1})}}

\title{Fully-Dynamic Minimum Spanning Forest with Improved Worst-Case Update Time}
\author{Christian Wulff-Nilsen
        \footnote{Department of Computer Science,
                  University of Copenhagen,
                  \texttt{koolooz@di.ku.dk},
                  \texttt{http://www.diku.dk/$_{\widetilde{~}}$koolooz/}.}}

\date{}
\begin{document}

\maketitle
\begin{abstract}
We give a Las Vegas data structure which maintains a minimum spanning forest in an $n$-vertex edge-weighted dynamic graph undergoing updates consisting of any mixture of edge insertions and deletions. Each update is supported in $O(n^{1/2 - c})$ expected worst-case time for some constant $c > 0$ and this worst-case bound holds with probability at least $1 - n^{-d}$ where $d$ is a constant that can be made arbitrarily large. This is the first data structure achieving an improvement over the $O(\sqrt n)$ deterministic worst-case update time of Eppstein et al., a bound that has been standing for nearly $25$ years. In fact, it was previously not even known how to maintain a spanning forest of an unweighted graph in worst-case time polynomially faster than $\Theta(\sqrt n)$. Our result is achieved by first giving a reduction from fully-dynamic to decremental minimum spanning forest preserving worst-case update time up to logarithmic factors. Then decremental minimum spanning forest is solved using several novel techniques, one of which involves keeping track of low-conductance cuts in a dynamic graph. An immediate corollary of our result is the first Las Vegas data structure for fully-dynamic connectivity where each update is handled in worst-case time polynomially faster than $\Theta(\sqrt n)$ w.h.p.; this data structure has $O(1)$ worst-case query time.
\end{abstract}

\newpage

\section{Introduction}\label{sec:Intro}
A minimum spanning forest (MSF) of an edge-weighted undirected graph $G$ is a forest consisting of MSTs of the connected components of $G$. Dynamic MSF is one of the most fundamental dynamic graph problems with a history spanning more than three decades. Given a graph $G$ with a set of vertices and an initially empty set of edges, a data structure for this problem maintains an MSF $F$ under two types of updates to $G$, namely the insertion or the deletion of an edge in $G$. After each update to $G$, the data structure needs to respond with the updates to $F$, if any.

%When an edge is inserted, the data structure should report whether this edge becomes an edge of the MSF and if so, report what other edge leaves the MSF, if any. When an edge is deleted, the data structure should report a cheapest reconnecting edge that joins the MSF, if any.

An MSF of a graph with $m$ edges and $n$ vertices can be computed in $O(m\alpha(m,n))$ deterministic time~\cite{StaticMSTDeterministic} and in $O(m)$ randomized expected time~\cite{StaticMSTRandomized}. Hence, each update can be handled within either of these time bounds by recomputing an MSF from scratch after each edge insertion or deletion. By exploiting the fact that the change to the dynamic graph is small in each update, better update time can be achieved.

The first non-trivial data structure for fully-dynamic MSF was due to Frederickson~\cite{Frederickson} who achieved $O(\sqrt m)$ deterministic worst-case update time where $m$ is the number of edges in the graph at the time of the update. Using the sparsification technique, Eppstein et al.~\cite{Sparsification} improved this to $O(\sqrt n)$ where $n$ is the number of vertices.

Faster amortized update time bounds exist. Henzinger an King~\cite{HK95} showed how to maintain an MSF in $O(k\log^3n)$ amortized expected update time in the restricted setting where the number of distinct edge weights is $k$. The same authors later showed how to solve the general problem using $O(\sqrt[3]n\log n)$ amortized update time~\cite{HK97}. Holm et al.~\cite{HLT01} presented a data structure for fully-dynamic connectivity with $O(\log^2n)$ amortized update time and showed how it can easily be adapted to handle decremental (i.e., deletions only) MSF within the same time bound. They also gave a variant of a reduction of Henzinger and King~\cite{HK97a} from fully-dynamic to decremental MSF and combining these results, they obtained a data structure for fully-dynamic MSF with $O(\log^4n)$ amortized update time. This bound was slightly improved to $O(\log^4n/\log\log n)$ in~\cite{DynMSFESA}. A lower bound of $\Omega(\log n)$ was shown in~\cite{PatrascuDemaine} and this bound holds even for just maintaining the weight of an MSF in a plane graph with unit weights.

\subsection{Our results}
In this paper, we give a fully-dynamic MSF data structure with a polynomial speed-up over the $O(\sqrt n)$ worst-case time bound of Eppstein et al. Our data structure is Las Vegas, always correctly maintaining an MSF and achieving the polynomial speed-up w.h.p.~in each update. The following theorem states our main result.
\begin{theorem}\label{Thm:Main}
There is a Las Vegas data structure for fully-dynamic MSF which for an $n$-vertex graph has an expected update time of $O(n^{1/2 - c})$ for some constant $c > 0$; in each update, this bound holds in the worst-case with probability at least $1 - n^{-d}$ for a constant $d$ that can be made arbitrarily large.
\end{theorem}
We have not calculated the precise value of constant $c$ but it is quite small. From a theoretical perspective however, the $O(\sqrt n)$ bound is an important barrier to break. Furthermore, a polynomial speed-up is beyond what can be achieved using word parallelism alone unless we allow a word size polynomial in $n$. Indeed, our improvement does not rely on a more powerful model of computation than what is assumed in previous papers. To get our result, we develop several new tools some of which we believe could be of independent interest. We sketch these tools later in this section.

As is the case for all randomized algorithms and data structures, it is important that the random bits used are not revealed to an adversary. It is well-known that if all edge weights in a graph are unique, its MSF is uniquely defined. Uniqueness of edge weights can always be achieved using some lexicographical ordering in case of ties. This way, our data structure can safely reveal the MSF after each update without revealing any information about the random bits used.

\paragraph{Dynamic connectivity:}
An immediate corollary of our result is a fully-dynamic data structure for maintaining a spanning forest of an unweighted graph in worst-case time $O(n^{1/2 - c})$ with high probability. The previous best worst-case bound for this problem was $O(\sqrt n)$ by Eppstein et al.\cite{Sparsification}; if word-parallelism is exploited it, a slightly better bound of $O(\sqrt{n(\log\log n)^2/\log n})$ was shown by Kejlberg-Rasmussen et al.~\cite{PolyLogWorstCaseSpeedup}. There are Monte Carlo data structures for fully-dynamic connectivty by Kapron et al.\cite{WorstCasePolyLog} and by Gibb et al.\cite{Gibb} which internally maintain a spanning forest in polylogarithmic time per update. However, contrary to our data structure, these structures cannot reveal the spanning forest to an adversary. Kapron et al.~extend their result to maintaining an MSF in $\tilde O(L)$ time\footnote{We use $\tilde O$, $\tilde\Omega$, and $\tilde\Theta$ when suppressing $\log n$-factors.} per update where $L$ is the number of distinct weights. However, their data structure can only reveal the weight of this MSF. Furthermore, if all edge weights are unique, this bound becomes $\tilde O(m)$.

From our main result, we also immediately get the first Las Vegas fully-dynamic connectivity structure achieving w.h.p.~a worst-case update time polynomially faster than $\sqrt n$, improving the previous best Las Vegas bounds of Eppstein et al.\cite{Sparsification} and Kejlberg-Rasmussen et al.~\cite{PolyLogWorstCaseSpeedup}. By maintaining the spanning forest using a standard dynamic tree data structure with polynomial fan-out, our connectivity structure achieves constant worst-case query time.

\paragraph{Monte Carlo data structure:}
It is easy to modify our Las Vegas structure to a Monte Carlo structure which is guaranteed to handle each update in $O(n^{1/2 - c})$ worst-case time. This is done by simply terminating an update if the $O(n^{1/2 - c})$ time bound is exceeded by some constant factor $C$. By picking $C$ sufficiently large, we can ensure that this termination happens only with low probability in each update. An issue here is that once the Monte Carlo structure makes an error, subsequent updates are very likely to also maintain an incorrect MSF. This can be remedied somewhat by periodically rebuilding new MSF structures so that after a small number of updates, the data structure again maintains a correct MSF with high probability; we omit the details as our focus is on obtaining a Las Vegas structure.

\subsection{High-level description and overview of paper}\label{subsec:HighLevel}
In the rest of this section, we give an overview of our data structure as well as how the paper is organized. The description of our data structure here will not be completely accurate and we only highlight the main ideas.

Section~\ref{sec:Prelim} introduces some definitions and notation that will be used throughout the paper.

\paragraph{Restricted Decremental MSF Structure (Section~\ref{sec:RestrictedMSF})}
In Section~\ref{sec:RestrictedMSF}, we present a data structure for a restricted version of decremental MSF where the initial graph has max degree at most $3$ and where there is a bound $\Delta$ on the total number of edge deletions where $\Delta$ may be smaller than the initial number of edges.

The data structure maintains a recursive clustering of the dynamic graph $G = (V,E)$ where each cluster is a subgraph of $G$. This clustering forms a laminar family $\mathcal F$ (w.r.t.~subgraph containment) and can be represented as a rooted tree where the root corresponds to the entire graph $G$; for technical reasons, we refer to the root as a level $-1$-cluster and the children of an $i$-cluster are referred to as level $(i+1)$-clusters. The decremental MSF structure of Holm et al.~\cite{HLT01} also maintains a recursive clustering but ours differs significantly from theirs, as will become clear.

In our recursive clustering, the vertex sets of the level $0$-clusters form a partition $V$ and w.h.p., each level $0$-cluster is an expander graph and the number of inter-cluster edges is small. More specifically, the expansion factor of each expander graph is of the form $n^{-c_1}$ and the number of inter-cluster edges is at most $n^{1 - c_2}$ for some small positive constants $c_1$ and $c_2$. Such a partition is formed with a new algorithm that we present in Section~\ref{sec:PartitionExpanders}.

Next, consider a list of the edges of $E$ sorted by decreasing weight. This list is partitioned into $\ell = m^{\eps}$ sublists each of size $m/\ell$ for some small constant $\eps > 0$. These sublists correspond to suitable subsets $E_0,\ldots,E_{\ell-1}$ ordered by decreasing weight.

Each level $i$-cluster $C$ contains only edges from $E_i\cup\ldots\cup E_{\ell-1}$. To form the children of $C$ in $\mathcal F$, we remove from $C$ the edges in $E_i$ and partition the remaining graph into expander graphs as above; these expander graphs are then the children of $C$. The recursion stops when $C$ has size polynomially smaller than $n$.

Next, we form a new graph $G'$ from $G$ as follows. Initially, $G' = G$. For each $i$ and for each level $i$-cluster $C$, all the edges of $C - E_i$ between distinct child clusters of $C$ are added to an auxiliary structure $\mathcal M'$ that we describe below. In $G'$, their edge weights are artificially increased to a value which is smaller than the weight of any edge of $G$ in $E_{i+1}\cup\ldots\cup E_{\ell}$ and heavier than the weight of any edge of $G$ in $E_0\cup\ldots\cup E_i$. The edges added to $\mathcal M'$ keep their original weights in $G$. An example is shown in Figure~\ref{fig:NestingMSF}.

Now, we have an auxiliary structure $\mathcal M'$ containing a certain subset $E'$ of edges of $G$ and a recursive clustering of $G'$. Because of the way we defined edge weights in $G'$, an MSF $M'$ of this graph has the nice property that it is consistent with the recursive clustering: for any cluster $C$, $M'$ restricted to $C$ is an MSF of $C$. This could also have been achieved if we had simply deleted the edges from $G'$ whose weights were artificially increased above; however, it is important to keep them in $G'$ in order to preserve the property that clusters are expander graphs.

Assuming for now that clusters do not become disconnected during updates, it follows from this property that we can maintain $M'$ by maintaining an MSF for each level independently where level $(i+1)$-clusters are regarded as vertices of the MSF at level $i$. The global MSF $M'$ is then simply the union of (the edges of) these MSFs. Each edge deletion in $G$ only requires an MSF at one level to be updated and we show that the number of edges at this level is polynomially smaller than $n$, allowing us to maintain $M$' in time polynomially faster than $\Theta(\sqrt n)$.

We add the edges of $M'$ to $\mathcal M'$. In order to maintain an MSF $M$ of $G$, we show that it can be maintained as an MSF of the edges added to $\mathcal M'$. This follows easily from observations similar to those of Eppstein et al.~\cite{Sparsification} combined with the fact that any edge that was increased in $G'$ belongs to $\mathcal M'$ with its original weight. We show that the number of non-tree edges in the graph maintained by $\mathcal M'$ is polynomially smaller than $n$. $\mathcal M'$ is an instance of a new data structure (Section~\ref{sec:FewNonTreeEdges}) which maintains an MSF of a graph in $\tilde O(\sqrt h)$ worst-case time per update where $h$ is an upper bound on the number of non-tree edges ever present in the graph. Hence, maintaining $M$ can be done in time polynomially faster than $\Theta(\sqrt n)$.

The main obstacle to overcome is to handle disconnected clusters. If a level $(i+1)$-cluster becomes disconnected, this may affect the MSF at level $i$ and changes can propagate all the way down to level $-1$ (similar to what happens in the data structure in~\cite{HLT01}). Our analysis sketched above then breaks down. However, this is where we exploit the fact that w.h.p., each cluster $C$ is initially an expander graph. This implies that, assuming the total number $\Delta$ of edge deletions is not too big, $C$ can only become disconnected along a cut where one side is small.

Whenever an edge has been deleted from a cluster $C$, a data structure (Sections~\ref{sec:DecDSLowCondCuts},~\ref{sec:LowCondCutsSparsification}, and~\ref{sec:XPrune}) is applied which ``prunes'' off parts of $C$ so that w.h.p., the pruned $C$ remains an expander graph. Because of the property above, only small parts need to be pruned off. As we show, this can be handled efficiently for $\Delta$ polynomially slightly bigger than $\sqrt n$. With a reduction (Section~\ref{sec:Reduction}) from fully-dynamic MSF to the restricted decremental MSF problem with this value of $\Delta$, the main result of the paper follows.

\paragraph{Reduction to decremental MSF (Section~\ref{sec:Reduction})}
In Section~\ref{sec:Reduction}, we give a reduction from fully-dynamic MSF to a restricted version of decremental MSF where the initial $n$-vertex graph has degree at most $3$ and where the total number of edge deletions allowed is bounded by a parameter $\Delta = \Delta(n)$. The reduction is worst-case time-preserving, meaning roughly that if we have a data structure for the restricted decremental MSF problem with small worst-case update time then we also get a data structure for fully-dynamic MSF with small worst-case update time. This is not the case for the reduction presented in~\cite{HLT01} since it only ensures small amortized update time for the fully-dynamic structure.

More precisely, our reduction states that if the data structure for the restricted decremental problem has preprocessing time $P(n)$ and worst-case update time $U(n)$ then there is a fully-dynamic structure with worst-case update time $\tilde O(P(n)/\Delta(n) + U(n))$.

To get this result, we modify the reduction of Holm et al.~\cite{HLT01}. In their reduction, $O(\log n)$ decremental structures (which do not have a $\Delta$-bound on the total number of edge deletions) are maintained. During updates, new decremental structures are added and other decremental structures are merged together. The main reason why this reduction is not worst-case time-preserving is that a merge is done during a single update and this may take up to linear time.

We modify the reduction using a fairly standard deamortization trick of spreading the work of merging decremental structures over multiple updates. This gives the desired worst-case time-preserving reduction from fully-dynamic to decremental MSF. We then show how to further reduce the problem to the restricted variant considered in Section~\ref{sec:RestrictedMSF}.

\paragraph{Fully-dynamic MSF with few non-tree edges (Section~\ref{sec:FewNonTreeEdges})}
In Section~\ref{sec:FewNonTreeEdges}, we present a fully-dynamic MSF structure which has an update time of $\tilde O(\sqrt h)$ where $h$ is an upper bound on the number of non-tree edges ever present in the graph. At a high level, this structure is similar to that of Frederickson~\cite{Frederickson} in that it maintains a clustering of each tree of the MSF $M$ into subtrees of roughly the same size. However, because of the bound on the number of non-tree edges, we can represent $M$ in a more compact way as follows. Consider the union of all paths in $M$ between endpoints of non-tree edges. In this subforest $M'$ of $M$, consider all maximal paths whose interior vertices have degree $2$. The compact representation is obtained from $M'$ by replacing each such path by a single ``super edge''; see Figure~\ref{fig:FewNonTreeEdges}. The compact version of $M'$ only has size $O(h)$.

The update time for Frederickson's structure is bounded by the maximum of the number of clusters and the size of each cluster so to get the $(\sqrt m)$ bound, his structure maintains $O(\sqrt m)$ clusters each of size $O(\sqrt m)$. We use essentially the same type of clustering as Frederickson but for the compact representation of $M$, giving $O(\sqrt h)$ clusters each of size $O(\sqrt h)$. Using a data structure similar to Frederickson for the compact clustering, we show that $M$ can be maintained in $\tilde O(\sqrt h)$ worst-case time per update. Here we get some additional log-factors since we make use of the top tree data structure in~\cite{TopTree} to maintain, e.g., the compact representation of $M$.

\paragraph{Partitioning a graph into expander subgraphs (Section~\ref{sec:PartitionExpanders})}
In Section~\ref{sec:PartitionExpanders}, we present a near-linear time algorithm to partition the vertex set $V$ of an $n$-veretx constant-degree graph such that w.h.p., each set in this partition induces an $n^{-c_1}$-expander graph and the number of edges between distinct sets is $n^{1 - c_2}$ for suitable positive constants $c_1$ and $c_2$. The algorithm is a recursive variant of the \texttt{Partition} algorithm of Spielman and Teng~\cite{SpielmanTeng}.

For our application of this result in Section~\ref{sec:RestrictedMSF}, we need each expander graph $H$ to respect a given partition $\mathcal C$ of $V$, meaning that each $C\in\mathcal C$ is either contained in $V(H)$ or disjoint from $V(H)$. Ensuring this is a main technical challenge in this section.

\paragraph{Decremental Maintenance of Expander Graphs (Sections~\ref{sec:DecDSLowCondCuts},~\ref{sec:LowCondCutsSparsification}, and~\ref{sec:XPrune})}
In Section~\ref{sec:DecDSLowCondCuts}, we present a decremental data structure which, given an initial expander graph $H$ of degree at most $3$ (such as one from Section~\ref{sec:RestrictedMSF}), outputs after each update a subset of vertices such that at any point, there exists a subset $W$ of the set of vertices output so far so that $H[V(H) - W]$ is guaranteed to be connected; furthermore, w.h.p., the set output in each update is small. As we show, this is exactly what is needed in Section~\ref{sec:RestrictedMSF} where we require clusters to be connected at all times and where the vertices pruned off each cluster is small in each update.

This data structure relies on a procedure in Section~\ref{sec:XPrune} which we refer to as \texttt{XPrune}. It detects low-conductance cuts in a decremental graph (which is initially an expander graph) and prunes off the smaller side of such a cut while retaining the larger side.

\texttt{XPrune} uses as a subroutine the procedure \texttt{Nibble} of Spielman and Teng~\cite{SpielmanTeng}. Given a starting vertex $s$ in a (static) graph, \texttt{Nibble} computes (approximate) probability distributions for a number of steps in a random walk from $s$. For each step, \texttt{Nibble} attempts to identify a low-conductance cut based on the probability mass currently assigned to each vertex. Spielman and Teng show that if the graph has a low-conductance cut then \texttt{Nibble} will find such a cut for at least one choice of $s$.

In Section~\ref{sec:XPrune}, we show how to adapt \texttt{Nibble} from a static to a decremental setting roughly as follows. In the preprocessing step, \texttt{Nibble} is started from every vertex in the graph and if a low-conductance cut is found, the smaller side is pruned off. Now, consider an update consisting of the deletion of an edge $e$. We cannot afford to rerun \texttt{Nibble} from every vertex as in the preprocessing step. Instead we show that there is only a small set of starting vertices for which \texttt{Nibble} will have a different execution due to the deletion of $e$. We only run \texttt{Nibble} from starting vertices in this small set; these vertices can easily be identified since they are exactly those for which \texttt{Nibble} in some step sends a non-zero amount of probability mass along $e$ in the graph just prior to the deletion.

Hence, we implicitly run \texttt{Nibble} from every starting vertex after each edge deletion so if there is a low-conductance cut, \texttt{XPrune} is guaranteed to find such a cut. When the smaller side of a cut is pruned off, a similar argument as sketched above implies that \texttt{Nibble} only needs to be rerun from a small number of starting vertices on the larger side.

In order to have \texttt{XPrune} run fast enough, we need an additional trick which is presented in Section~\ref{sec:LowCondCutsSparsification}. Here we show that w.h.p., the conductance of every cut in a given multigraph is approximately preserved in a subgraph obtained by sampling each edge independently with probability $p$; this assumes that $p$ and the min degree of the original graph are not too small. This is somewhat similar to Karger's result that the value of each cut is preserved in a sampled subgraph~\cite{Karger}. We make use of this new result in Section~\ref{sec:XPrune} where we run \texttt{Nibble} on the sampled subgraph rather than the full graph. Combined with the above implicit maintenance of calls to \texttt{Nibble}, this gives the desired performance of \texttt{XPrune}.

We conclude the paper in Section~\ref{sec:ConclRem}.

%\begin{figure}%[!ht]
%\centerline{\scalebox{0.80}{\input{ExpanderClusters.pstex_t}}}
%\caption{}
%\label{fig:ExpanderClusters}
%\end{figure}

\section{Preliminaries}\label{sec:Prelim}
We consider only finite undirected graphs and unless otherwise stated, they are simple. An edge-weighted graph is written on the form $G = (V,E,w)$ where $w:E\rightarrow\mathbb R$; we sometimes simply write $G = (V,E)$ even if $G$ is edge-weighted.

For a simple graph or a multigraph $H$, $V(H)$ denotes its vertex set and $E(H)$ denotes its edge set. If $H$ is edge-weighted, we regard any subset $E$ of $E(H)$ as a set of weighted edges and if the edge weight function $w: E(H)\rightarrow\mathbb R$ for $H$ is not clear from context, we write $E(w)$ instead of $E$. We sometimes abuse notation and regard $E(H)$ as a graph with edge set $E(H)$ and vertex set consisting of the endpoints of edges in $E(H)$. When convenient, we regard the edge set of a minor of $H$ as a subset of $E(H)$ in the natural way.

Given two edge-weighted graphs $G_1 = (V_1,E_1,w_1)$ and $G_2 = (V_2,E_2,w_2)$, we let $G_1\cup G_2$ denote the multigraph with vertex set $V_1\cup V_2$ and edge set $E_1\cup E_2$; if both $E_1$ and $E_2$ contain an edge between the same vertex pair $(u,v)$, we keep both edges in $G_1\cup G_2$, one having weight $w_1(u,v)$ and the other having weight $w_2(u,v)$.

In the rest of this section, let $G = (V,E,w)$ be an edge-weighted graph. A \emph{component} of $G$ is a connected component of $G$ and we sometimes regard it as a subset of $V$. For $W\subseteq V$, $G[W]$ is the subgraph of $G$ induced by $W$. When $V$ is clear from context, we say that $W$ \emph{respects} another subset $C$ of $V$ if either $C\subseteq W$ or $C\cap W = \emptyset$. We extend this to a collection $\mathcal C$ of subsets of $V$ and say that $W$ respects $\mathcal C$ if $W$ respects each set in $\mathcal C$; in this case, we let $\mathcal C(W)$ denote the collection of sets of $\mathcal C$ that are contained in $W$. For a subgraph $H$ of $G$, we say that $H$ respects $C$ resp.~$\mathcal C$ if $V(H)$ respects $C$ resp.~$\mathcal C$.

A \emph{cut} of $G$ or of $V$ is a pair $(V_1,V_2)$ such that $V_1\cup V_2 = V$ and $V_1\cap V_2 = \emptyset$. When $V$ is clear from context, we identify a cut $(V_1,V_2)$ with $V_1$ or with $V_2$.

For a subset $S$ of $V$, denote by $\delta_G(S)$ the number of edges of $E$ crossing the cut $(S, V - S)$, i.e., $\delta_G(S) = |E\cap S\times(V-S)|$. The \emph{volume} $\mbox{Vol}_G(S)$ of $S$ in $G$ is the number of edges of $G$ incident to $S$. Assuming both $S$ and $V - S$ have positive volume in $G$, the \emph{conductance} $\Phi_G(S)$ of $S$ (or of $(S, V - S)$) is defined as $\Phi_G(S) = \delta_G(S)/\min\{\mbox{Vol}_G(S),\mbox{Vol}_G(V - S)\}$ (this is called sparsity in~\cite{SpielmanTeng}). When $G$ is clear from context, we define, for $S\subseteq W\subseteq V$, $\delta_W(S) = \delta_{G[W]}(S)$, $\mbox{Vol}_W(S) = \mbox{Vol}_{G[W]}(S)$, and $\Phi_W(S) = \Phi_{G[W]}(S)$. We extend the definitions in this paragraph to multigraphs in the natural way.

Given a real value $\gamma > 0$, we say that $G$ is a \emph{$\gamma$-expander graph} and that $G$ has \emph{expansion} $\gamma$ if for every cut $(S, V - S)$, $\delta_G(S)\geq\gamma\min\{|S|, |V - S|\}$. Note that if $G$ is connected and has constant degree then $\Phi_G(S) = \Theta(\delta_G(S)/\min\{|S|, |V - S|\})$ for every $S\notin\{\emptyset,V\}$; thus, in this special case, $G$ has expansion $\Theta(\gamma)$ iff every such cut has conductance $\Omega(\gamma)$.

%Let $v$ be a degree two vertex in a graph $G$. By \emph{shortcutting} $v$ (in $G$), we mean replacing its incident edges $(u_1,v)$ and $(u_2,v)$ with the single edge $(u_1,u_2)$. When we talk about shortcutting a vertex, we implicitly assume that it has degree two.% A \emph{shortcut} of a graph $G$ is a graph obtained by a sequence of shortcuttings of $G$. If the sequence is maximal, we refer to the graph as the \emph{maximal shortcut}.

%Let $G = (V,E)$ be a graph. Consider a graph $H = (V',E')$ with $V'\subseteq V$ where each edge $(u,v)$ of $E'$ corresponds to a simple path $P(u,v)$ in  $G$ intersecting $V'$ only in $u$ and $v$ and where any two such paths are vertex-disjoint except possibly in their endpoints. Then we denote by $G[H]$ the subgraph of $G$ obtained from $H$ by replacing each edge with its corresponding path. We refer to edges of $H$ as \emph{super edges} (of $H$ w.r.t.~$G$).

%For a forest $F = (V,E)$ and a vertex set $V'$ (not necessarily contained in $V$), let $F'$ be the subforest of $F$ consisting of the union of edges of the simple paths of $F$ with both endpoints in $V'$. We denote by $F[V']$ the forest with vertex set contained in $V$ obtained from $F'$ by a maximal sequence of shortcuttings of vertices of $V(F')\cap(V - V')$.

We let $\MSF{G}$ resp.~$\MST{G}$ denote an MSF resp.~MST of $G$; in case this forest resp.~tree is not unique, we choose the MSF resp.~MST that has minimum weight w.r.t.~some lexicographical ordering of edge weights. For instance, consider assigning a unique index between $1$ and $n$ to each vertex. If two distinct edges $e_1 = (u_1,v_1)$ and $e_2 = (u_2,v_2)$ have the same weight, we regard $e_1$ as being cheaper than $e_2$ iff the index pair corresponding to $(u_1,v_1)$ is lexicograpically smaller than the index pair corresponding to $(u_2,v_2)$. We extend $\MSF{G}$ and $\MST{G}$ to the case where $G$ is a multigraph.

The fully-dynamic MSF problem is the problem of maintaining an MSF $F$ of an $n$-vertex edge-weighted dynamic simple graph $G$ under updates where each update is either the insertion or the deletion of a single edge. Initially, $G$ contains no edges.

The following is well-known and easy to show for the dynamic MSF problem. When an edge $e = (u,v)$ is inserted into $G$, $e$ becomes a new tree edge (of $F$) if it connects two distinct trees in $F$. If $e$ has both endpoints in the same tree, it becomes a tree edge if the heaviest edge $f$ on the $u$-to-$v$ path in $F$ has weight greater than $e$, and $f$ becomes a non-tree edge; otherwise $e$ becomes a non-tree edge. No other changes happen to $F$. After such an insertion, a data structure for the problem should report whether $e$ becomes a tree edge and if so, it should report $f$ if it exists.

When an edge $e = (u,v)$ is deleted, if $(u,v)$ is a non-tree edge, no updates occur in $F$. Otherwise, $F$ is correctly updated by adding a cheapest reconnecting edge (if any) for the two new trees of $F$ containing $u$ and $v$, respectively. The data structure should report such an edge if it exists.

Decremental MSF is the same problem as fully-dynamic MSF except that we only permit edge deletions; here we have an initial graph with an initial MSF and we allow a preprocessing step (which in particular needs to compute the initial MSF). Both fully-dynamic and decremental MSF extend to multigraphs but unless otherwise stated, we consider these problems for simple graphs. When convenient, we identify a fully-dynamic or a decremental MSF structure with the dynamic graph that it maintains an MSF of.

Our data structure uses the top tree structure of Alstrup et al.~\cite{TopTree}. We assume that the reader is familiar with this structure, including concepts like top tree clusters and top tree operations like \texttt{create}, \texttt{join}, \texttt{split}, \texttt{link}, and \texttt{cut}.

We shall assume the Word-RAM model of computation with standard operations where each word consists of $\Theta(\log n)$ bits plus extra bits (if needed) to store the weight of an edge. We use this model to get a cleaner description of our data structure; with only a logarithmic overhead, our time bound also applies for a pointer machine having the same word size and the same operations as in the Word-RAM model.

We use the notation $O_{f(n)}(\cdot)$, $\Omega_{f(n)}(\cdot)$, and $\Theta_{f(n)}(\cdot)$ when suppressing a factor of $f(n)^{\Theta(1)}$ or $f(n)^{-\Theta(1)}$ so that, e.g., a function $h(n)$ is $\Theta_{f(n)}(g(n))$ if $h(n) = O(g(n)f(n)^{c_1})$ and $h(n) = \Omega(g(n)f(n)^{-c_2})$ for some constants $c_1,c_2\ge 0$.

\section{Restricted Decremental MSF Structure}\label{sec:RestrictedMSF}
In this section, we present our data structure for a restricted version of decremental MSF where for an $n$-vertex graph, the total number of edge deletions allowed is upper bounded by a parameter $\Delta = \Delta(n)$. The following theorem, whose proof can be found in Section~\ref{sec:Reduction}, will imply that this suffices to obtain our fully-dynamic MSF structure.
\begin{theorem}\label{Thm:Reduction}
Let a decremental MSF structure be given which for an $n$-vertex graph of max degree at most $3$ and for constants $c_P\ge 1$ and $0 < c_U,c_{\Delta} < 1$ has preprocessing time at most $n^{c_P}$ and supports up to $n^{c_{\Delta}}$ edge deletions each in worst-case time at most $n^{c_U}$. Then there is a fully-dynamic MSF structure which for an $n$-vertex dynamic graph has worst-case update time $O((n^{c_U} + n^{c_P - c_{\Delta}})\log n)$. If for the decremental structure the preprocessing time and update time bounds hold w.h.p.~then in each update, w.h.p.~the fully-dynamic structure spends no more than $O((n^{c_U} + n^{c_P - c_{\Delta}})\log n)$ worst-case time.
\end{theorem}

We shall specify $\Delta$ later but it will be chosen slightly bigger than $\sqrt n$. Parts of the structure are regarded as black boxes here and will be presented in detail in later sections. We assume that the input graph $G = (V,E,w:E\rightarrow\mathbb R)$ has max degree at most $3$ and we will give a data structure with update time polynomially less than $\Theta(\sqrt n)$. In the following, we let $M$ denote the decremental MSF $\MSF{G}$ of $G$ that our data structure should maintain.

%We first transform $G$ to have constant degree by replacing each vertex $v$ with a simple path $v_1v_2\ldots v_{\deg(v)}$. For arbitrary orderings of edges around vertices in the original graph, each original edge $e = (u,v)$ is replaced by a new edge $(u_i,v_j)$ of the same weight where in the original graph, $e$ is the $i$th edge incident to $u$ and the $j$th edge incident to $v$. The new edges belonging to the simple paths are given a weight that is smaller than any edge weight in the original graph. This way, an MSF in the new graph $G$ defines an MSF in the original graph by contracting each simple path to its original vertex. Hence, by identifying each edge $(u_i,v_j)$ with the original edge $(u,v)$, it follows that if we can maintain an MSF in the tranformed graph then we can also maintain an MSF in the original graph with the same worst-case update time. In the following, $G$ will denote the transformed graph, both initially and during edge deletions.

A key invariant of our data structure is that it maintains a subgraph of $G$ having the same MSF as $G$ but having polynomially less than $n$ non-tree edges at all times. This allows us to apply the data structure of the following theorem whose proof is delayed until Section~\ref{sec:FewNonTreeEdges}.
\begin{theorem}\label{Thm:MSFFewNonTreeEdges}
Let $H = (V,E_H)$ be a dynamic $n$-vertex graph undergoing insertions and deletions of weighted edges where the initial edge set $E_H$ need not be empty and where the number of non-tree edges never exceeds the value $h$. Then there is a data structure which after $O(n\log n + h\log^2n)$ worst-case preprocessing time can maintain $F = \MSF{H}$ in $O(\sqrt h\log^{3/2}n)$ worst-case time per update where an update is either the insertion or the deletion of an edge in $H$ or a batched insertion of up to $\Theta(\sqrt{h/\log n})$ edges in $H$, assuming this batched insertion does not change $F$.
\end{theorem}
The data structure in Theorem~\ref{Thm:MSFFewNonTreeEdges} is at a high level similar to those of Frederickson~\cite{Frederickson} and Eppstein et al.~\cite{Sparsification} and for this reason, we shall refer to each instance of it as an \emph{FFE structure} (Fast Frederickson/Eppstein et al.) and denote it by $\FFE{H}$.

\subsection{Preprocessing}\label{subsec:PreprocRestrictedDecMSF}
Let $\eps < 1$ be some small positive constant which will be specified later; for now, we only require it to be chosen such that $\ell = m^{\eps}$ is an integer that divides $m$. In the first part of the preprocessing, we sort the weights of edges of the initial graph $G$ in non-decreasing order and assign a rank to each edge between $0$ and $m-1$ according to this order, i.e., the edge of rank $0$ has minimum weight and the edge of rank $m - 1$ has maximum weight. We redefine $w$ such that $w(e)$ equals the rank of each edge $e$. MSF $M$ w.r.t.~these new weights is also an MSF w.r.t.~the original weights and uniqueness of edge weights implies uniqueness of $M$. In particular, $M$ does not reveal any information about the random bits used by our data structure so we may assume that the sequence of edge deletions in $G$ is independent of these bits.

We compute the initial MSF $M$ using Prim's algorithm implemented with binary heaps.\footnote{We could have chosen the faster MSF algorithm in~\cite{StaticMSTDeterministic} but it is more complicated and will not improve the overall performance of our data structure.} It will be convenient to assume that each component of the initial graph $G$ contains at least $n^{\eps}$ vertices. This can be done w.l.o.g.~since we can apply the data structure of Eppstein et al.~for every other component, requiring a worst-case update time of $O(n^{\eps/2})$ which is polynomially less than $\sqrt n$.

Next, Frederickson's \texttt{FINDCLUSTERS} procedure~\cite{Frederickson} is applied to $M$, giving a partition of $V$ into subsets each of size between $n^{\eps}$ and $3n^{\eps}$ and each inducing a subtree of $M$; here we use the fact that $G$ and hence $M$ has degree at most $3$. Let $\mathcal C_M$ denote the collection of these subsets. For each $C\in\mathcal C_M$, we refer to $M[C]$ as an \emph{$M$-cluster}. We denote by $E(\mathcal C_M)$ the union of edges of $M$-clusters.

%Next a multigraph $G_M$ is obtained from $G$ by contracting each tree edge with both endpoints in the same $M$-cluster. Self-loops are retained in $G_M$ and each edge of $G_M$ is identified with the corresponding edge of $G$. Note that the vertex set of $G_M$ is the set of $M$-clusters.

For $i = -1,\ldots,\ell-1$, let $E_i$ be the set of edges of $E - E(\mathcal C)$ of weights in the range $[m - (i+1)m/\ell,m - im/\ell)$. Note that $E_{-1} = \emptyset$; this set is only defined to give a cleaner description of the data structure. For $i = -1,\ldots,\ell-1$, let $E_{\ge i} = \cup_{j = i}^{\ell-1} E_j$, $E_{\le i} = \cup_{j = -1}^i E_j$, $G_{\ge i} = (V, E_{\ge i}\cup E(\mathcal C_M))$, and $G_{\le i} = (V, E_{\le i}\cup E(\mathcal C_M))$.

\paragraph{Computing a laminar family of clusters:}
Next, a recursive procedure is executed which outputs a family $\mathcal F$ of subgraphs of $G$ that all respect $\mathcal C$. We refer to these as \emph{level $i$-clusters} where $i\in\{-1,\ldots,\ell-1\}$. Collectively (i.e., over all $i$), we refer to them as \emph{$\mathcal F$-clusters} in order to distinguish them from $M$-clusters. Family $\mathcal F$ will be laminar w.r.t.~subgraph containment. We need the following theorem whose proof can be found in Section~\ref{sec:PartitionExpanders}.
\begin{theorem}\label{Thm:RCPartition}
Let $H$ be a constant-degree graph with vertex set $V$ and let $\mathcal C$ be a partition of $V$ into subsets each of size $\Theta(n^{\eps})$ and each inducing a connected subgraph of $H$. Let $c > 0$ and $\xi > 0$ be given constants. There is an algorithm which, given $H$, $\mathcal C$, and any non-empty set $W\subseteq V$ of size $\Omega(n^{1 - \eps})$ respecting $\mathcal C$, outputs a partition $\mathcal X$ of $W$ respecting $\mathcal C$ such that with probability at least $1 - 1/n^c$, the following three conditions hold for suitable $\gamma = \Omega(n^{-2\eps})$ and $\lambda = n^{-\eps/2^{O(1/\xi)}}$:
\begin{enumerate}
\item $H[X]$ is a $\gamma$-expander graph for each $X\in\mathcal X$,
\item the number of edges of $H$ between distinct sets of $\mathcal X$ is at most $\lambda\sum_{X\in\mathcal X}|X|\log(|W|/|X|)$, and
\item the worst-case time for the algorithm is $\tilde O(|W|^{1 + 5\eps + \xi})$.
\end{enumerate}
\end{theorem}
We shall pick $\xi = \eps$ in Theorem~\ref{Thm:RCPartition} in the following. We may assume that $\lambda > n^{-\eps}$.

The recursive procedure takes as input an integer $i$ and a set of level $i$-clusters and outputs the level $j$-clusters contained in these level $i$-clusters for $j > i$. The first recursive call is given as input $i = -1$ and $G$ as the single level $-1$-cluster.

In the general recursive step, for each level $i$-cluster $C$, the algorithm of Theorem~\ref{Thm:RCPartition} is applied with $H = G_{\ge i+1}$, $\mathcal C = \mathcal C_M$, and $W = V(C)$, giving a partition $\mathcal X(C)$ of $V(C)$ respecting $\mathcal C_M$ such that for suitable $\gamma = \Theta(n^{-2\eps})$ and $\lambda = n^{-1/2^{O(1/\eps)}}$, the following holds w.h.p.,
\begin{enumerate}
\item $G_{\ge i+1}[X]$ is a $\gamma$-expander graph for each $X\in\mathcal X(C)$, and
\item the are at most $\lambda\sum_{X\in\mathcal X(C)}|X|\log(|V(C)|/|X|)$ edges of $E_{\ge i+1}$ between distinct sets in $\mathcal X(C)$.
\end{enumerate}

The graphs $G_{\ge i+1}[X]$ for all $X\in\mathcal X(C)$ are defined to be level $(i+1)$-clusters. If $i < \ell - 1$ the procedure recurses with $i+1$ and with these level $(i+1)$-clusters. The recursion stops when level $i$-cluster $C$  has at most $m^{1 - \eps} = \Theta(n^{1 - \eps})$ edges of $E_{\ge i}$; this ensures that the lower bound on $|W|$ in Theorem~\ref{Thm:RCPartition} is satisfied for each application of this theorem.

The laminar family $\mathcal F$ of all the clusters is represented as a rooted tree in the natural way where the root is the single level $-1$-cluster $G$ and a level $i$-cluster has as children the level $(i+1)$-clusters contained in it. %Each edge $e\in E$ is assigned a \emph{level} $\ell(e)$ which is the largest $i\in\{-1,\ldots,\ell - 1\}$ such that $e$ belongs to a level $i$-cluster. Edge levels will remain unchanged over all updates.

For any subset $F$ of edges of $E$ and for any $\mathcal F$-cluster $C$, we let $F(C)$ be the subset of edges of $F$ belonging to $C$ and having endpoints in distinct children of $C$ in $\mathcal F$; note that $F(C) = \emptyset$ if $C$ is a leaf of $\mathcal F$. We let $E'$ be the union of $E_{\ge i+1}(C)$ over all $i$ and all level $i$-clusters $C$.

Next, a new graph $G' = (V,E,w')$ is formed where for each level $i$-cluster $C$ the weight $w'(e)$ of each $e\in E_{\ge i+1}(C)$ is set to $m - (i+1)m/\ell - \frac 1 2$; note that this ensures that for all $e_1\in E_{\le i}$ and all $e_2\in E_{\ge i+1}$, $w(e_1) > w'(e) > w(e_2)$. For all other edges $e$ of $E$, we define $w'(e) = w(e)$. An example is shown in Figure~\ref{fig:NestingMSF}. Forest $\MSF{G'}$ is computed and an FFE structure $\mathcal M = \FFE{E'(w)\cup\MSF{G'}}$ is initialized.
\begin{figure}%[!ht]
\centerline{\scalebox{0.80}{\input{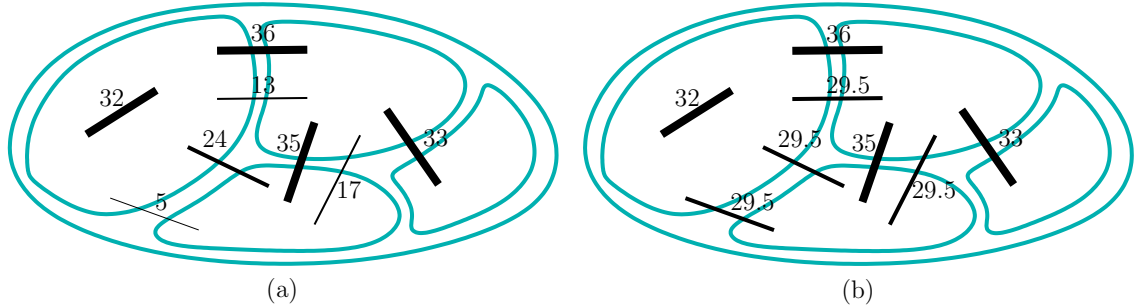}}}
\caption{(a): A level $i$-cluster $C$ is shown with four level $(i+1)$-child clusters, for $i = \ell - 4$. Letting $m/\ell = 10$, we have $[m - (i+1)m/\ell,m - im/\ell) = [30,40)$. Edges of $C$ not belonging to its children are shown together with their $w$-weights where thick edges are more expensive than thin edges. (b): The same clusters and edges but with the modified $w'$-weights.}
\label{fig:NestingMSF}
\end{figure}

\subsection{Updates}
We now describe how our data structure handles updates. First, we extend some of the above definitions from the preprocessing step to any point in the sequence of updates as follows. $M$-clusters are the components (trees) of the graph consisting of the initial $M$-clusters minus the edges removed so far. Hence, when an edge of an $M$-cluster $C$ is removed, the two new trees obtained replace $C$ as $M$-clusters. $\mathcal F$-clusters are the initial $\mathcal F$-clusters minus the edges deleted so far. Note that $\mathcal F$ remains a laminar family over all updates. Finally, $E'$, $G'$, and $E(\mathcal C_M)$ are the initial $E'$, $G'$, and $E(\mathcal C_M)$, respectively, minus the edges removed so far.

Data structure $\mathcal M$ maintains an MSF for the dynamic graph $E'(w)\cup\MSF{G'}$. Lemma~\ref{Lem:MSFDecompose} below implies that this MSF is $M$. To show it, we use the following result of Eppstein et al.~\cite{Sparsification}.
\begin{Lem}[\cite{Sparsification}, Lemma 4.1]\label{Lem:MSFSparsification}
Let $H$ be an edge-weighted multigraph and let $H_1$ and $H_2$ be two subgraphs of $H$ such that $H = H_1\cup H_2$. Then $\MSF{H} = \MSF{H_1\cup\MSF{H_2}}$.
\end{Lem}
The result was not stated for multigraphs in~\cite{Sparsification} but immediately generalizes to these.
\begin{Lem}\label{Lem:MSFDecompose}
Let $H = (V_H,E_H,w_H)$ be an edge-weighted graph, let $E_H'\subseteq E_H$, and let $H' = (V_H,E_H,w_H')$ where $w_H'(e) = w_H(e)$ for all $e\in E_H - E_H'$ and $w_H'(e) > w_H(e)$ for all $e\in E_H'$. Then $\MSF{H} = \MSF{E_H'(w_H)\cup\MSF{H'}}$.
\end{Lem}
\begin{proof}
By Lemma~\ref{Lem:MSFSparsification}, we have
\[
  \MSF{H} = \MSF{E_H'(w_H)\cup H'[E_H - E_H']}
          = \MSF{E_H'(w_H)\cup H'}
          = \MSF{E_H'(w_H)\cup\MSF{H'}}.
\]
\end{proof}
\begin{Cor}
With the above definitions, $M = \MSF{E'(w)\cup\MSF{G'}}$.
\end{Cor}
As we show later, the number of non-tree edges of $\mathcal M$ is at all times polynomially smaller than $n$. Hence, by Theorem~\ref{Thm:MSFFewNonTreeEdges}, it suffices to give an efficient data structure to maintain $\MSF{G'}$. We present this in the following. In the rest of this section, all edge weights are w.r.t.~$w'$ unless otherwise stated. An advantage
of considering $G'$ rather than $G$ is that $\MSF{G'}$ behaves nicely w.r.t.~the laminar family $\mathcal F$ as the following lemma shows.
\begin{Lem}\label{Lem:MSFLaminar}
For any $\mathcal F$-cluster $C$, $\MSF{G'}[V(C)] = \MSF{C}$.
\end{Lem}
\begin{proof}
Observe that $E(\mathcal C_M)\subseteq E(\MSF{G'})$. Hence, we can obtain $\MSF{G'}$ by running a Kruskal-type algorithm on the edges of $E - E(\mathcal C_M) = E_{\ge 0}$ where the initial forest has edge set $E(\mathcal C_M)$.

Given a level $i$-cluster $C$, we have $E(\mathcal C_M)\cap E(C)\subseteq E(\MSF{C})$. By definition of $w'$, all edges of $E(C)\cap E_{\ge i} = E(C) - E(\mathcal C_M)$ are cheaper than all other edges of $E_{\ge 0}$ incident to $C$. Hence, Kruskal's algorithm processes all edges of $E(C) - E(\mathcal C_M)$ before any other edge of $E_{\ge 0}$ incident to $C$ so it will form the spanning forest $\MSF{G'}[V(C)]$ of $C$ as part of $\MSF{G'}$. It must be a cheapest such spanning forest of $C$ since otherwise, the cost of $\MSF{G'}$ could be reduced.
\end{proof}

We now present a data structure $\mathcal M'$ that maintains $\MSF{G'}$. At a high level, this structure is similar to $\mathcal M$ as it makes use of an FFE structure. The edge set of $\mathcal M'$ is maintained using smaller dynamic structures for the various $\mathcal F$-clusters; these structures are described below.

We say that a level $i$-cluster is \emph{small} if initially it contained at most $m^{1 - \eps}$ edges of $E_{\ge i}$; otherwise, the cluster is \emph{large}. Note that a large cluster must have children in $\mathcal F$ since otherwise, it is a level $(\ell - 1)$-cluster and $|E_{\ge\ell-1}|\le m/\ell = m^{1 - \eps}$. Thus small clusters are leaves in $\mathcal F$ while large clusters are interior nodes. We shall make the simplifying assumption that each large cluster is connected over all updates. This is a strong assumption and we shall later focus on how to get rid of it.

Part of $\mathcal M'$ is a data structure $\mathcal M_{\mathit{small}}$ which maintains $M_{\mathit{small}} = \MSF{C_{\mathit{small}}}$ where $C_{\mathit{small}}$ is the union of all small $\mathcal F$-clusters. This structure consists of an FFE structure (in fact, Frederickson's original structure suffices here) for each small $\mathcal F$-cluster which is initialized during preprocessing. For large clusters, we use more involved data structures which we present in the following.

\subsubsection{Compressed clusters}
For each level $i$ and each large level $i$-cluster $C$, we define the \emph{compressed} level $i$-cluster $\overline C$ as the multigraph obtained from $C$ as follows. First, each large child cluster $C'$ of $C$ is contracted to a single vertex called a \emph{large cluster vertex}, and self-loops incident to this new vertex are removed. Second, for each small child cluster $C'$ of $C$, its edge set is replaced by $\MSF{C'}$. Figure~\ref{fig:CompressedCluster}(a) and (b) illustrate $C$ and $\overline C$, respectively. We define three subgraphs of $\overline C$:
\begin{figure}%[!ht]
\centerline{\scalebox{0.80}{\input{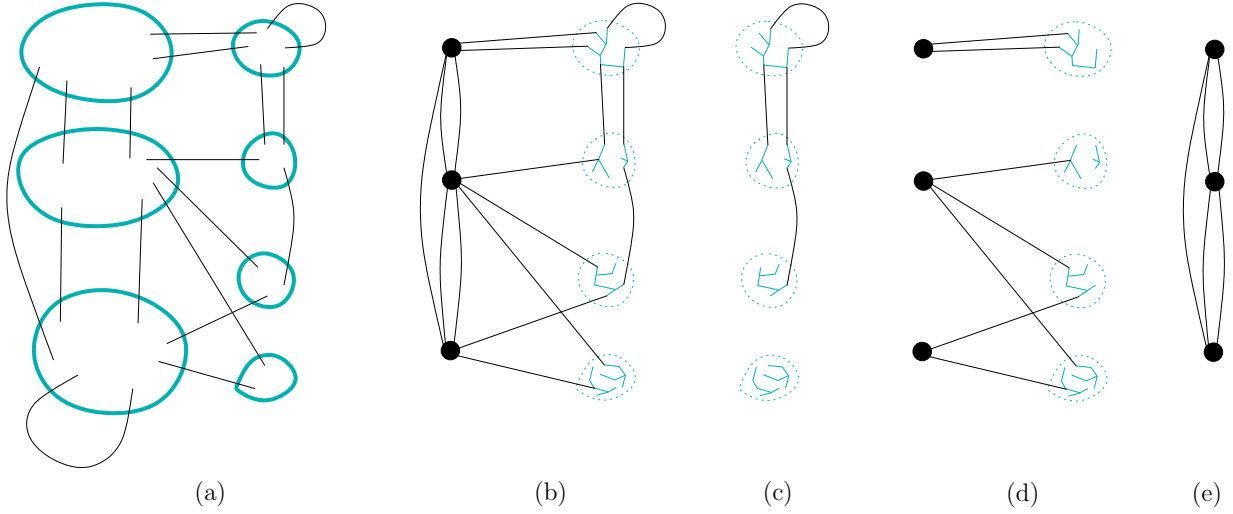}}}
\caption{(a): A level $i$-cluster $C$ with three large child clusters (left) and four small child clusters (right). Edges of $C$ not belonging to its child clusters are shown. (b): compressed cluster $\overline C$ with an MSF for each child cluster shown. Large cluster vertices are shown in black. (c)--(e): Graphs $G_1(\overline C)$, $G_2(\overline C)$, and $G_3(\overline C)$, respectively.}
\label{fig:CompressedCluster}
\end{figure}
\begin{description}
\item [$G_1(\overline C)$:] consists of the union of $\MSF{C'}$ over all small child clusters $C'$ of $C$ as well as the edges of $\overline C$ with both endpoints in small child clusters of $C$ (Figure~\ref{fig:CompressedCluster}(c)),
\item [$G_2(\overline C)$:] consists of the large cluster vertices of $\overline C$, $\MSF{C'}$ for each small child cluster $C'$ of $C$, and the edges of $\overline C$ having a large cluster vertex as one endpoint and having the other endpoint in a small child cluster of $C$ (Figure~\ref{fig:CompressedCluster}(d)),
\item [$G_3(\overline C)$:] consists of the subgraph of $\overline C$ induced by its large cluster vertices (Figure~\ref{fig:CompressedCluster}(e)).
\end{description}
Note that $G_1(\overline C)$, $G_2(\overline C)$, and $G_3(\overline C)$ together cover all vertices and edges of $\overline C$. Define $M_1(\overline C) = \MSF{G_1(\overline C)}$, $M_2(\overline C) = \MSF{G_2(\overline C)}$, and $M_3(\overline C) = \MSF{G_3(\overline C)}$. Data structure $\mathcal M'$ will use an FFE structure for the graph defined as the union of $M_{\mathit{small}}$ and of $M_1(\overline C)$, $M_2(\overline C)$, and $M_3(\overline C)$ over all compressed clusters $\overline C$. This FFE structure, which we denote by $\FFE{\mathcal M'}$, is initialized during preprocessing. By Lemma~\ref{Lem:MSFLaminar}, it will maintain $\MSF{G'}$ as desired. As we show later, $\FFE{\mathcal M'}$ contains polynomially less than $n$ non-tree edges at all times so that it can be updated efficiently.

Let $\overline C$ be a given compressed cluster. It remains to give efficient data structures that maintain $M_1(\overline C)$, $M_2(\overline C)$, and $M_3(\overline C)$. We maintain $M_1(\overline C)$ using an FFE structure for $G_1(\overline C)$, initialized during preprocessing. In the following, we present structures maintaining $M_2(\overline C)$ and $M_3(\overline C)$.

%Each update to $M_1$ takes $\tilde O(\sqrt{|E_{\ge i}(C)|}) = \tilde O(\sqrt{(\gamma + 1/\ell)n})$ time. Note that at most one such $M_1$ is updated per edge deletion since a deleted edge only affects such an $M_1$ if it was either in a leaf child of $C$ or it was an edge of $E_{\ge i}(C)$.

%Each update to $M_1$ causes at most one update in $\mathcal M'$ which can be done in $\tilde O(\sqrt{n^{1 - \eps}})$ time. When $M_2$ or $M_3$ is updated, $\mathcal M'$ may need to be updated as well. We delay the description of how to do this until we deal with disconnected large clusters.

\subsubsection{Maintaining $M_2(\overline C)$}\label{subsubsec:MaintainM2}
To maintain $M_2(\overline C)$ and $M_3(\overline C)$ efficiently, we shall exploit the fact that both $G_2(\overline C)$ and $G_3(\overline C)$ have a subset of only $O(n^{1-\eps})$ large cluster vertices and (ignoring in $G_2(\overline C)$ the edges of $\MSF{C'}$ for all small child clusters $C'$ of $C$) all edges of these graphs are incident to this small subset.

Forest $M_2(\overline C)$ is represented as a top tree. In the following, we shall abuse notation slightly and refer to this top tree as $M_2(\overline C)$. Each top tree cluster $K$ of $M_2(\overline C)$ has as auxiliary data a pair $(V_{\mathit{large}}(K),E_{\mathit{large}}(K))$ where $V_{\mathit{large}}(K)$ is the set of large cluster vertices of $\overline C$ contained in $K$ and $E_{\mathit{large}}(K)$ contains, for each large cluster vertex $v\in V(\overline C) - V_{\mathit{large}}(K)$ a minimum-weight edge $e_K(v)$ having $v$ as one endpoint and having the other endpoint in $K$; if no such edge exists, $e_K(v)$ is assigned some dummy edge $e_{\mathit{nil}}$ whose endpoints are undefined and whose weight is infinite.

In order to maintain $M_2(\overline C)$, we first describe how to maintain auxiliary data under the basic top tree operations \texttt{create}, \texttt{split}, and \texttt{join} for $M_2(\overline C)$. When \texttt{create} outputs a new cluster $K$ consisting of a single edge, we form $V_{\mathit{large}}(K)$ as the set of at most one large cluster vertex among the endpoints of the edge. Then $E_{\mathit{large}}(K)$ is computed by letting $e_K(v)$ be a cheapest edge incident to both $v$ and $K$ (or $e_{\mathit{nil}}$ if undefined), for each large cluster vertex $v\in V(\overline C) - V_{\mathit{large}}(K)$.

When a \texttt{split}$(K)$ operation is executed for a top tree cluster $K$, we simply remove $V_{\mathit{large}}(K)$ and $E_{\mathit{large}}(K)$. Finally, when two top tree clusters $K_1$ and $K_2$ are joined into a new top tree cluster $K$ by \texttt{join}$(K_1,K_2)$, we first form the set $V_{\mathit{large}}(K) = V_{\mathit{large}}(K_1)\cup V_{\mathit{large}}(K_2)$. Then we form $E_{\mathit{large}}(K)$ by letting $e_K(v)$ be an edge of minimum weight among $e_{K_1}(v)$ and $e_{K_2}(v)$, for each large cluster vertex $v\in V(\overline C) - V_{\mathit{large}}(K)$.

%two top tree clusters $K_1$ and $K_2$ of $M_2(\overline C)$ are merged into a new parent cluster $K$. First, we set $V_{\mathit{large}}(K) = V_{\mathit{large}}(K_1)\cup V_{\mathit{large}}(K_2)$. Next, for each large cluster vertex $v\in\overline C - V_{\mathit{large}}(K)$, we set $e_K(v)$ to be the edge of minimum weight among $e_{K_1}(v)$ and $e_{K_2}(v)$; this forms $E_{\mathit{large}}(K)$.

We are now ready to describe how to maintain $M_2(\overline C)$ when an edge $e$ is deleted from $G_2(\overline C)$.

\paragraph{Deleting a non-tree edge:}
Assume first that $e\notin M_2(\overline C)$. Then the topology of $M_2(\overline C)$ is unchanged. If $e$ is incident to a large cluster vertex then let $u_{\mathit{small}}$ be the other endpoint of $e$ ($u_{\mathit{small}}$ cannot be a large cluster vertex); in this case the auxiliary data for each top tree cluster containing $u_{\mathit{small}}$ needs to be updated. We do this bottom-up by first applying \texttt{create} to replace each leaf cluster containing $u_{\mathit{small}}$ with a new leaf cluster and applying \texttt{join} to update all non-leaf clusters containing $u_{\mathit{small}}$.

Note that the new set of top tree clusters is identical to the old set, only their auxiliary data are updated.

%We process these top tree clusters in a bottom-up fashion in $M_2(\overline C)$. Let $K$ be the current top tree cluster. If $K$ is a leaf cluster, corresponding to a single edge $e'$ of $M_2(\overline C)$, we recompute $E_{\mathit{large}}(K)$ from scratch by picking, for each large cluster vertex $v\in\overline C - V_{\mathit{large}}(K)$, a cheapest edge incident to both $v$ and $e'$ (or $e_{\mathit{nil}}$ if undefined). Otherwise, we update the auxiliary data for $K$ by merging the top tree clusters of its children as described above.

\paragraph{Deleting a tree edge:}
Now assume that $e$ belongs to a tree $T$ of $M_2(\overline C)$. Top tree $M_2(\overline C)$ is updated with the operation \texttt{cut}$(e)$. If $e$ belongs to $\MSF{C'}$ for some small child cluster $C'$ of $C$ then $e$ also belongs to $M_{\mathit{small}}$. In this case, if a reconnecting edge was found for $M_{\mathit{small}}$, it is added to $M_2(\overline C)$ as a reconnecting edge for $T$. By Lemma~\ref{Lem:MSFLaminar}, this is the cheapest reconnecting edge for $T$. Top tree $M_2(\overline C)$ is updated using a \texttt{link}-operation.

Now assume that no reconnecting edge was found in $M_{\mathit{small}}$ (which may also happen if $e$ did not belong to $\MSF{C'}$ for any small child cluster $C'$ of $C$). Let $T_1$ and $T_2$ be the two subtrees of $T - e$. After having computed top trees for $T_1$ and $T_2$, let $K_1$ resp.~$K_2$ be the root top tree cluster representing $T_1$ resp.~$T_2$. A cheapest reconnecting edge (if any) is of one of the following two types: a cheapest edge connecting a large cluster vertex in $T_2$ with a vertex of $K_1$ or a cheapest edge connecting a large cluster vertex in $T_1$ with a vertex of $K_2$. We shall only describe how to identify the first type of edge as the second type is symmetric. First, we identify from $K_1$ the set $V_{\mathit{large}}(K_1)$. Then the desired edge is identified as an edge $e_{K_1}(v)\in E_{\mathit{large}}(K_1)$ of minimum weight over all large cluster vertices $v\in\overline C - V_{\mathit{large}}(K_1)$. Having found a cheapest reconnecting edge $e'$ for $T$, if $e'\neq e_{\mathit{nil}}$, we add $e'$ to $M_2(\overline C)$ to reconnect $T$. In the top tree, this is supported by a \texttt{link}-operation.

\subsubsection{Maintaining $M_3(\overline C)$}
Maintaining $M_3(\overline C)$ is quite simple. For all distinct pairs of large cluster vertices $(u,v)$ in $\overline C$, the initial set of edges between $u$ and $v$ in $G_3(\overline C)$ are stored during preprocessing in a list $L(u,v)$ sorted in increasing order of weight. A graph $G_3'(\overline C)$ is formed, containing a cheapest edge (if any) between each such pair $(u,v)$. The initial $M_3(\overline C)$ is computed from $G_3'(\overline C)$ using Prim's algorithm with binary heaps. Whenever an edge $(u,v)$ is deleted from $G_3'(\overline C)$, it is also deleted from $L(u,v)$ and a cheapest remaining edge (if any) between $u$ and $v$ is identified from $L(u,v)$ and added to $G_3'(\overline C)$. Whenever a tree edge is deleted from $M_3(\overline C)$, a simple linear-time algorithm is used to find a cheapest replacement edge by scanning over all edges of $G_3'(\overline C)$.

\subsection{Performance}
We now analyze the performance of the data structure presented above. We start with the preprocessing step. % We shall let $T_{\mathit{Partition}}(n')$ denote an upper bound on the worst-case running time of the black-box algorithm that partitions a level $i$-cluster with $n'$ vertices of $V$ into level $(i+1)$-clusters. In the analysis below, we assume that the function $T_{\mathit{Partition}}$ is a polynomial of degree bigger than $1$. In Section ??, we implement this black-box algorithm and show that the degree of $T_{\mathit{Partition}}$ is close to $1$, i.e., the running time is close to linear.

\subsubsection{Preprocessing}
Prim's algorithm finds $M$ in $O(n\log n)$ time. Having found $M$, $\mathcal C_M$ can be found in $O(n)$ time since this is the time bound for Frederickson's \texttt{FINDCLUSTERS} procedure.

The time to compute $\mathcal F$ is dominated by the total time spent by the algorithm in Theorem~\ref{Thm:RCPartition}. For each $i$, the total vertex size of all level $i$-clusters is at most $n$ since their vertex sets are pairwise disjoint. Hence, the total size of all sets $W$ given to the algorithm is $O(n\ell) = O(n^{1 + \eps})$. By the third part of Theorem~\ref{Thm:RCPartition}, w.h.p.~the total time for computing $\mathcal F$ is $O_{n^{\eps}}(n)$.

By Theorem~\ref{Thm:MSFFewNonTreeEdges}, the FFE structures $\mathcal M$ and $\mathcal M_{\mathit{small}}$ can be initialized in $O(n\log^2n)$ worst-case time. This is also the case for the FFE structures of graphs $G_1(\overline C)$ since these graphs are compressed versions of subgraphs of $G$ that are pairwise both vertex- and edge-disjoint, implying that their total size is $O(n)$. Finally, to bound the time to initialize \FFE{\mathcal M'}, note that the graph consisting of the union of $M_{\mathit{small}}$ and MSFs $M_1(\overline C)$, $M_2(\overline C)$, and $M_3(\overline C)$ over all $\overline C$ contain a total of $O(n)$ edges and at most $n$ vertices of $G$. Furthermore, the total number of large cluster vertices is $O(n^{\eps}\ell) = O(n^{2\eps})$. Hence, the total worst-case time spent on initializing FFE structures is $O((n+n^{2\eps})\log^2n) = O_{n^{\eps}}(n)$.

We conclude that w.h.p., the total worst-case preprocessing time is $O_{n^\eps}(n)$.

%Since $T_{\mathit{Partition}}$ has degree bigger than $1$ and since the procedure that computes $\mathcal F$ has recursion depth at most $\ell$, the total time spent on partitioning $\mathcal F$-clusters into smaller clusters is $O(\ell T_{\mathit{Partition}}(n))$. Finally, computing $G'$ takes $O(n)$ time, computing $\MSF{G'}$ takes $O(n\log n)$ time, and initializing $\mathcal M$ takes $O(n)$ time. Hence, the total preprocessing time is $O(\ell\cdot T_{\mathit{Partition}}(n))$.

\subsubsection{Updates}
Now we bound the update time of our data structure. We start by bounding the time to update $\mathcal M$ after a single edge deletion in $G$. Recall that $\mathcal M = \FFE{E'(w)\cup\MSF{G'}}$. A single edge deletion in $G$ can cause at most one edge deletion in $E'$, at most one edge deletion in $\MSF{G'}$, and (in case a tree edge was deleted from $\MSF{G'}$) at most one edge insertion in $\MSF{G'}$. Hence, $E'(w)\cup\MSF{G'}$ and thus $\mathcal M$ can be updated with a constant number of edge insertions/deletions.

By Theorem~\ref{Thm:MSFFewNonTreeEdges}, in order to bound the time to update $\mathcal M$ after a single edge insertion/deletion, we need to bound the number of non-tree edges of $\mathcal M$. We do this in the following lemma.
\begin{Lem}
At any time during the sequence of updates, the number of non-tree edges of $\mathcal M$ is $\tilde O(\Delta + \lambda n)$.
\end{Lem}
\begin{proof}
Observe that edges of $E(\mathcal C_M)$ are edges of $M$ (since they belonged to $M$ initially and since we only delete edges from $G$). In particular, edges of $\MSF{G'}$ belonging to $E(\mathcal C_M)$ are tree edges of $\mathcal M$. Furthermore, if each $M$-cluster is contracted to a vertex in $\MSF{G'}$ then the number of remaining edges is at most the number of $M$-clusters minus $1$. The initial number of $M$-clusters in a tree of $M$ is $O(n^{1 - \eps})$ and the number of $M$-clusters can increase by at most $1$ per edge deletion in $G$. Since we have a bound of $\Delta$ on the total number of edge deletions in $G$, we conclude that at all times, the number of non-tree edges of $\mathcal M$ is $O(n^{1 - \eps} + \Delta + |E'|)$.

Next, we bound $|E'|$. By the second property of Theorem~\ref{Thm:RCPartition}, for $i = -1,\ldots,\ell - 1$, and for each non-leaf level $i$-cluster $C$, $|E_{\ge i+1}(C)| \le \lambda\sum_{X\in\mathcal X(C)}|X|\log(|V(C)|/|X|) = \lambda(|V(C)|\log|V(C)| - \sum_{X\in\mathcal X(C)}|X|\log|X|)$ where $\mathcal X(C)$ is the partition of $V(C)$ found by the algorithm in Theorem~\ref{Thm:RCPartition}. By a telescoping sums argument applied to laminar family $\mathcal F$, it follows that $|E'| = O(\lambda n\log n) = \tilde O(\lambda n)$. Since $n^{1-\eps}\le \lambda n$, the lemma follows.
\end{proof}

To also bound the time to update $\mathcal M'$, we similarly bound its number of non-tree edges. Observe that the compressed clusters are pairwise edge-disjoint. Since we assume that no large cluster becomes disconnected, it follows that at most one compressed cluster is affected by an edge deletion in $G$; let $\overline C$ be such a cluster. Then the number of edge insertions/deletions in each of $M_1(\overline C)$, $M_2(\overline C)$, and $M_3(\overline C)$ is $O(1)$. Similarly, the number of edge insertions/deletions in $M_{\mathit{small}}$ is $O(1)$. Hence, the number of updates required in $\mathcal M'$ is $O(1)$ so it suffices to bound the number of non-tree edges of $\mathcal M'$.
\begin{Lem}
At any time during the sequence of updates, the number of non-tree edges of $\mathcal M'$ is $O(n^{1 - \eps} + \Delta)$.
\end{Lem}
\begin{proof}
At any time, the number of $M$-clusters is $O(n^{1 - \eps} + \Delta)$ and the edges of $M$-clusters are all tree edges in $\mathcal M'$. Contracting $M$-clusters to vertices in $M_{\mathit{small}}$ gives a forest with $O(n^{1 - \eps} + \Delta)$ edges. For $i = 1,2,3$, consider the graph consisting of the union of $M_i(\overline C)$ over all compressed clusters $\overline C$ and $\MSF{C'}$ over all small clusters $C'$. This graph is a forest and contracting all $M$-clusters gives a forest with $O(n^{1 - \eps} + \Delta)$ edges. This shows the lemma.
\end{proof}

Combining the above with Theorem~\ref{Thm:MSFFewNonTreeEdges}, it follows that the total time to update $\mathcal M$ and $\mathcal M'$ is $\tilde O(\sqrt{\Delta + \lambda n})$. Furthermore, $\mathcal M_{\mathit{small}}$ can be maintained within this time as well since each small cluster has size $O(n^{1 - \eps})$ and at most one such cluster is affected by an edge deletion in $G$. We now have the following corollary.
\begin{Cor}\label{Cor:TimeFFEStructs}
After each edge deletion in $G$, the total time to update $\mathcal M$, $\mathcal M'$, and $\mathcal M_{\mathit{small}}$ is $\tilde O(\sqrt{\Delta + \lambda n})$.
\end{Cor}

\paragraph{Maintaining $M_1$-forests:}
It remains to bound the time over all $\overline C$ to update forests $M_1(\overline C)$, $M_2(\overline C)$, and $M_3(\overline C)$ after an edge deletion in $G$. We first focus on $M_1$-forests. Note that at most one forest $M_1(\overline C)$ needs to be updated after such a deletion. Since the edges of $\MSF{C'}$ over all child clusters $C'$ of $C$ are all tree edges of $M_1(\overline C)$, the number of non-tree edges in $G_1(\overline C)$ is $O(|E_{\ge i}(C)|) = O(|E_i| + |E_{\ge i+1}(C)| = O(n^{1 - \eps} + \lambda n\log n) = O(\lambda n\log n)$ so maintaining the $M_1$-forests can be done in $\tilde O(\sqrt{\lambda n})$ per edge deletion in $G$.

\paragraph{Maintaining $M_2$-forests:}
The data structure for maintaining $M_2$-forests is described in Section~\ref{subsubsec:MaintainM2}. To efficiently support the \texttt{join} of top tree clusters, the large cluster vertices of each compressed cluster $\overline C$ are arbitrarily labeled from $1$ to $k$ where $k$ is the number of large cluster vertices in $\overline C$. For each top tree cluster $K$ of $M_2(\overline C)$, the set $V_{\mathit{large}}(K)$ is represented as an array of $k$ bits where the $i$th bit is $1$ iff large cluster vertex $i$ belongs to $V_{\mathit{large}}(K)$. Note that $k = O(n^{\eps})$. The set $E_{\mathit{large}}(K)$ is represented as an array of length $k$ where the $i$th entry contains the edge $e_K(v)\in E(\overline C)$ where $v$ is the $i$th large cluster vertex.

With this representation of auxiliary data, it is easy to see that each \texttt{join} of two top tree clusters in $M_2(\overline C)$ and each \texttt{split} can be done in $O(k) = O(n^{\eps})$ time. Since $G$ has constant degree, we can support \texttt{create} within this time bound as well. No more than $O(\log n)$ of these operations are required in $M_2(\overline C)$ in each update, taking a total of $O(n^{\eps}\log n)$ time. From our description in Section~\ref{subsubsec:MaintainM2}, it is easy to see that finding a minimum-weight replacement edge can be done in linear time in the size of the auxiliary data stored in two top tree clusters, i.e., in time $O(n^{\eps})$.

We conclude that maintaining $M_2$-forests can be done in $\tilde O(n^{\eps})$ time per edge deletion in $G$.

\paragraph{Maintaining $M_3$-forests:}
As observed above, the number of large cluster vertices in a compressed cluster $\overline C$ is $O(n^{\eps})$ and hence $|G_3'(\overline C)| = O(n^{2\eps})$. Maintaining the graph $G_3'(\overline C)$ can be done in constant time per edge deletion in $G$ and the brute-force algorithm to find a cheapest replacement edge in $M_3(\overline C)$ can be done in $O(|G_3'(\overline C)|) = O(n^{2\eps})$ time.

We can now summarize the results of this subsection.
\begin{Lem}\label{Lem:MainPerformance}
Assume that no large $\mathcal F$-cluster becomes disconnected during a sequence of at most $\Delta$ edge deletions to $G$. Then w.h.p., the data structure of this section has worst-case preprocessing time $O_{n^{\eps}}(n)$ and worst-case update time $\tilde O(\sqrt{\Delta + n^{1-1/2^{O(1/\eps)}}} + n^{2\eps})$.
\end{Lem}

\subsection{Handling disconnected large clusters}
We now remove the simplifying assumption that large clusters do not become disconnected. To handle the general case, the following theorem is crucial; its proof can be found in Section~\ref{sec:DecDSLowCondCuts}.
\begin{theorem}\label{Thm:DisconnectExpanderGraph}
Let $c > 0$ be a constant and let $H$ be a dynamic graph of max degree at most $3$ and with $n\ge |V(H)| = \Omega(n^{1-\eps})$ which undergoes a sequence of at most $\Delta = \Omega(\sqrt{|V(H)|}))$ updates each of which is an edge deletion. Assume that w.h.p., $H$ is initially a $\gamma$-expander graph where $\gamma = O_{n^{\eps}}(1)$. Then there is a dynamic data structure for $H$ which w.h.p.~has worst-case preprocessing time $\tilde O(n) + O_{n^{\eps}}(\Delta^2/\sqrt n)$. If the sequence of updates is independent of the random bits used by the data structure then in the $k$th update, the data structure outputs a subset $V_k$ of $V(H)$ such that
\begin{enumerate}
\item $H[V(H) - W_k]$ is connected just after the update for some subset $W_k$ of $\cup_{i = 1}^k V_i$, and
\item w.h.p., $V_k$ has size $O(n^{1/2 - 4\eps})$ and is output in $O(n^{1/2 - 4\eps}) + O_{n^{\eps}}(\Delta/n^{1/4} + \Delta^4/n^2)$ worst-case time.
\end{enumerate}
\end{theorem}

Given this theorem, the modification to the data structure described in the previous subsections is quite simple. The preprocessing step is extended by setting up an instance $\mathcal D(C)$ of the data structure of Theorem~\ref{Thm:DisconnectExpanderGraph} for each large cluster $C$.

Now, consider an update where an edge $e$ is to be deleted from $G$. It will prove useful to split the update into two phases where $e$ is not deleted until the second phase. In the first phase, the following is done for each large cluster $C$ and the at most one large child cluster $C'$ of $C$ containing $e$. First, $\mathcal D(C')$ is updated with the deletion of $e$. Letting $V'\subseteq V(C')$ be the set output by $\mathcal D(C')$, all edges of $C$ incident to $V'$ are inserted into $\FFE{\mathcal M'}$, excluding those edges already present in this structure. Then all edges incident to $V'$ are removed from $G_2(\overline C)$ and $G_3(\overline C)$ and forests $M_2(\overline C)$ and $M_3(\overline C)$ are updated accordingly; new edges added to these forests are inserted into $\FFE{\mathcal M'}$ but edges removed from the forests are not removed from this structure. In the second phase, the same is done as in the previous subsections.

\subsubsection{Correctness}
We show that with the above modifications, our data structure still maintains MSF $M$ of $G$. It suffices to show that $\mathcal M'$ correctly maintains $\MSF{G'}$ and by Lemma~\ref{Lem:MSFLaminar}, this follows if we can show that for any $\mathcal F$-cluster $C$, $\FFE{\mathcal M'}$ contains $\MSF{C}$ after each update. We show the latter by induction on the height of the subtree of $\mathcal F$ rooted at $C$. The base case where the height is $0$ is straightforward since then $C$ is a small cluster and $\FFE{\mathcal M'}$ contains $M_{\mathit{small}}$ at all times and thus also $\MSF{C}$.

Now assume that the height is positive and that the claim holds for smaller heights and consider an update where an edge $e$ is to be deleted from $G$. Assuming that the claim holds at the beginning of the update, we will show that it also holds at the end of the update.

Consider the end of the first phase. Observe that the claim must hold at this point since it did so at the beginning of the update and in the first phase we only add edges to $\FFE{\mathcal M'}$. For the analysis, we construct a subgraph $D$ of $C$ by initializing $D = C$ and then doing the following for each large child cluster $C'$ of $C$. Let $V_{C'}$ be the union of subsets output by $\mathcal D(C')$ so far and let $W_{C'}$ be a subset of $V_{C'}$ such that $C'[V(C') - W_{C'}]$ is connected; such a subset must exist by Theorem~\ref{Thm:DisconnectExpanderGraph}. We remove from $D$ all edges of $E(C) - E(C')$ incident to $V_{C'}$ as well as remove all vertices of $W_{C'}$; see Figure~\ref{fig:DGraph}. Define $C'[V(C') - W_{C'}]$ to be a large child cluster of $D$.
\begin{figure}%[!ht]
\centerline{\scalebox{0.80}{\input{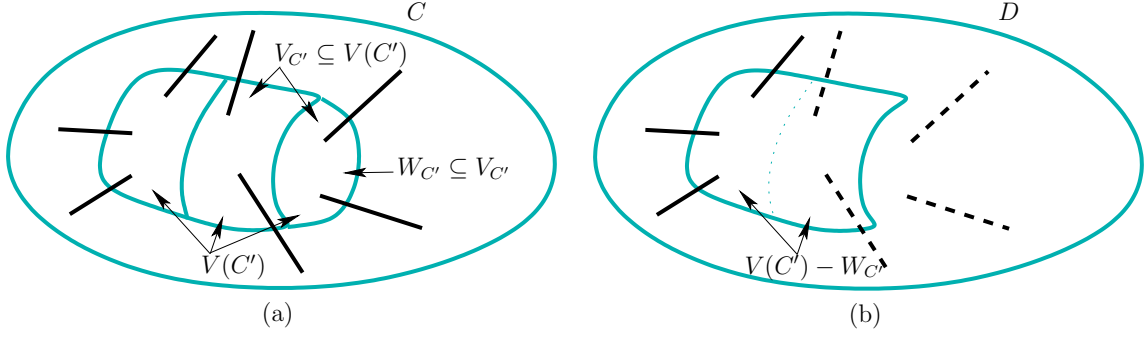}}}
\caption{(a): A cluster $C$ and a large child cluster $C'$ with subsets $W_{C'}\subseteq V_{C'}\subseteq V(C')$. In $\overline C$, $C'$ is contracted to a large cluster vertex. (b): Subgraph $D$ obtained from $C$ by removing, for each large child cluster $C'$, the vertex set $W_{C'}$ as well as all edges (dashed) incident to $V_{C'}$. Subgraph $C'[V(C') - W_{C'}]$ is connected. In $\overline D$, this subgraph is contracted to a large cluster vertex and is identified with the large cluster vertex in $\overline C$ obtained by contracting $C'$. Among the edges shown, the dashed ones are exactly those belonging to $F - e$.}
\label{fig:DGraph}
\end{figure}

We have defined the large child clusters of $D$ and we define the small child clusters of $D$ to be the small child clusters of $C$. With these definitions, let $\overline D$ be obtained from $D$ exactly in the same manner as $\overline C$ is obtained from $C$. For each large child cluster $C'$ of $C$, there is a unique large child cluster $D'$ of $D$ such that $D'\subseteq C'$. For all such pairs $(C',D')$, we identify the large cluster vertex in $\overline C$ corresponding to $C'$ with the large cluster vertex in $\overline D$ corresponding to $D'$. At the end of the first phase, we then have $\overline C = \overline D$.

Now, consider the end of the second phase of the update. At this point, $e$ has been deleted from $C$. Since the large cluster vertices of $\overline C$ are identified with large cluster vertices of $\overline D$, they correspond to subgraphs of $C$ which by Theorem~\ref{Thm:DisconnectExpanderGraph} are connected. Let $F'$ be the union of these subgraphs and let $C'$ be the union of all child clusters of $C$; note that $F' = C'[V(D)]$. By the induction hypothesis, $\FFE{\mathcal M'}$ contains $\MSF{C'}$. Let $F = E(C) - (E(C')\cup E(D))$; see Figure~\ref{fig:DGraph}(b). In the first phase, the edges of $F$ were all inserted into $\FFE{\mathcal M'}$ and they must still be present in this structure since they were all removed from $\overline C$ in the first phase. Hence, $\FFE{\mathcal M'}$ contains $F\cup\MSF{C'}\cup M_1(\overline C)\cup M_2(\overline C)\cup M_3(\overline C)$ where we view $M_i(\overline C) = M_i(\overline D)$ as a subset of edges of $G'$, for $i = 1,2,3$. With a proof similar to that of Lemma~\ref{Lem:MSFLaminar}, we have $\MSF{F'\cup G_i(\overline D)} = \MSF{F'}\cup M_i(\overline D)$. Hence, by Lemma~\ref{Lem:MSFSparsification}, $\FFE{\mathcal M'}$ also contains
\begin{align*}
  \MSF{F\cup\MSF{C'}\cup M_1(\overline C)\cup M_2(\overline C)\cup M_3(\overline C)}
  & = \MSF{F\cup C'\cup\MSF{F'}\cup M_1(\overline D)\cup M_2(\overline D)\cup M_3(\overline D)}\\
  & = \MSF{F\cup C'\cup F'\cup G_1(\overline D)\cup G_2(\overline D)\cup G_3(\overline D)}\\
  & = \MSF{F\cup C'\cup F'\cup\overline D}\\
  & = \MSF{F\cup C'\cup (D - e)}\\
  & = \MSF{C},
\end{align*}
which shows the induction step. Hence, with the above modification, $\mathcal M$ correctly maintains $M$.

%Since edges are not removed from 
%As edges are deleted from $G_2(\overline C)$ and $G_3(\overline C)$ and let $F''\subseteq F'$ be the , the FFE structure contains 

%At any point during the processing of $C$ described in this subsection, the FFE structure maintained by $\mathcal M'$ contains  
%showing the induction step. Hence, with the above modification, $\mathcal M$ correctly maintains $M$.

\subsubsection{Performance}
We now analyze the additional preprocessing and update time required with the above modifications. The total number of large clusters is $O(n^{2\eps})$ so by Theorem~\ref{Thm:DisconnectExpanderGraph}, w.h.p.~the additional preprocessing time is $O_{n^{\eps}}(n + \Delta^2/\sqrt n)$.

Now, consider the deletion of an edge $e$ from $G$. The number of large clusters containing $e$ is $O(n^{\eps})$ and by keeping pointers from each edge to the large clusters containing it, these clusters can be identified in $O(n^{\eps})$ time.

Consider one such large cluster $C$. By Theorem~\ref{Thm:DisconnectExpanderGraph}, updating $\mathcal D(C)$ takes $O(n^{1/2 - 4\eps}) + O_{n^{\eps}}(\Delta/n^{1/4} + \Delta^4/n^2)$ time with high probability. Let $V'$ be the set output by this update. During the first phase, no changes are made to $\MSF{G'}$ so all edges inserted into $\FFE{\mathcal M'}$ when $C$ is processed must belong to $E - E(\MSF{G'})$. By Theorem~\ref{Thm:MSFFewNonTreeEdges}, they can thus be inserted with batched insertions into the FFE structure, taking a total worst-case time of $\tilde O(n^{1/2 - 4\eps})$. By our earlier analysis of the data structures maintaining $M_2$- and $M_3$-forests, it follows that removing edges incident to $V'$ from $G_2(\overline C)$ and $G_3(\overline C)$ and updating $M_2(\overline C)$ and $M_3(\overline C)$ accordingly takes $O(|V'|n^{2\eps}) = O(n^{1/2 - 2\eps})$ worst-case time. Summing over all $C$, it follows that w.h.p., each update can be supported in $\tilde O(n^{1/2 - \eps}) + O_{n^{\eps}}(\Delta/n^{1/4} + \Delta^4/n^2)$ worst-case time.

Combining the above with Lemma~\ref{Lem:MainPerformance}, we are now ready to choose $\Delta$ in order to obtain Theorem~\ref{Thm:Main}. We have shown that w.h.p., preprocessing time for the structure of this section is $O_{n^\eps}(n + \Delta^2/\sqrt n)$ and update time is $n^{1/2 - 1/2^{O(1/\eps})} + \tilde O(\sqrt\Delta) + O_{n^{\eps}}(\Delta/n^{1/4} + \Delta^4/n^2)$; this is under the assumption that $\Delta = \Omega(\sqrt n)$ since we applied Theorem~\ref{Thm:DisconnectExpanderGraph} above. By Theorem~\ref{Thm:Reduction}, this gives a fully-dynamic MSF structure which for any update requires $n^{1/2 - 1/2^{O(1/\eps})} + \tilde O(\sqrt\Delta) + O_{n^{\eps}}(n/\Delta + \Delta/n^{1/4} + \Delta^4/n^2)$ time with high probability. Picking constant $\eps$ sufficiently small and picking suitable $\Delta = n^{1/2 + \Theta(\eps)}$ gives an update time of $n^{1/2 - 1/2^{O(1/\eps)}}$. This shows Theorem~\ref{Thm:Main} except for the expected time bound. The latter can easily be obtained as follows. If in an update the $n^{1/2 - 1/2^{O(1/\eps)}}$ time bound is exceeded, the data structure can update the MSF deterministically in $O(n)$ time (scanning over all edges) and then rebuild a new data structure for the next update. Since the $O(n)$ time is only spent with low probability, we get an expected time bound of $n^{1/2 - 1/2^{O(1/\eps)}}$.

%Note that the updates  none of the edges of $C$ incident to $V'$ that are inserted into the FFE structure maintained by $\mathcal M'$ belong to $\MSF{C}$ which by Lemma~\ref{Lem:MSFLaminar} is equal to $\MSF{G'}\cap C$. Hence, all these edges belong to $E - E(\MSF{G'})$ and by Theorem~\ref{Thm:MSFFewNonTreeEdges}, they can be inserted into the FFE structure using $O(1)$ batched insertions, taking a total worst-case time of ??. By our earlier analysis of the data structures maintaining $M_2$- and $M_3$-forests, it follows that removing edges incident to $V'$ from $G_2(\overline C)$ and $G_3(\overline C)$ and updating $M_2(\overline C)$ and $M_3(\overline C)$ accordingly takes $O(|V'|n^{2\eps})$ worst-case time. The new edges appearing in $M_2(\overline C)$ and $M_3(\overline C)$ must also belong to $E - E(\MSF{G'})$ so they can be batch-inserted into the FFE structure in ?? worst-case time. Summing over all $C$, it follows that w.h.p., each update can be supported in ?? worst-case time.

\section{Reduction to Decremental MSF}\label{sec:Reduction}
In this section, we give a reduction from fully-dynamic MSF to a restricted form of decremental MSF, showing Theorem~\ref{Thm:Reduction}.

\subsection{The reduction of Holm et al.}
Holm et al.~\cite{HLT01} gave a reduction from fully-dynamic to decremental minimum spanning forest. Unfortunately, this reduction will not suffice for our problem since it is not worst-case time-preserving, implying that with a decremental structure having small \emph{worst-case} update time, the reduction only yields a fully-dynamic structure with small \emph{amortized} update time. In the following, we sketch a variant of the reduction in~\cite{HLT01} but where we assume that for the black-box decremental structure, we have a bound on its worst-case update time. In the next subsection, we modify it to a worst-case time-preserving reduction.

It will be convenient to reduce from fully-dynamic MSF in a simple graph to decremental MSF in a multigraph. Assume that for an $n$-vertex multigraph with initially $m$ edges, we have a black-box decremental MSF structure with preprocessing time $P(m,n)$ and worst-case update time at most $U(m,n)$. To simplify our bounds, we shall assume that $P$ and $U$ are non-decreasing in $m$ and $n$. Since we may assume that $P$ and $U$ are bounded by polynomial functions in $m$ and $n$, we may assume that a constant-factor increase in $m$ or $n$ increases $P$ and $U$ by no more than a constant factor. We will obtain a fully-dynamic MSF structure with $O(\log^3n + U(m,n)\log n + \sum_{i = 0}^L \frac 1{2^i}P(\min\{m,2^i\log n\},\min\{n,2^i\log n\}))$ amortized update time over any sequence of updates where $m$ is the maximum number of edges present in $G$ over all updates and $L = \lceil\lg m\rceil$.

Let $G = (V,E)$ be the dynamic (simple) graph, let $F$ be the MSF of $G$, and let $\mathcal F$ be the fully-dynamic MSF structure that maintains $F$. This structure consists of pairs of decremental MSF structures, $(\mathcal D_0,\mathcal D_0'),\ldots,(\mathcal D_L,\mathcal D_L')$. We let $G_i$ denote the multigraph and let $F_i$ denote the MSF of $G_i$ maintained by $\mathcal D_i$. Furthermore, we let $N_i = E(G_i) - E(F_i)$. Similarly, we define $G_i'$, $F_i'$, and $N_i'$ for $\mathcal D_i'$.  Initially, all multigraphs $G_i$ and $G_i'$ are empty. In the general step, we require that every non-tree edge of $G$ is a non-tree edge of one of the multigraphs $G_0,G_0',\ldots,G_L,G_L'$; with the same proof as in~\cite{HLT01}, this ensures that whenever an edge of $F$ is deleted, a cheapest reconnecting edge (if any) is one of the reconnecting edges identified by the decremental MSF structures.

MSF structure $\mathcal F$ will need an auxiliary operation that, given two sets of edges $E_1$ and $E_2$, outputs a decremental structure as follows. First, a new multigraph is formed consisting of the union of $E_1\cup E_2$ and a subgraph $F'$ of $F$ consisting of all simple paths in $F$ between vertex pairs $(u,v)$ where $u$ resp.~$v$ is an endpoint of an edge of $E_1\cup E_2$. The latter ensures that any non-tree edge of $G$ belonging to $E_1\cup E_2$ is a non-tree edge of the new multigraph formed. Then a decremental structure is initialized for this new multigraph and the structure is output.

\paragraph{Edge insertions:}
We now describe how updates are handled. At the end of each update, regardless of whether it is an insertion or deletion, a cleanup procedure is applied which we describe below. First we describe the first part of the update. We start with an insertion of an edge $e$ into $G$. If $e$ connects distinct trees in $F$, $e$ is added to $F$ and no further updates are done. Now assume that the endpoints of $e$ are connected by a path $P$ in $F$. If $e$ is lighter than the heaviest edge $f$ on $P$, $e$ replaces $f$ in $F$ and the auxiliary operation is applied with $E_1 = \{f\}$ and $E_2 = \emptyset$; let $\mathcal D$ be the structure output by this operation. If $\mathcal D_0$ is empty, we set it equal to $\mathcal D$ and otherwise we set $\mathcal D_0'$ equal to $\mathcal D$. Conversely, if $e$ is heavier than $f$, we we do the same as just described but with $E_1 = \{e\}$ and $E_2 = \emptyset$.

\paragraph{Edge deletions:}
Now, consider the deletion of an edge $e$ from $G$. First, in each (multigraph represented by the) decremental structure containing $e$, let $P$ be a maximal path containing $e$ whose interior vertices have degree $2$. Viewing $P$ as a single ``super edge'' (defined below), $P$ is removed and the decremental structure outputs at most one reconnecting edge. Let $R$ be the set of replacement edges found by all decremental structures. If $e\in F$, we delete $e$ from $F$ and reconnect $F$ with the cheapest reconnecting edge from $R$, if any. Finally, we apply the same procedure as for edge insertions but with $E_1 = R$ and $E_2 = \emptyset$.

\paragraph{The cleanup procedure:}
We next describe the cleanup procedure which is applied at the end of each update. First, we need a definition. Assign time steps $0,1,\ldots$ to the updates and for each integer $k\ge 0$, define a \emph{$k$-interval} as an interval of the form $[\ell k,(\ell+1)k)$ where $\ell$ is an integer.

Now consider an update $j$. For all $i$ in increasing order, if $j$ is divisible by $2^{i+1}$, i.e., if $j$ is the beginning of a $2^{i+1}$-interval, we do as follows. The auxiliary operation is applied to $N_i\cup N_i'$, giving a new decremental structure $\mathcal D_{i+1}''$. Then $\mathcal D_i$ and $\mathcal D_i'$ are made empty and if $\mathcal D_{i+1}$ is empty, we update it to $\mathcal D_{i+1}''$, otherwise we update $\mathcal D_{i+1}'$ to $\mathcal D_{i+1}''$.

\paragraph{Implementation and analysis:}
To show correctness, note that when a new decremental structure is about to be added to a pair $(\mathcal D_i,\mathcal D_i')$, either $\mathcal D_i$ or $\mathcal D_i'$ must be empty since we only add such a new structure at the beginning of a $2^i$-interval and both $\mathcal D_i$ and $\mathcal D_i'$ are made empty at the start of a $2^{i+1}$-interval. Thus we maintain the invariant that every non-tree edge of $G$ is a non-tree edge of some decremental structure. Correctness now follows using the same analysis as in~\cite{HLT01}.

We now sketch the implementation details. We maintain $F$ as a top tree, allowing us to insert and delete edges in $F$ and to find the lightest/heaviest edge on a path; each top tree operation takes $O(\log n)$ time.

For performance reasons, we would like $|E(G_i)| = O(|N_i|)$ (and similarly $|E(G_i')| = O(|N_i'|)$). This is done by modifying the auxiliary operation above so that instead of explicitly including the subgraph $F'$ of $F$ in the new multigraph $G_i$, it instead adds \emph{super edges} each of which corresponds to a maximal path in $F'$ where interior nodes have degree $2$. As shown in~\cite{HLT01}, this compact representation of $G_i$ has size $O(|N_i|)$ and can be identified in $O(|N_i|\log n)$ time with a suitable top tree $F_i$ (and $F_i'$ for $G_i'$) of $F$. In total, we maintain $2(L+1)$ such top trees, one for each decremental structure.

In each update, we add at most $L$ edges to either $\mathcal D_0$ or $\mathcal D_0'$. Hence, at the beginning of every $2^{i+1}$-interval, $|N_i\cup N_i'| \le 2^{i+1}L$ so in the cleanup phase, it takes $O(2^i\log^2n)$ time to form the multigraph of $\mathcal D_{i+1}''$ plus $O(P(\min\{m,2^i\log n\},\min\{n,2^i\log n\}))$ time to initialize $\mathcal D_{i+1}''$. When applying a top tree to form a subgraph $F'$ of $F$ as described above, the information in this top tree changes. After having formed $F'$, we undo these changes so that the top tree is ready to form the next such tree. Undoing the changes can be done within the time to form $F'$.

Since the work just described is only done every $2^{i+1}$ updates it follows by summing over all $i$ that the amortized cost per update for the cleanup phase is
\[
  O(\log^3n + \sum_{i = 0}^L\frac 1{2^i}P(\min\{m,2^i\log n\},\min\{n,2^i\log n\})).
\]

In the first part of an update, we delete at most one edge from each decremental structure. Since the bound on $|N_i\cup N_i'|$ implies $|E(G_i)\cup E(G_i')| = O(2^iL)$, the first part of an update takes $O(\sum_{i = 0}^L U(\min\{m,2^iL\},\min\{n,2^iL\})) = O(U(m,n)\log n)$ time. Updating all top trees in an update to the new forest $F$ takes a total of $O(\log^2n)$ time.

We conclude that each update takes amortized time
\[
  O(\log^3n + U(m,n)\log n + \sum_{i = 0}^L\frac 1{2^i}P(\min\{m,2^i\log n\},\min\{n,2^i\log n\})).
\]

\subsection{A worst-case time-preserving reduction}
We now modify $\mathcal F$ to get a similar worst-case update time bound for this structure. A standard deamortization trick is used of spreading the work of constructing a new decremental structure over multiple updates rather than doing all the work in a single update.

For this to work, we introduce, in addition to each pair $(\mathcal D_i,\mathcal D_i')$, an additional pair $(\mathcal B_i,\mathcal B_i')$ of decremental structures. We can think of $\mathcal B_i$ and $\mathcal B_i'$ as snapshots of $\mathcal D_i$ and $\mathcal D_i'$ at the beginning of each $2^{i+1}$-interval and we use these snapshots to build $\mathcal D_{i+1}''$ in the background during this interval.

More precisely, we modify $\mathcal F$ so that in the beginning of a $2^{i+1}$-interval $I$, we move $\mathcal D_i$ to  $\mathcal B_i$ and $\mathcal D_i'$ to $\mathcal B_i'$ and identify $\mathcal D_i$ and $\mathcal D_i'$ with empty decremental structures. We now start forming $\mathcal D_{i+1}''$ as described above but with $\mathcal B_i$ and $\mathcal B_i'$ rather than $\mathcal D_i$ and $\mathcal D_i'$. The work for forming $\mathcal D_{i+1}''$ is spread evenly over each update of the first half of $I$. In the second half, we delete edges from $\mathcal D_{i+1}''$ at ``double speed''. More precisely, in the $k$th update of the second half of $I$, we delete from $\mathcal D_{i+1}''$ the edges deleted from $G$ in the $2k$th and $(2k+1)$th update of $I$. Hence, at the end of the last update of $I$, $\mathcal D_{i+1}''$ is up-to-date with the current graph $G$. At this point, if $\mathcal D_{i+1}$ is empty, we update it to $\mathcal D_{i+1}''$ and otherwise we update $\mathcal D_{i+1}'$ to $\mathcal D_{i+1}''$.

All other parts of $\mathcal F$ are updated exactly as in the previous version.

\paragraph{Implementation and analysis:}
In the previous version of $\mathcal F$, $\mathcal D_{i+1}''$ was constructed and included in $(\mathcal D_{i+1},\mathcal D_{i+1}')$ at the beginning of a $2^{i+1}$-interval $I$. At this point, both $\mathcal D_i$ and $\mathcal D_i'$ could be made empty since every non-tree edge of $G$ which is a non-tree edge of $\mathcal D_i$ or $\mathcal D_i'$ is a non-tree edge of either $\mathcal D_{i+1}$ or $\mathcal D_{i+1}'$. In the new version of $\mathcal F$, $\mathcal D_{i+1}''$ is instead added to $(\mathcal D_{i+1},\mathcal D_{i+1}')$ at the end of $I$. Hence, during the updates of $I$, each non-tree edge of $G$ which was a non-tree edge of $\mathcal D_{i+1}''$ in the old version is a non-tree edge of either $\mathcal D_i$ or $\mathcal D_i'$ in the new version. Correctness now follows since the new version also maintains the invariant that every non-tree edge of $G$ is a non-tree edge of some structure $\mathcal D_i$ or $\mathcal D_i'$.

We change the implementation such that the top trees $F_i$ and $F_i'$ are not updated during the first half of a $2^{i+1}$-interval $I$. Hence, in each update in the first half of $I$, both $F_i$ and $F_i'$ are top trees representing $F$ at the start of $I$, allowing $\mathcal D_{i+1}''$ to be formed correctly. During the last half of $I$, $F_i$ and $F_i'$ are updated at double speed in the same way that $\mathcal D_{i+1}''$ is updated in this half of $I$.

We now bound the worst-case update time of $\mathcal F$. We first focus on the time to construct and update a decremental structure $\mathcal D_{i+1}''$ during a $2^{i+1}$-interval $I$. In the first half of $I$, $O(\frac 1{2^i}(2^i\log^2n + P(\min\{m,2^i\log n\},\min\{n,2^i\log n\})))$ time per update is spent on forming $\mathcal D_{i+1}''$. In the last half of $I$, $O(\log n + U(\min\{m,2^i\log n\},\min\{n,2^i\log n\}))$ time is spent on updating $\mathcal D_{i+1}''$ and top trees $F_i$ and $F_i'$ at double speed. Summing over all $i$ gives a worst-case time bound of $O(\log^3n + U(m,n)\log n + \sum_{i = 0}^L\frac 1{2^i}P(\min\{m,2^i\log n\},\min\{n,2^i\log n\}))$.

In every update, we delete at most one edge from each $\mathcal D_i$- and $\mathcal D_i'$-structure. This is $O(U(m,n)\log n)$ worst-case time. We can now conclude this subsection with the following theorem.
\begin{theorem}\label{Thm:WorstCasePreservingReduction}
Let a decremental MSF structure be given which for an $n$-vertex multigraph with initially $m$ edges has $P(m,n)$ preprocessing time and $U(m,n)$ worst-case update time where both $P$ and $U$ are non-decreasing. Then there is a fully-dynamic MSF structure with worst-case update time $O(\log^3n + U(m,n)\log n + \sum_{i = 0}^{\lceil\lg m\rceil}\frac 1{2^i}P(\min\{m,2^i\log n\},\min\{n,2^i\log n\}))$ for an $n$-vertex dynamic graph in which the number of edges never exceeds the value $m$.
\end{theorem}

\subsection{Reduction to Restricted Decremental MSF}
We now reduce the decremental MSF problem further, first to the case where we have an initial $n$-vertex (simple) graph of max degree at most $3$ and then further to the case where the total number of edge deletions is bounded by some parameter $\Delta$ that may be smaller than the initial number of edges.

\paragraph{Reduction to degree at most $3$:}
Let $G$ be an edge-weighted multigraph with $m$ edges and $n$ vertices. Construct a new simple graph $G'$ of degree at most $3$ from $G$ as follows. For each vertex $u$ of $G$, let $v_1,\ldots,v_k$ be its neighbors. We replace $u$ with $k$ copies, $u_1,\ldots,u_k$, replace edges incident to $u$ with edges $(u_i,v_i)$ for $i = 1,\ldots,k$ without changing their weights, and add edges $(u_i,u_{i+1})$ for $i = 1,\ldots,k-1$; the weight of each such edge $(u_i,u_{i+1})$ is chosen to be smaller than any edge weight in $G$. We identify each edge $(u_i,v_i)$ with its corresponding edge in $G$.

It is easy to see that an MSF of $F$ can be obtained from an MSF of $G'$ by contracting all edges that are not present in $G$. It follows that if we have a decremental MSF structure for an $n'$-vertex graph of degree at most $3$ with preprocessing time at most $P(n')$ and worst-case update time at most $U(n')$ where $P$ and $U$ are non-decreasing, then there is a decremental MSF structure of $G$ with preprocessing time $O(P(m + n))$ and worst-case update time $O(U(m + n))$.

\paragraph{Reduction to at most $\Delta$ deletions:}
We shall make a further reduction from decremental MSF in an $n$-vertex graph of degree at most $3$ to the same problem but where we have a bound $\Delta = \Delta(n)$ on the total number of edge deletions permitted. Assume for simplicity that $\Delta$ is divisible by $3$ and assume access to a decremental MSF structure $\mathcal D_{\Delta}$ which for an $n$-vertex graph of degree at most $3$ has preprocessing time at most $P(n)$ and supports up to $\Delta(n)$ edge deletions each in worst case time at most $U(n)$.

Now, let $G$ be an $n$-vertex graph of degree at most $3$. We obtain a decremental MSF structure for $G$ as follows. At the start of each $\Delta/3$-interval $I_1$, a new instance $\mathcal D_{\Delta}(I_1)$ of $\mathcal D_{\Delta}$ is initialized for the current graph $G$; the work for this is spread evenly over the updates of $I_1$. In the $\Delta/3$-interval $I_2$ following $I_1$, we delete edges at double speed from $\mathcal D_{\Delta}(I_1)$ (similar to what we did in the previous subsection when setting up an instance $\mathcal D_{i+1}''$) so that at the end of $I_2$, $\mathcal D_{\Delta}(I_1)$ is up to date with the current graph $G$. At this point, $2\Delta/3$ edges have been deleted from $\mathcal D_{\Delta}(I_1)$ so it still supports an additional $\Delta/3$ deletions. At the beginning of the $\Delta/3$-interval $I_3$ following $I_2$, $\mathcal D_{\Delta}(I_1)$ the active structure, responsible for maintaining the MSF of $G$ during the updates of $I_3$.

Hence, in any $\Delta/3$-interval, one instance of $\mathcal D_{\Delta}$ is being initialized, one instance has edges deleted from it at double speed, and one instance is the active one, maintaining the MSF of $G$.

\paragraph{Proving Theorem~\ref{Thm:Reduction}:}
We are now ready to prove Theorem~\ref{Thm:Reduction}.

Using the approach above gives a decremental structure for an $n$-vertex graph of degree at most $3$ with preprocessing time $O(n^{c_P})$ and worst-case update time $O(n^{c_U} + n^{c_P - c_{\Delta}})$. The vertex-splitting argument above gives a decremental structure with preprocessing time $O((m+n)^{c_P})$ and worst-case update time $O((m+n)^{c_U} + (m+n)^{c_P - c_{\Delta}})$ for an $n$-vertex multigraph with initially $m$ edges. By Theorem~\ref{Thm:WorstCasePreservingReduction}, we get a fully-dynamic MSF structure with worst-case update time $O(((m+n)^{c_U} + (m+n)^{c_P - c_{\Delta}})\log n + (m+n)^{c_P-1}) = O(((m+n)^{c_U} + (m+n)^{c_P - c_{\Delta}})\log n)$. Applying Theorem 3.3.2 of~\cite{Sparsification} now gives the first part of the theorem.

To show the second part, assume that the preprocessing and update time bound of the restricted decremental MSF structure hold with probability at least $1 - (n')^C$ for an $n'$-vertex graph where $C$ is a large constant $C$. Let $n$ be the number of vertices in the graph for the fully-dynamic problem that we reduce from. Note that $n'$ may be $\Theta(1)$ in our reduction, meaning that the preprocessing and update time bounds may fail with constant probability. To handle this, we modify the decremental structure so that if $n' < n^{c_U/2}$, the preprocessing step consists of computing the initial MSF in $O(n'\log n') = O(n^{c_P/2})$ time with Prim's algorithm, and each update is handled in $O(n') = O(n^{c_U/2})$ time by a simple deterministic linear-time update procedure. This ensures that for any value of $n'$, the decremental structure achieves a preprocessing bound of $O((n' + \sqrt n)^{c_P})$ and a worst-case update time bound of $O((n'+\sqrt n)^{c_U})$ with probability at least $1 - n^{c_UC}$. Going through the reduction steps as above now shows the second part of the theorem.

\section{Fully-Dynamic MSF With Few Non-tree Edges}\label{sec:FewNonTreeEdges}
In this section, prove Theorem~\ref{Thm:MSFFewNonTreeEdges}. We consider an edge-weighted graph $H = (V, E_H, w: E_H\rightarrow\mathbb R)$ and present a dynamic data structure with update time polynomially faster than $\sqrt n$ when this property holds. The dynamic problem considered here differs from the standard fully-dynamic MSF problem in that we allow extra operations giving more control to the user regarding the structure of the forest maintained; in fact, the forest is not required to span $H$.

Another way the problem differs from the standard fully-dynamic version is that we do not require that we start with $E_H = \emptyset$ and we allow a preprocessing phase for this initial graph, as stated in Theorem~\ref{Thm:MSFFewNonTreeEdges}.

In addition to $H$, we assume that the preprocessing algorithm is given an initial forest $F$ in $H$ as input. This forest will be maintained during edge updates. We keep a partition of $F$ into subtrees whose vertex sets are called \emph{regions}. Regions are similar to Frederickson's clusters except that we do not balance them by size but by the number of endpoints of non-tree edges they contain. Hence, in our application, the average region size will be polynomially greater than $\sqrt n$, assuming the number of non-tree edges is polynomially smaller than $n$. For a region $R$, we denote by $N(R)$ the number of vertices of $R$ that are incident to edges in $E_H - E(F)$. Each region is given a unique label in $\{1,\ldots,n\}$. Our data structure will not fix such a label for regions $R$ where $N(R) = 0$ and $R$ is a component of $H$; this is a technicality that will simplify our implementation.

We assign some auxiliary information to each edge $e$ of $H$. A bit $f(e)$ is $1$ iff $e\in F$. Another bit $h(e)$ is $1$ iff $e\in H - F$ or $e$ is an inter-region edge of $F$. If $h(e) = 1$, we assign to $e$ a label $\mathcal R(e)$ which is a pair of labels denoting the regions containing the endpoints of $e$; if $h(e) = 0$, we leave $\mathcal R(e)$ undefined. Our data structure will only store region labels in such pairs $\mathcal R(e)$.

\subsection{The region forest}
We use a top tree structure $\hat F$ called the \emph{region forest} to maintain and dynamically merge and split regions, similarly to what is needed for clusters in Frederickson's data structure. Forest $\hat F$ has the same edge set as $F$ and supports various operations that we focus on in the following.

In $\hat F$, certain vertices are marked. More specifically, a vertex is marked iff it is incident to an edge $e$ with $h(e) = 1$. In addition to the operations below, $\hat F$ supports the $O(\log n)$ time operations of finding a nearest marked vertex to a given vertex, marking and unmarking a vertex, and inserting and deleting an edge.

When implementing \texttt{MergeRegion} and \texttt{SplitRegion} below, it will prove useful to have $\hat F$ support the auxiliary operation \texttt{FindRegionTree}$(v)$ for a $v\in V(\hat F)$. Letting $R_v$ be the region containing $v$ and letting $M(R_v)$ be the set of marked vertices in $R_v$, this operation returns the tree $T_v$ consisting of the union of all simple paths between pairs of marked vertices; see Figure~\ref{fig:FewNonTreeEdges}.
\begin{figure}%[!ht]
\centerline{\scalebox{0.80}{\input{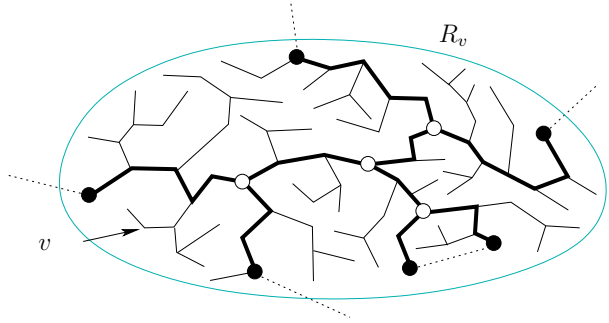}}}
\caption{A region $R_v$ containing a vertex $v$ and having a set $M(R_v)$ of six marked vertices (black) which are incident to tree edges leaving $R_v$ or to non-tree edges incident to $R_v$ (dotted segments). The tree $T_v$ consists of both black and white vertices and contains nine edges, one for each fat path between two consecutive black or white vertices.}
\label{fig:FewNonTreeEdges}
\end{figure}

Letting $b = |M(R_v)|$, we can support \texttt{FindRegionTree}$(v)$ in $O(b\log n)$ as follows. First we find a nearest marked vertex $v'$ to $v$ in $\hat F$. Among the at most three edges of $H$ incident to $v'$, delete from $\hat F$ those that belong to $F$ and leave $R_v$ by checking their $f$ and $h$ bitmaps. Then unmark $v'$ and repeat the procedure recursively on $v'$ until no more marked vertices are encountered. The edges deleted are exactly the inter-region edges of $F$ leaving $R_v$ so this makes $\hat F[R_v]$ a component of $\hat F$. The set of unmarked vertices is exactly $M(R_v)$. Extending $\hat F$ with additional operations as described in~\cite{HLT01} (see the implementation subsection for fully-dynamic minimum spanning tree) allows us to obtain $T_v$ in $O(b\log n)$ time. Finally, we clean up by marking the unmarked vertices of $\hat F$ and inserting back the deleted edges. This clean-up step also takes $O(b\log n)$ time. Note that $|T_v| = O(b)$.

\subsection{Merging and splitting regions}
The following two types of operations in $\hat F$ allow for merging and splitting regions, respectively:
\begin{description}
\item [\texttt{MergeRegion}$(v)$:] merges the region $R_v$ containing $v$ with a region incident to $R_v$ in $F$, assuming such a region exists; the two regions are thus replaced by their union,
\item [\texttt{SplitRegion}$(v,t)$:] splits the region $R_v$ containing $v$ into subregions such that for each such subregion $R$, $\frac 1 3 t\leq N(R)\leq t$; it is assumed that $3\leq t\leq N(R_v)$.
%\item [\texttt{connect}$(u,v)$]: assuming $u$ and $v$ are in distinct trees of $F_2$, adds an edge from $H$ that connects these trees, assuming such an edge exists.
\end{description}
Although these operations apply to $\hat F$, we shall require them to also correctly update the auxiliary information stored at edges of $H$.

\subsubsection{Implementing \texttt{MergeRegion}$(v)$}
We first apply \texttt{FindRegionTree}$(v)$ to $\hat F$, giving $T_v$. For each vertex of $T_v$, we check if any of its incident edges in $H$ is an inter-region edge of $F$ (again by checking the $f$ and $h$ bitmaps of these edges). If no such edge is found, $R_v$ is not incident in $F$ to any region. Otherwise, let $e = (v_1,v_2)$ be one such edge where $v_1\in R_v$. From label $\mathcal R(e)$, we obtain the label $\ell_v$ of $R_v$ and the label $\ell_v'$ of a region $R_v'$ which is incident to $R_v$ in $F$ with $v_2\in R_v'$. We will merge $R_v$ and $R_v'$ into a new region with label $\ell_v$. To do this, we first apply \texttt{FindRegionTree}$(v_2)$ to get $T_{v_2}$. We then visit all edges of $H$ incident to $V(T_{v_2})$. For each such edge $e'$, if $\mathcal R(e')$ is defined, we update each occurence of $\ell_v'$ in this set to $\ell_v$. Finally we set $h(e)$ to $0$ and delete $\mathcal R(e)$ as $e$ is no longer an inter-region edge of $F$.

The running time of \texttt{MergeRegion}$(v)$ is $O(b\log n)$ where $b$ is the total number of marked vertices in $R_v$ and $R_v'$, i.e., the number of marked vertices in the merged region.

\subsubsection{Implementing \texttt{SplitRegion}$(v,t)$}
We first apply \texttt{FindRegionTree}$(v)$ to $\hat F$ to get $T_v$. Let $b$ be the number of marked vertices in $T_v$. We then apply Frederickson's linear-time \texttt{FINDCLUSTERS} procedure~\cite{Frederickson} to partition $T_v$ into subtrees $T_1,\ldots,T_k$ whose vertex sets $R_i = V(T_i)$ form a partition of $V(T_v)$ where $\frac 1 3 t\leq |V(T_i)|\leq t$ for $i = 1,\ldots,k$. Let $\ell_1,\ldots,\ell_i\in\{1,\ldots,n\}$ be new unique labels for $R_1,\ldots,R_k$, i.e., distinct labels none of which are equal to labels of existing regions.

For $i = 1,\ldots,k$ and for each edge $e$ of $H - F$ incident to $R_i$, we update to $\ell_i$ the labels in $\mathcal R(e)$ for each endpoint of $e$ in $T_i$. For each edge $e$ of $F$ leaving $R_i$, we set $h(e)$ to $1$ since $e$ is now an inter-region edge of $F$; furthermore, we update to $\ell_i$ the label in $\mathcal R(e)$ for the endpoint of $e$ in $T_i$.

Excluding the time to find unique labels, the above takes $O(b\log n)$ time where $b$ is the number of marked vertices in $R_v$. To quickly find $\ell_1,\ldots,\ell_k$, our data structure maintains a dynamic list $\mathcal L$ consisting of those labels in $\{1,\ldots,n\}$ that are currently not used by any region. Whenever two regions are merged, the now unused label of one of the two regions is added to $\mathcal L$. When a new region is formed, we extract the first label from $\mathcal L$ and use it to label this region. Each operation on $\mathcal L$ takes constant time if we use a linked list. In particular, finding $\ell_1,\ldots,\ell_k$ can be done in $O(k) = O(b)$ time. Total time for \texttt{SplitRegion} is thus $O(b\log n)$.

\subsection{Edge insertions}
We now describe operations that maintain $\hat F$ and the auxiliary information stored at edges of $H$ under edge insertions and deletions in $H$. We start with insertions. Since we later need to have some control over which edges belong to $F$ (equivalently, to $\hat F$), we extend the insert operation with an extra argument specifying whether the new edge should be a tree edge or a non-tree edge; if it should be a tree edge, it is assumed that it connects two distinct trees in the current forest $F$. Note that this allows for $F$ to be a non-spanning forest of $H$; we later present an operation to find a minimum-weight connecting edge between two trees in $F$. This will be needed if we want to maintain $F$ as an MSF of $H$.

Consider the insertion of an edge $e = (u,v)$ into $H$. Assume first that it should be added as a non-tree edge. We look for a nearest marked vertex $u'$ to $u$ in $\hat F$. If $u'$ exists, we obtain from one of its incident edges the label $\ell_u$ of the region $R_u$ containing $u$. If $u'$ does not exist, we extract from $\mathcal L$ a new label $\ell_u$ for $R_u$. Similarly, we find the label $\ell_v$ of the region $R_v$ containing $v$. We then add $e$ to $H$, mark $u$ and $v$ in $\hat F$, set $f(e) = 0$, $h(e) = 1$, and $\mathcal R(e) = (\ell_u,\ell_v)$.

Now suppose that $e$ should be added as a tree edge. Prior to the insertion, $e$ must connect two distinct trees in $F$ so it will be an inter-region edge of $F$. Let $R_u$ and $R_v$ be the regions containing $u$ and $v$, respectively. With the same procedure as above, we find labels $\ell_u$ and $\ell_v$ for $R_u$ and $R_v$, respectively. We then insert $e$ into $H$ and $\hat F$ and set $f(e) = 1$, $h(e) = 1$, and $\mathcal R(e) = (\ell_u,\ell_v)$.

The running time for inserting $e$ is $O(\log n)$.

\subsection{Edge deletions}
Now consider the operation of deleting an edge $e = (u,v)$ from $H$. We do not require this operation to look for a reconnecting edge if $e$ belongs to $F$ since we later give an operation for this.

We start by removing $e$ from $H$. Assume first that $e$ was in $F$; we can check this in constant time by inspecting $f(e)$. Now, consider the subcase that $h(e) = 1$. Then $e$ must be an inter-region edge of $F$. In this case, we unmark $u$ resp.~$v$ in $\hat F$ unless that vertex is still incident to an edge $e'$ in $H$ with $h(e') = 1$ after the deletion of $e$. In case the region $R$ containing $u$ resp.~$v$ no longer has a marked vertex, we must have that $N(R) = 0$ and $R$ is a component of $H$; in this case, the label of that region is added to $\mathcal L$. Finally, we delete $e$ from $\hat F$.

Now assume that $e$ was in $F$ and $h(e) = 0$. Then $e$ must have both endpoints in some region $R_e$. The deletion of $e$ will split $R_e$ into two subregions. Prior to deleting $e$, we apply \texttt{FindRegionTree}$(u)$ to $\hat F$ to get $T_u$ which contains exactly the marked vertices of $R_e$. We obtain $T_u'$ resp.~$T_v'$ which is the tree in $T_u - \{e\}$ containing $u$ resp.~$v$. Using the same procedure as in the implementation of \texttt{SplitRegion} above, we split $R_e$ into two subregions containing $T_u'$ and $T_v'$, respectively, and give unique labels to each of them. Finally, we delete $e$ from $\hat F$.

The remaining case is when $e$ was a non-tree edge. This is handled in the same way as the case when $e$ was a tree edge and $h(e) = 1$.

Total time for handling an edge deletion is $O(b\log n)$ where $b$ is the number of marked vertices in $R_e$. We summarize the results above in the following lemma.
\begin{Lem}\label{Lem:MaintainHatF}
A call to \texttt{MergeRegion} or \texttt{SplitRegion} can be done in $O((r+\rho)\log n)$ time where $r$ is the maximum value $N(R)$ of any region $R$ and $\rho$ is the maximum number of regions; this includes the time for updating $\hat F$, $f$- and $h$-bitmaps, and pairs of labels $\mathcal R(e)$ for edges $e\in H - F$. Updating these for an edge insertion in $H$ takes $O(\log n)$ time and takes $O((r+\rho)\log n)$ time for an edge deletion in $H$.
\end{Lem}
\begin{proof}
The lemma follows from the above and from the observation that the number of marked vertices in a region is $O(r+\rho)$.
\end{proof}

\subsection{Finding minimum-weight connecting edges}
We now extend our data structure to support the operation \texttt{connect}$(u,v)$ which, given two vertices $u$ and $v$ in distinct trees of $F$, finds a minimum-weight edge of $H$ (if any) connecting these trees.

To support \texttt{connect}$(u,v)$, we introduce a new top tree structure $\hat F'$. Like $\hat F$, it contains the same edge set as $F$. Each top tree cluster $C$ in $\hat F'$ has as auxiliary data two lists $\mathcal L_{\mathit{edge}}(C)$ and $\mathcal L_{\mathit{region}}(C)$. List $\mathcal L_{\mathit{edge}}(C)$ consists of the pairs $(e,\ell)$ where $e$ is a minimum-weight edge in $H - E(F)$ with at least one endpoint in $C$ and at least one endpoint in the region with label $\ell$. List $\mathcal L_{\mathit{region}}(C)$ consists of the labels $\ell$ of regions sharing vertices with $C$. We implement both lists as red-black trees where elements are kept in sorted order by their $\ell$-value.

If we can maintain $\hat F'$, we can support \texttt{connect}$(u,v)$ as we show in the following. Let $C_u$ and $C_v$ be the root clusters of $\hat F'$ corresponding to the trees $T_u$ and $T_v$ in $F$ containing $u$ and $v$, respectively; identifying these top tree clusters from $u$ and $v$ can be done in $O(\log n)$ time. Observe that if there is a connecting edge for $T_u$ and $T_v$, a cheapest such edge has one endpoint in $C_u$ and the other endpoint in a region of $T_v$. It can be chosen as an edge $e$ of minimum weight over all pairs $(e,\ell)\in\mathcal L_{\mathit{edge}}(C_u)$ where $\ell\in\mathcal L_{\mathit{region}}(C_v)$. Searching in parallel through the two lists in sorted order, $e$ is identified in $O(|\mathcal L_{\mathit{edge}}(C_u)| + |\mathcal L_{\mathit{region}}(C_v)|) = O(\rho)$ time. Total time for \texttt{connect}$(u,v)$ is thus $O(\rho + \log n)$.

%This suggests the following approach to finding a connecting edge. First, we look for a label in an element of $\mathcal L(C_u)\cap\mathcal L(C_v)$; this can be done in time $O(|L(C_u)| + |L(C_v)|)$ since the two lists are sorted by label value. If no such label is found, no connecting edge exists. Otherwise, let $\ell$ be such a label. We apply the following recursive procedure to find a connecting edge. For each of the at most four pairs of top tree clusters $(C_u',C_v')$ where $C_u'$ resp.~$C_v'$ is a child top tree cluster of $C_u$ resp.~$C_v$, we recurse on one such pair for which $\ell$ belongs to an element of $\mathcal L(C_u')\cap\mathcal L(C_v')$. The recursion stops once we reach a leaf top tree cluster corresponding to a connecting edge $e$ which we output. Total time for \texttt{connect}$(u,v)$ is $O(\rho\log n)$.

\subsubsection{Maintaining $\hat F'$}
It remains to describe how $\hat F'$ is maintained. There are two types of updates to $\hat F'$, topological and non-topological changes. The topological changes happen when an edge is deleted from or inserted into $\hat F'$ which causes updates to top tree clusters. The non-topological changes happen when $H - F$ changes or labels in pairs $\mathcal R(e)$ are updated for edges $e$ in $H - F$.

\paragraph{Topological changes:}
Supporting a topological change in $\hat F'$ reduces to supporting a sequence of $O(\log n)$ top tree operations \texttt{create}$()$, \texttt{join}$(A, B)$, and \texttt{split}$(C)$. When a leaf top tree cluster $C$ for an edge $e$ is constructed with \texttt{create}$()$, we can obtain $\mathcal L(C)$ in constant time since the endpoints of $e$ are incident to only a constant number of edges of $H$. Supporting \texttt{split}$(C)$ takes $O(\rho)$ time since we simply remove $\mathcal L(C)$ which has length $O(\rho)$.

It remains to support \texttt{join}$(A, B)$. For the output top tree cluster $C$, we compute $\mathcal L_{\mathit{region}}(C)$ by traversing $\mathcal L_{\mathit{region}}(A)$ and $\mathcal L_{\mathit{region}}(B)$ in parallel in sorted order and merging these into a single list with duplicates removed. To compute $\mathcal L_{\mathit{edge}}(C)$, we similarly merge $\mathcal L_{\mathit{edge}}(A)$ and $\mathcal L_{\mathit{edge}}(B)$ into a single list but instead of removing duplicates, we do as follows: if we encounter two elements $(e_1,\ell)$ and $(e_2,\ell)$ with the same region label $\ell$, we only add one of the elements to $\mathcal L_{\mathit{edge}}(C)$, namely the one whose edge has the smaller weight. This correctly computes the auxiliary data for $C$ and takes $O(\rho)$ time.

It follows that each topological change in $\hat F'$ can be supported in $O(\rho\log n)$ time.

\paragraph{Non-topological changes:}
A non-topological change in $\hat F'$ occurs when a label changes in a pair $\mathcal R(e)$ for an edge $e\in H - F$ and when an edge is inserted into or deleted from $H - F$.

Consider first the case where a label in $\mathcal R(e)$ changes from $\ell$ to $\ell'$ for an endpoint $u$ of some edge $e\in H - F$. Since $\hat F'$ has constant degree, the number of top tree clusters containing $u$ on each level in the binary rooted tree representation of $\hat F'$ is $O(1)$ for a total of $O(\log n)$ over all levels, and we can find these top tree clusters in $O(\log n)$ time. We process them bottom-up. Let $C$ be the current top tree cluster. If $C$ is a leaf cluster, we can update its auxiliary data in $O(1)$ time since $H$ has $O(1)$ degree. Otherwise, let $A$ and $B$ be the child top tree clusters of $C$. To update $\mathcal L_{\mathit{region}}(C)$, we search for $\ell$ in $\mathcal L_{\mathit{region}}(A)$ and $\mathcal L_{\mathit{region}}(B)$. If $\ell$ is not found in either of the two lists, it is removed from $\mathcal L_{\mathit{region}}(C)$. We then add $\ell'$ to $\mathcal L_{\mathit{region}}(C)$ if it is not already in this list. To update $\mathcal L_{\mathit{edge}}(C)$, we search for an entry of the form $(e_A,\ell)$ in $\mathcal L_{\mathit{edge}}(A)$ and an entry of the form $(e_B,\ell)$ in $\mathcal L_{\mathit{edge}}(B)$; in case $e_A$ resp.~$e_B$ is undefined, regard it as a dummy edge of infinite weight. We remove the entry of the form $(e_C,\ell)$ in $\mathcal L_{\mathit{edge}}(C)$ (if any). If at least one of $e_A$ and $e_B$ is defined, we add a new pair $(e_C',\ell)$ to $\mathcal L_{\mathit{edge}}(C)$ where $e_C'$ is an edge of smaller weight among $e_A$ and $e_B$. Similar operations are done for label $\ell'$.

Using standard red-black tree operations to update the lists on each of the $O(\log n)$ levels, it follows that the update to $\hat F'$ caused by a single label change can be supported in $O(\log^2n)$ time. In a similar manner, the update to $\hat F'$ caused by the insertion or deletion of an edge in $H - F$ can be supported in $O(\log^2n)$ time.

From the above observations, we can now bound the time for a call \texttt{connect}$(u,v)$ in $\hat F'$ and for maintaining $\hat F'$.
\begin{Lem}\label{Lem:MaintainHatFPrime}
A call \texttt{connect}$(u,v)$ takes $O(\rho\log n)$ time where $\rho$ is the maximum number of regions in a tree of $F$. Updating $\hat F'$ after a call to \texttt{MergeRegion} or \texttt{SplitRegion} can be done in $O(r\log^2n)$ time where $r$ is the maximum value $N(R)$ of any region $R$. Updating $\hat F'$ takes $O(\log^2n)$ time after an edge insertion or deletion in $H - F$, takes $O(\rho\log n)$ time after an edge insertion in $F$, and takes $O(r\log^2n + \rho\log n)$ time after an edge deletion in $F$.
\end{Lem}
\begin{proof}
The first part of the lemma was shown above.

Neither a call \texttt{MergeRegion} nor a call \texttt{SplitRegion} results in topological changes to $\hat F'$. The number of label changes for such a call is $O(r)$ and by the above, the corresponding non-topological changes in $\hat F'$ can be supported in $O(r\log^2n)$ time.

The insertion of deletion of an edge in $H - F$ causes a non-topological change in $\hat F'$ which can be supported in $O(\log^2n)$ time, as shown above. Inserting an edge in $F$ causes a topological change in $\hat F'$ which we can support in $O(\rho\log n)$ time, again by the above.

Deleting an edge $e$ from $F$ causes a topological change in $\hat F'$ which can be supported in $O(\rho\log n)$ time. If $h(e) = 0$, the deletion of $e$ splits the region $R_e$ containing $e$. This causes $O(r)$ labels to be updated and the corresponding non-topological changes to $\hat F'$ can be supported in $O(r\log^2n)$ time.
\end{proof}

\subsection{Keeping regions balanced}
We now focus on maintaining regions in such a way that we can simultaneously get good bounds for the above defined parameters $r$ and $\rho$. We leave these values unspecified for now since their optimal choices will be easier to derive later.

\subsubsection{Preprocessing}
The preprocessing step is as follows. We are given a forest $F$ in $H$ as part of the input such that the number of non-tree edges is at most $h$. In $O(n+h)$ time, we can obtain $f$- and $h$-bitmaps and label pairs $\mathcal R(e)$ and find a partition into regions such that $N(R)\leq r$ for each region $R$ and such that for any region $R$ which is not equal to a component of $H$, $N(R)\geq r/3$. We obtain such a set of regions by applying Frederickson's \texttt{FINDCLUSTERS} procedure~\cite{Frederickson} to each tree in $F$ with the slight modification that during the construction of regions, instead of keeping track of number of vertices, the modified procedure keeps track of the number of endpoints of non-tree edges that are incident to the region currently being built; here, a non-tree edge with both endpoints in a region contributes a value of $2$ to the number of these endpoints.

Each region is given an arbitrary unique label from $\{1,\ldots,n\}$ except regions $R$ where $N(R) = 0$ and $R$ is a component of $H$. The list $\mathcal L$ of unused labels is set up in $O(n)$ time. In time $O(n+h)$, we initialize another list $\mathcal L'$ consisting of tuples $(\ell_R, v_R, s_R)$ for each label $\ell_R$ of a region $R$ where $v_R\in R$ is a representative vertex of $R$, and $s_R = N(R)$. We keep these tuples sorted by both label value and by $s_R$-value and like $\mathcal L$, list $\mathcal L'$ will be maintained during updates. We implement $\mathcal L'$ using two red-black trees, one for each of the two sorted orders, so that each update to it takes $O(\log n)$ time; to simplify the presentation, we shall simply refer to $\mathcal L'$ as a list. It is not hard to extend the above operations without an increase in running time so that whenever a labeled region $R$ is updated, its representative vertex $v_R$ and $s_R$-value are updated accordingly.

Setting up $\hat F$ can be done in $O(n)$ time. Within the same time bound, we can set up $\hat F'$, excluding the time to form auxiliary data for each top tree cluster. The latter can be done in $O(h\log^2n)$ time since the total length of all lists $\mathcal L_{\mathit{region}}(C)$ and $\mathcal L_{\mathit{edge}}$ over all top tree clusters $C$ on a single level of the top tree is $O(h)$ and it takes $O(h\log n)$ time to sort them.

\subsubsection{Handling updates}
Now, consider edge insertions and deletions. For a suitably large constant $c$ (that we leave unspecified for now), our data structure will maintain the following invariants:
\begin{description}
\item [Upper bound invariant:] $N(R)\leq cr\log n$ for any region $R$,
\item [Lower bound invariant:] $N(R)\geq \frac 1 3 r$ for any region $R$ which is not equal to a component of $H$.
\end{description}
The preprocessing step ensures that both invariants hold initially. The lower bound invariant implies that at all times, the maximum number of regions in any tree of $F$ is $O(h/r)$.

By Lemmas~\ref{Lem:MaintainHatF} and~\ref{Lem:MaintainHatFPrime}, updating $\hat F$, $\hat F'$, $f$- and $h$-bitmaps, and label pairs $\mathcal R(\cdot)$ after an edge update in $H$ takes $O(r\log^2n + \rho\log n)$ time.

\paragraph{Maintaining the lower bound invariant:}
An edge update in $H$ may cause the invariants to be violated. To reestablish them, we first focus on the lower bound invariant. It can be violated because an edge of $H - F$ is deleted or if an intra-region edge of $F$ is deleted, causing a region to be split in two.

Consider first the case where an edge $e\in H - F$ is deleted. From $\mathcal R(e)$, we obtain labels of the at most two regions incident to $e$. For each such region $R$, we obtain from its label $\ell_R$ the tuple $(\ell_R,v_R,s_R)$ in $\mathcal L'$ and we check if $s_R = N(R)$ violates the invariant. If so, we apply \texttt{MergeRegion}$(v_R)$. If $R$ was not merged with any region, $R$ must be a component of $H$ and hence cannot violate the invariant. Otherwise, we obtain $s_{R'} = N(R')$ and the label $\ell_{R'}$ for the merged region $R'$ together with its representative vertex $v_{R'}$ with a call to \texttt{FindRegionTree}$(u_R)$, and we replace $(\ell_R,v_R, s_R)$ with $(\ell_{R'},v_{R'}, s_{R'})$ in $\mathcal L'$. Note that $R'$ cannot violate the lower bound invariant since the region that $R$ was merged with did not violate it.

Now, consider the case where an intra-region edge $e$ belonging to $F$ is deleted. For each of the two subregions, $R$, of the split region, we obtain its label $\ell_R$ by applying \texttt{FindRegionTree} to the endpoint of $e$ in $R$ and from it identify a marked vertex in $R$; $\ell_R$ is obtained from $\mathcal R(e')$ for one of the edges $e'$ incident to this vertex. In case $R$ contains no marked vertex, it has no incident non-tree edge and so $R$ must be a component of $H$ in which case it does not violate the invariant. Otherwise, we proceed as above by applying \texttt{MergeRegion}$(v_R)$ to maintain the invariant for $R$.

The time for the updates above is $O(r\log^2n+\rho\log n)$ by Lemmas~\ref{Lem:MaintainHatF} and~\ref{Lem:MaintainHatFPrime}.

\paragraph{Maintaining the upper bound invariant:}
Next we focus on maintaining the upper bound invariant. Note that initially, $N(R)\leq r$ for each region $R$. To ensure that $N(R)$ never exceeds $r$ by more than a logarithmic factor, we employ the following simple greedy procedure. After each update, find a tuple $(\ell_R, v_R, s_R)$ in $\mathcal L'$ with maximum $s_R$-value and apply \texttt{SplitRegion}$(v_R,r)$.

Total time per update for applying this greedy procedure is $O(r\log^2n+\rho\log n)$. To show that the greedy procedure maintains the upper bound invariant, we need Lemma~\ref{Lem:Pebble} below which follows fairly easily from a result in~\cite{Pebble}.

First, we need some definitions. For a vector $\vec v$ in $\mathbb R^{\mathbb N}$, let $\vec v[i]$ denote the $i$th coordinate of $\vec v$. We let $\delta(\vec v)$ denote the vector obtained from $\vec v$ by replacing $\vec v[i]$ with $\max\{\vec v[i] - r, 0\rangle$ for all $i\in\mathbb N$. For two vectors $\vec v$ and $\vec w$ in $\mathbb R^{\mathbb N}$, we say that $\vec v$ \emph{dominates} $\vec w$ if $\vec v[i]\geq \vec w[i]$ for all $i\in\mathbb N$.
\begin{Lem}\label{Lem:Pebble}
Consider a finite dynamic set of objects distributed into bins $b_1,b_2,\ldots$ such that initially each bin contains at most $r$ objects and such that each update is of one of the following types:
\begin{description}
\item [\texttt{addtobin}$(o_1,o_2,\ldots)$:] adds $o_i$ objects to $b_i$ for $i = 1,2,\ldots$ where $\sum_{i = 1}^{\infty}o_i \leq r$,
\item [\texttt{removefrombin}$(i,k)$:] deletes $\min\{k,n_i\}$ objects from $b_i$ where $n_i$ is the number of objects in $b_i$ just prior to this update,
\item [\texttt{splitbin}:] picks a bin $b_i$ with maximum number of objects; if $b_i$ contains more than $r$ objects, these objects are distributed into empty bins such that each such bin contains at most $r$ objects after the update.
\end{description}
Let $k\in\mathbb N$ be a constant and consider a sequence of $n$ updates such that every subsequence of $k$ consecutive updates includes at least one call to \texttt{splitbin}. Then at any time during this sequence, the maximum number of objects in any bin is $O(r\log n)$.
\end{Lem}
\begin{proof}
We can view a distribution $o_1,o_2,\ldots$ of objects into bins $b_1,b_2,\ldots$ as the vector $\langle o_1,o_2,\ldots\rangle\in\mathbb R^{\mathbb N}$. Now, consider a sequence of $n$ updates such that every subsequence of $k$ consecutive updates includes at least one call to \texttt{splitbin}. Let $\vec u_0,\vec u_1,\ldots,\vec u_n$ be the sequence of vectors where $\vec u_0$ is the distribution vector for objects prior to the first update and $\vec u_i$ is the distribution vector just after the $i$th update, $i = 1,\ldots,n$.

We will define a different sequence of distribution vectors $\vec v_0,\vec v_1,\ldots,\vec v_n$ such that $\vec v_i$ dominates $\delta(\vec u_i)$ for $i = 0,\ldots,n$. The initial vector $\vec v_0$ is the zero vector $\vec 0$; this vector is equal to $\delta(\vec u_0)$ and hence dominates it.

Having defined vectors $\vec v_0,\ldots,\vec v_{i-1}$ dominating $\delta(\vec u_0),\ldots,\delta(\vec u_{i-1})$, respectively, we will define a vector $\vec v_i$ dominating $\delta(\vec u_i)$. If the $i$th update is \texttt{addtobin}$(o_1,o_2,\ldots)$, we let $\vec v_i = \vec v_{i-1} + \langle o_1,o_2,\ldots\rangle$, ensuring that $\vec v_i$ dominates $\delta(\vec u_i)$. If the $i$th update is \texttt{removefrombin}$(i,k)$, we let $\vec v_i = \vec v_{i-1}$ which clearly also ensures that $\vec v_i$ dominates $\delta(\vec u_i)$.

Finally, if the $i$th update is \texttt{splitbin}, let $j_{\max}$ be a coordinate of maximum value in $\vec v_{i-1}$ and let $j_{\min}$ be a coordinate of value $0$ in $\vec v_{i-1}$ (such coordinates must exist). Let $b_{j_{\max}'}$ be the bucket which is split by \texttt{splitbin}. We let $\vec w_i$ be the vector such that $\vec w_i[j] = \vec v_{i-1}[j]$ for $j\in\mathbb N - \{j_{\max},j_{\min}\}$ and $\vec w_i[j_{\min}] = \vec w_i[j_{\max}] = \lfloor\frac 1 2 \vec v_{i-1}[j_{\max}]\rfloor$. Note that for each $j\in\mathbb N$, $\delta(\vec u_i)[j]$ is either equal to $\delta(\vec u_{i-1})[j]$ or is equal to $0$; this follows since \texttt{splitbin} distributes objects from $b_{j_{\max}'}$ into empty buckets so that each of these buckets contains at most $r$ objects after the update. We define $\vec v_i$ as $\vec w_i$ with coordinates $j_{\max}$ and $j_{\max}'$ swapped and it follows that $\vec v_i$ dominates $\delta(\vec u_i)$.

For $i = 0,\ldots,n$ and for all $j\in\mathbb N$, we have shown that $\vec u_i[j]\leq\delta(\vec u_i)[j] + r\leq\vec v_i[j] + r$. The lemma will thus follow if we can show that each coordinate in $\vec v_i$ is $O(r\log n)$, for $i = 0,\ldots,n$.

Note that the total value added to all coordinates of $\vec v$-vectors is $O(r)$ between two consecutive \texttt{splitbin} updates or before the first \texttt{splitbin} or after the last \texttt{splitbin} update; this holds since each such subsequence consists of at most $k = O(1)$ updates. In each of these subsequences, we can only increase the values of coordinates of these vectors, never decrease them. Each \texttt{splitbin} update corresponds to the splitting operation from~\cite{Pebble} to the piles defined by the $\vec v$-vectors. It follows from that paper that for $i = 0,\ldots,n$, each coordinate of $\vec v_i$ is $O(r\log n)$, as desired.
\end{proof}

By viewing regions as bins and the number of objects in a bin corresponding to a region $R$ as $N(R)$, it follows from Lemma~\ref{Lem:Pebble} that our upper bound invariant holds for a sufficiently large constant $c$. The lemma actually implies an additional result, namely that an update consisting of a batched insertion of any number $k = O(r)$ of edges in $H - F$ can be supported while still maintaining both the upper and lower bound invariants. The lemmas above imply that the time for this operation plus the time to reestablish the invariants is $O(k\log^2n + r\log^2n + \rho\log n) = O(r\log^2n + \rho\log n)$.

By setting $r = \sqrt{h/\log n}$ and $\rho = h/r = \sqrt{h\log n}$, we get the following lemma which may be of independent interest.
\begin{Lem}\label{Lem:FewNonTreeEdges}
Let $H = (V,E_H)$ be a dynamic $n$-vertex graph undergoing insertions and deletions of weighted edges where the initial edge set $E_H$ need not be empty and where the number of non-tree edges never exceeds the value $h$. Then there is a data structure which after $O(n + h\log^2n)$ worst-case preprocessing time can maintain a forest $F$ in $H$ in $O(\sqrt h\log^{3/2}n)$ worst-case time per update where an update is either inserting an edge in $F$, deleting an edge in $H$, the operation \texttt{connect}, or a batched insertion of up to $\Theta(\sqrt{h/\log n})$ edges in $H - F$. An arbitrary initial forest $F$ may be specified as part of the input to the preprocessing algorithm.
\end{Lem}

We are now ready to prove Theorem~\ref{Thm:MSFFewNonTreeEdges}. With Prim's algorithm implemented with binary heaps, we can compute the initial MSF $F$ in $O((n + h)\log n)$ worst-case time. Setting up the data structure $\mathcal D$ of Lemma~\ref{Lem:FewNonTreeEdges} takes $O(n + h\log^2n)$ worst-case time. In $O(n)$ time, a top tree $\hat F$ for $F$ is set up to support queries of the form ``given vertices $u$ and $v$ in $\hat F$, is there a $uv$-path in $\hat F$ and if so, what is the heaviest edge on this path?''. With standard top tree operations, each edge insertion/deletion and each query in $\hat F$ can be executed in $O(\log n)$ time. In the following, when we refer to updates in $F$ and $H$, these updates are applied to $\mathcal D$.

Supporting the insertion of an edge $e = (u,v)$ in $H$ is done as follows. First, we query $\hat F$ with the pair $(u,v)$. If no $uv$-path exists, $e$ is inserted in $F$.

Now, suppose a $uv$-path does exist in $\hat F$ and let $e_{\max}$ be the heaviest edge on this path that the query to $\hat F$ returns. If $w(e) \ge w(e_{\max})$, $e$ is inserted as a non-tree edge in $H$. Otherwise, $e_{\max}$ is deleted from $F$ and reinserted as a non-tree edge in $H$ and $e$ is inserted in $F$.

Supporting the deletion of an edge $e$ from $H$ is done as follows. First, we apply a delete operation to delete $e$ from $H$. If $e$ was in $F$, \texttt{connect}$(u,v)$ is applied and if an edge is returned, it is inserted as a tree edge in $F$.

It follows from Lemma~\ref{Lem:FewNonTreeEdges} and the above description that $F$ can be maintained in $O(\sqrt h\log^{3/2}n)$ worst-case time per update where an update is the insertion of deletion of a single edge in $H$. It also follows that a batched insertion that does not change $F$ can be supported within this time bound as well since all the edges inserted must belong to $H - F$. This shows Theorem~\ref{Thm:MSFFewNonTreeEdges}.

\section{Partitioning a Graph Into Expander Subgraphs}\label{sec:PartitionExpanders}
In this section, we prove Theorem~\ref{Thm:RCPartition} from Section~\ref{subsec:PreprocRestrictedDecMSF} by giving an algorithm which, given any non-empty subset $W$ of $V$ of size $\Omega(n^{1-\eps})$ respecting $\mathcal C$, finds a partition $\mathcal X$ of $W$ satisfying the requirements in the theorem. To simplify notation, we shall only present the algorithm for the case where $W = V$. To generalize this to arbitrary subsets $W$ respecting $\mathcal C$, one issue is that the lower bound on the probability that the algorithm succeeds is of the form $1 - O(1/|W|^c)$ rather than $1 - O(1/n^c)$. However, since $W = \Omega(n^{1-\eps})$, we can get a probability of $1 - O(1/n^d)$ for an arbitrarily big constant $d$ by choosing $c = d/(1-\eps)$. We take care of the remaining issues at the end of this section.

%The goal is to obtain a partition $\mathcal X$ of $V$ such that w.h.p., each induced graph $G[X]$ for $X\in\mathcal X$ is a $\gamma$-expander and the number of inter-expander cluster edges of $G$ is at most $\lambda n$.

For any $\theta > 0$, define $\theta_+ = \theta^3$ and its inverse $\theta_- = \theta^{1/3}$. Spielman and Teng~\cite{SpielmanTeng} presented a procedure called \texttt{Partition} with the properties stated in the following lemma.
\begin{Lem}[\cite{SpielmanTeng}]\label{Lem:Partition}
Let $H = (V_H, E_H)$ be a graph and let $\theta > 0$. Let $S\subseteq V_H$ satisfy $\mbox{Vol}_{V_H}(S)\leq\frac 2 3 \mbox{Vol}_{V_H}(V_H)$ and $\Phi_{V_H}(S)\leq\theta_+$. Let $\{D_j\}$ be the sets of cuts output by \texttt{Partition}$(H,\theta)$, and let $D = \cup_j D_j$. Then $\mbox{Vol}_{V_H}(D)\leq\frac{65}{72}\mbox{Vol}_{V_H}(V_H)$, and the following two properties hold
\begin{enumerate}
\item with probability at least $1 - 1/|E_H|^3$,
\[
  \mbox{either }\max\mbox{Vol}_{V_H}(D_j)\geq\frac{35}{144}\mbox{Vol}_{V_H}(V_H)\mbox{ or }
  \mbox{Vol}_{V_H - D}(S\cap (V_H - D))\leq\frac 1 2 \mbox{Vol}_{V_H}(S),
\]
\item with probability $1 - \tilde O(1/(\theta^5|E_H|^3))$, $\Phi_{V_H}(D) = \tilde O(\theta)$.
\end{enumerate}
\texttt{Partition} runs in time $\tilde O(|E_H|/\theta^5)$.
\end{Lem}
This lemma is not quite identical to~\cite{SpielmanTeng} since they define $\theta_+$ as $\theta^3/(14^4\ln^2(|E_H|e^4))$ whereas we use the simpler definition $\theta_+ = \theta^3$. It is easy to see that with this simplification, we only lose $\log n$-factors in the second condition and in the running time of Lemma~\ref{Lem:Partition}.

With the definitions in Theorem~\ref{Thm:RCPartition}, we shall need a variant of \texttt{Partition} which when applied to $H[W]$ ensures that the output cut respects $\mathcal C$. This variant, which we refer to as \texttt{CPartition}, is presented in Lemma~\ref{Lem:CPartition} below. \texttt{CPartition} applies \texttt{Partition} to a graph $H'[W]$ obtained from $H[W]$ by replacing $H[C]$ with a sparse $1$-expander graph for $V(C)$ for each $C\in\mathcal C$ belonging to $W$. Then \texttt{CPartition} modifies the cut $(D',W-D')$ output by \texttt{Partition} to a new cut $(D,W-D)$ respecting $\mathcal C$ such that if $(D',W-D')$ had sufficiently low conductance in $H'[W]$ then $(D,W-D)$ has low conductance in $H[W]$.

Before presenting Lemma~\ref{Lem:CPartition}, we need Lemma~\ref{Lem:SimpleExpander} which shows how to efficiently find sparse $1$-expander graphs, and we need Lemma~\ref{Lem:LowConductanceMultigraph} which implies that we can modify a cut as sketched above.

We say that a graph is \emph{nowhere dense} if there is a constant $c$ such that every subgraph $S$ has at most $c|V(S)|$ edges.
\begin{Lem}\label{Lem:SimpleExpander}
Let $W$ be a vertex set of size $s$ and let $c > 0$ be any given constant. There is an $O(s)$ worst-case time algorithm constructing a simple graph $H$ for $W$ such that with probability $1 - O(1/s^c)$, $H$ is a nowhere dense $1$-expander graph for $W$ of max degree $O(\log s)$.

%Let $d > 0$ be a given constant integer, let $V$ be a vertex set of size $n$ and let $H$ be a multigraph with vertex set $V$ obtained by adding, for each $v\in V$, $d$ edges all with one endpoint in $v$ and the other endpoint chosen independently and uniformly at random in $V$. For sufficiently big constant $d$, $H$ is a simple $1$-expander with probability $1 - O(1/n^{d - 1})$.
\end{Lem}
\begin{proof}
Let $d > 0$ be a constant integer, to be specified later. In the proof, we shall implicitly assume that $s$ is larger than $d$ by a sufficiently big constant factor. We consider the algorithm that in a first phase constructs a simple graph $H$ with vertex set $W$ by adding, for each $v\in W$, $d$ edges all with one endpoint in $v$ and where the $i$th endpoint is chosen independently and uniformly at random among the remaining $s - i$ endpoints in $W - v$, for $i = 1,\ldots,d$. Note that $H$ need not be simple since it may contain two edges between any given pair of vertices $(u,v)$. In a second phase, the algorithm replaces each such pair by a single edge $(u,v)$. It is easy to implement the algorithm to have worst-case running time $O(ds) = O(s)$. Furthermore, any subgraph $H'$ of $H$ has $O(d|V(H')|) = O(|V(H')|)$ edges; this follows e.g.~by observing that the edges of $H'$ can be directed so that each vertex has at most $d$ outgoing edges. Hence, $H$ is nowhere dense. We will show that for sufficiently large $d$, $H$ satisfies also the remaining conditions of the lemma.

We first show that w.h.p., $H$ is a $1$-expander graph. We will use that for any integers $a\ge b\ge 0$, $\binom a b\le (ae/b)^b$ where we abuse notation and define $(a'/0)^{0} = 1$ for $a'\ge 0$. Furthermore, we exploit the fact that for any $t > 0$, the real function $x\mapsto (t/x)^x$ with domain $(0,\infty)$ achieves its maximum at $x = t/e$ with value $e^{t/e}$.

Let $k\in\{1,\ldots,\lfloor s/2\rfloor\}$ be given. We show that w.h.p., for every cut where the smaller side contains exactly $k$ vertices, the number of edges crossing this cut is greater than $k$. We shall only count the subset of edges chosen by the first phase of the algorithm when it processes the vertices on the side of the cut of size $k$. This number will be a lower bound on the final number of edges crossing the cut, obtained after the second phase. Hence, in the analysis, we can ignore this second phase. By a union bound, the probability that at least one cut with one side having size $k$ has fewer than $k$ edges crossing it is at most
\begin{align*}
\sum_{i = 0}^k\left(1 - \frac {k+d} s\right)^i\left(\frac k s\right)^{dk-i}\binom{dk}{i}\binom s k
  & \le \sum_{i = 0}^k\left(\frac k s\right)^{dk - i}\left(\frac{dke}i\right)^i\left(\frac{se}k\right)^k\\
  & \le \sum_{i = 0}^k\left(\frac k s\right)^{(d-1)k - i}\left(\frac{k}i\right)^i(de^2)^k\\
  & \le \sum_{i = 0}^k e^{(d-2)k\ln(k/s)}e^{k/e}(de^2)^k\\
  &   = (k+1)\left(e^{-(d-2)\ln(s/k) + 1/e + 2 + \ln d}\right)^k.
\end{align*}
We now consider two cases, $k\le\sqrt s$ and $k > \sqrt s$. If $k\le\sqrt s$, we get an upper bound on the probability of
\[
  (k+1)\left(e^{-\frac 1 2 (d-2)\ln(s) + 1/e + 2 + \ln d}\right)^1\le se^{-\frac 1 2 (d-2)\ln s + 1/e + 2 + \ln d}\le e^{-\frac 1 2 (d-4)\ln s + 1/e + 2 + \ln d}.
\]
We can choose $d$ sufficiently large so that this is at most $s^{-d/3 - 1}$.

Now, assume that $k > \sqrt s$. Using the fact that $k\le s/2$, we get an upper bound on the probability of
\[
  (k+1)\left(e^{-(d-2)\ln 2 + 1/e + 2 + \ln d}\right)^{\sqrt s}
  \le \left(e^{-(d-2)\ln 2 + 1/e + 3 + \ln d}\right)^{\sqrt s}.
\]
For sufficiently large constant $d$, this is also at most $s^{-d/3 - 1}$.

Taking a union bound over all choices for $k$ and picking constant $d$ sufficiently big, it follows that $H$ is a $1$-expander graph with probability at least $1 - 1/s^{d/3}$.

Finally, we show that w.h.p., $H$ has max degree $O(\log s)$. Fix a vertex $v\in W$ and let $v_1,\ldots,v_{s-1}$ be an arbitrary ordering of the remaining vertices. Introduce indicator variables $X_1,\ldots,X_{s-1}$ where for $i = 1,\ldots,s-1$, $X_i = 1$ iff $v$ is chosen as a random neighbor of $v_i$ when the first phase of the algorithm processes $v_i$. In total, $d$ edges are added from $v_i$. For $j = 1,\ldots,d$, consider the $j$th edge added from $v_i$. The probability that its other endpoint is $v$ is $0$ if $v$ was already chosen as an endpoint of one of the previous $j-1$ edges or the probability is at most $1/(s-j)\le 1/(s-d)$ since there are $s-j$ endpoints available for the $j$th edge. Since we may assume that $d \le (s+1)/2$, a union bound gives $\Pr(X_i = 1) \le\sum_{j = 1}^d 1/(s-d)\le 2d/(s-1)$ for $i = 1,\ldots,s-1$.

We observe that variables $X_1,\ldots,X_{s-1}$ are independent Poisson trials and that $X = \sum_{i = 1}^{s-1}X_i$ has expectation $\mu = E[X]\le 2d$. Let $\delta$ be the value such that $2d\ln s = (1+\delta)\mu$. Note that $1+\delta\ge\ln s$. We may assume that $e/\ln s \le 1/e$ and a Chernoff bound gives
\[
  \Pr(X > 2d\ln s) = \Pr(X > (1+\delta)\mu) < \left(\frac e{1+\delta}\right)^{(1+\delta)\mu}
                 \le \left(\frac e{\ln s}\right)^{2d\ln s}
                 \le s^{-2d}. 
\]
Observe that $v$ has degree $d + X$ after the first phase and hence degree at most $d + X$ after the second phase. It follows that with probability at least $1 - s^{-2d}$, $v$ has degree at most $d + 2d\ln s = O(\log s)$ after the second phase. A union bound shows that with probability at least $1 - s^{1-2d}$, $H$ has degree $O(\log s)$.

By a union bound and by picking $d\ge 3c$, it follows that with probability $1 - O(1/s^c)$, $H$ is a simple $1$-expander graph of max degree $O(\log s)$.
\end{proof}

\begin{Lem}\label{Lem:LowConductanceMultigraph}
Let $H$ and $\mathcal C$ be as in Theorem~\ref{Thm:RCPartition}. Let $H'$ be the graph obtained from $H$ by replacing, for each $C\in\mathcal C$, the edges of $H[C]$ with a simple nowhere dense $1$-expander graph of $V(C)$ with $O(|V(C)|)$ edges. Then for any subset $W$ of $V$ respecting $\mathcal C$ and for any cut $(S,W - S)$ of $W$,
\begin{enumerate}
\item $\Phi_{H[W]}(S) = O(\Phi_{H'[W]}(S))$ and if $(S,W-S)$ respects $\mathcal C$ then $\Phi_{H[W]}(S) = \Theta(\Phi_{H'[W]}(S))$,
\item $\Phi_{H'[W]}(S) = O(\Phi_{H[W]}(S)n^{\eps})$,
\item there is an $O(|W|)$ time algorithm which, assuming $\Phi_{H'[W]}(S)$ is less than a sufficiently small constant, obtains from $(S,W-S)$ a cut $(S',W-S')$ that respects $\mathcal C$ such that $\mbox{Vol}_{H'[W]}(S') = \Theta(\mbox{Vol}_{H'[W]}(S))$, $\mbox{Vol}_{H'[W]}(W-S') = \Theta(\mbox{Vol}_{H'[W]}(W-S))$, and $\Phi_{H'[W]}(S') = O(\Phi_{H'[W]}(S))$.
\end{enumerate}
\end{Lem}
\begin{proof}
We split the proof into three parts, corresponding to the three cases in the lemma.
\paragraph{Part $1$:}
First note that since every set of $\mathcal C$ contained in $W$ has size greater than $1$ and induces a connected subgraph of both $H[W]$ and $H'[W]$, every vertex has degree at least $1$ in both $H[W]$ and $H'[W]$. Since $H[W]$ has constant degree, we have $\mbox{Vol}_{H[W]}(S) = \Theta(|S|)$ and $\mbox{Vol}_{H[W]}(W-S) = \Theta(|W-S|)$, and since $H'[W]$ is nowhere dense, we have $\mbox{Vol}_{H'[W]}(S) = \delta_{H'[W]}(S) + \Theta(|S|)$ and $\mbox{Vol}_{H'[W]}(W - S) = \delta_{H'[W]}(S) + \Theta(|W - S|)$.

Now, let $C\in\mathcal C$ be a subset intersecting both sides of $(S,W-S)$. The number of edges of $H[C]$ crossing $(S,W-S)$ is $O(\min\{|C\cap S|,|C\cap(W-S)|\})$ since $H$ has constant degree. The number of edges of $H'[C]$ crossing $(S,W-S)$ is at least $\min\{|C\cap S|,|C\cap(W-S)|\}$ since $H'[C]$ is a $1$-expander graph. Hence, $\delta_{H[W]}(S) = O(\delta_{H'[W]}(S))$.

To show that $\Phi_{H[W]}(S) = O(\Phi_{H'[W]}(S))$, we may assume that $\delta_{H'[W]}(S)\le\min\{|S|,|W - S|\}$ since otherwise, $\min\{\mbox{Vol}_{H'[W]}(S),\mbox{Vol}_{H'[W]}(W-S)\} = \delta_{H'[W]}(S) + \Theta(\min\{|S|,|W-S|\}) = \Theta(\delta_{H'[W]}(S))$, implying that $\Phi_{H'[W]}(S) = \Theta(1)$ and we trivially have $\Phi_{H[W]}(S)\le 1$. We get
\[
  \Phi_{H[W]}(S) = O\left(\frac{\delta_{H'[W]}(S)}{\min\{|S|,|W-S|\}}\right)
               = \Theta\left(\frac{\delta_{H'[W]}(S)}{\delta_{H'[W]}(S) + \min\{|S|,|W-S|\}}\right)
               = \Theta(\Phi_{H'[W]}(S)),
\]
as desired. Now assume that $(S,W-S)$ respects $\mathcal C$. Then $\delta_{H'[W]}(S) = \delta_{H[W]}(S)$ and note that $\delta_{H[W]}(S) = O(\min\{|S|,|W-S|\})$. Again we may assume that $\delta_{H'[W]}(S)\le\min\{|S|,|W - S|\}$ since otherwise, $\min\{\mbox{Vol}_{H'[W]}(S),\mbox{Vol}_{H'[W]}(W-S)\} = \delta_{H'[W]}(S) + \Theta(\min\{|S|,|W-S|\}) = \Theta(\delta_{H'[W]}(S))$ and $\min\{\mbox{Vol}_{H[W]}(S),\mbox{Vol}_{H[W]}(W-S)\} = \Theta(\min\{|S|,|W-S|\}) = \Theta(\delta_{H[W]}(S))$ so both $\Phi_{H'[W]}(S)$ and $\Phi_{H[W]}(S)$ are $\Theta(1)$. We get
\[
  \Phi_{H[W]}(S) = \Theta\left(\frac{\delta_{H'[W]}(S)}{\min\{|S|,|W-S|\}}\right)
               = \Theta\left(\frac{\delta_{H'[W]}(S)}{\delta_{H'[W]}(S) + \min\{|S|,|W-S|\}}\right)
               = \Theta(\Phi_{H'[W]}(S)).
\]

\paragraph{Part $2$:}
Let $C\in\mathcal C$ intersect both sides of $(S,W-S)$. The number of edges of $H[C]$ crossing $(S,W-S)$ is at least $1$ since $H[C]$ is connected. The number of edges of $H'[C]$ crossing $(S,W-S)$ is $O(n^{\eps})$ since $H'[C]$ is sparse. Hence, we have
\[
  \Phi_{H'[W]}(S) = \Theta\left(\frac{\delta_{H'[W]}(S)}{\delta_{H'[W]}(S) + \min\{|S|,|W-S|\}}\right) = O\left(\frac{n^{\eps}\delta_{H[W]}(S)}{\min\{|S|,|W-S|\}}\right) = O(n^{\eps}\Phi_{H[W]}(S)),
\]
as desired.

\paragraph{Part $3$:}
Let $\mathcal C_1\subseteq\mathcal C$ consist of the subsets $C$ intersecting both sides of $(S,W - S)$ and $|C - S|\leq |C\cap S|$ and let $\mathcal C_2\subseteq\mathcal C$ consist of the remaining sets intersecting both sides of $(S,W-S)$. Let $S' = (S\cup\cup_{C\in\mathcal C_1} C) - (\cup_{C\in\mathcal C_2}C)$. Clearly, $(S',W-S')$ respects $\mathcal C$ and can be formed in $O(|W|)$ time.

For each $C\in\mathcal C_1$, the number of edges of $H'[W]$ crossing $(C\cap S,C - S)$ is at least $|C - S|$ since $H'[C]$ is a $1$-expander graph. Since $H[W]$ has constant degree and since the edges of $H'[W]$ crossing $(C - S,(W - S) - C)$ all belong to $H[W]$, the number of such edges is $O(|C - S|)$. Similarly, for each $C\in\mathcal C_2$, the number of edges of $H'[W]$ crossing $(C\cap S,C - S)$ is at least $|C\cap S|$ while the number of edges crossing $(C\cap S,S - C)$ is $O(|C\cap S|)$. It follows that the number of edges of $H'[W]$ crossing $(S',W - S')$ is $O(\Phi_{H'[W]}(S)\min\{\mbox{Vol}_{H'[W]}(S),\mbox{Vol}_{H'[W]}(W-S)\})$.

Next, we consider the volumes of $S'$ and $W - S'$ in $H'[W]$. Note that the number of edges of $H'[W]$ crossing $(S,W - S)$ is at least $\sum_{C\in\mathcal C_2}|C\cap S|$ and hence $\sum_{C\in\mathcal C_2}|C\cap S| \le \Phi_{H'[W]}(S)\mbox{Vol}_{H'[W]}(S)$. Similarly, the number of edges of $H'$ crossing $(S,W - S)$ is at least $\sum_{C\in\mathcal C_1}|C - S|$ so $\sum_{C\in\mathcal C_1}|C - S|\le \Phi_{H'[W]}(S)\mbox{Vol}_{H'[W]}(W-S)$.

Let $C\in\mathcal C_2$. Since $H'[W]$ is nowhere dense, we have $|H'[C\cap S]| = O(|C\cap S|)$ and since $H$ has constant degree, we have $\delta_{H'[W]}(C\cap S) = O(|C\cap S|)$. Combining this gives $\mbox{Vol}_{H'[S]}(C\cap S) = O(|C\cap S|)$.

We can now bound the volume of $S'$ in $H'[W]$ from below as follows:
\begin{align*}
\mbox{Vol}_{H'[W]}(S') & = \mbox{Vol}_{H'[W]}(S) + \sum_{C\in\mathcal C_1}\mbox{Vol}_{H'[W]}(C - S) - \sum_{C\in\mathcal C_2}\mbox{Vol}_{H'[W]}(C\cap S)\\
               & \ge \mbox{Vol}_{H'[W]}(S) - \delta_{H'[W]}(S) - \sum_{C\in\mathcal C_2}\mbox{Vol}_{H'[S]}(C\cap S)\\
               & = \mbox{Vol}_{H'[W]}(S) - O(\Phi_{H'[W]}(S)\mbox{Vol}_{H'[W]}(S) + \sum_{C\in\mathcal C_2}|C\cap S|)\\
               & = (1 - O(\Phi_{H'[W]}(S)))\mbox{Vol}_{H'[W]}(S),
\end{align*}
and similarly, we get $\mbox{Vol}_{H'[W]}(W - S') = (1 - O(\Phi_{H'[W]}(S)))\mbox{Vol}_{H'[W]}(W - S)$. Hence, assuming $\Phi_{H'[W]}(S)$ is below a sufficiently small constant, we have $\mbox{Vol}_{H'[W]}(S') = \Theta(\mbox{Vol}_{H'[W]}(S))$ and $\mbox{Vol}_{H'[W]}(W-S') = \Theta(\mbox{Vol}_{H'[W]}(W-S))$ and hence $\Phi_{H'[W]}(S') = O(\Phi_{H'[W]}(S))$. This shows the third part of the lemma.
\end{proof}

We are now ready to present our algorithm \texttt{CPartition}.
\begin{Lem}\label{Lem:CPartition}
Let $\theta\ge n^{-\eps/2}$, let $c > 0$ be a constant, and let $H'$, $\mathcal C$, and $W$ be as in Theorem~\ref{Thm:RCPartition} and Lemma~\ref{Lem:LowConductanceMultigraph}. Let $(S,W-S)$ be a cut such that $\mbox{Vol}_{H'[W]}(S)\le\mbox{Vol}_{H'[W]}(W-S)$ and $\Phi_{H'[W]}(S)\leq\theta_+$. If $\theta$ is less than $\log^{-c'}n$ for a sufficiently large constant $c'$, then there is a constant $d$ with $0 < d < 1$ and an algorithm \texttt{CPartition}$(H'[W],\theta,c)$ which outputs a cut $(D,W - D)$ respecting $\mathcal C$ such that with probability at least $1 - 1/n^c$, we have $\mbox{Vol}_{H'[W]}(D)\leq(1-d)\mbox{Vol}_{H'[W]}(W)$ as well as the following two conditions:
\begin{enumerate}
\item $\mbox{Vol}_{H'[W]}(D) = \Omega(\mbox{Vol}_{H'[W]}(S))$,
\item $\Phi_{H'[W]}(D) = \tilde O(\theta)$.
\end{enumerate}
\texttt{CPartition} runs in worst-case time $\tilde O(|W|/\theta^5)$.
\end{Lem}
\begin{proof}
Before describing \texttt{CPartition}, we first modify \texttt{Partition}$(H,\theta)$ from Lemma~\ref{Lem:Partition} slightly so that if every component of $H = (V_H,E_H)$ has volume at most $\frac 1 2\mbox{Vol}_{V_H}(V_H)$ then a cut $D$ is output such that $\frac 1 3\mbox{Vol}_{V_H}(V_H)\le\mbox{Vol}_{V_H}(D)\le\frac 2 3\mbox{Vol}_{V_H}(V_H)$. Such a cut is obtained with a simple $O(|H|)$ time greedy algorithm that starts with $(D,V_H-D) = (\emptyset,V_H)$, then considers the components in order of decreasing volume, and adds the current component to the side of $(D,V_H - D)$ with smaller volume. This cut satisfies the requirements of Lemma~\ref{Lem:Partition} and allows us to only use randomization when $|E_H| = \Theta(|V_H|)$.

Now, we are ready to describe algorithm \texttt{CPartition}$(H'[W],\theta,c)$. It consists of an outer loop consisting of $C\lceil\log n\rceil$ iterations for some constant $C > 0$ to be specified later. Initially, $D = \emptyset$. In each iteration, the (modified) algorithm \texttt{Partition}$(H'[W],\theta)$ is called; let $D'$ be the union of sets output by this call. If the bound on $\Phi_{H'[W]}(D')$ in the second property of Lemma~\ref{Lem:Partition} does not hold then the next iteration is executed. Otherwise, the algorithm in the third part of Lemma~\ref{Lem:LowConductanceMultigraph} is applied with $D'$ playing the role of $S$, giving a new cut $(S',W-S')$ respecting $\mathcal C$. If $\min\{\mbox{Vol}_{H'[W]}(S'),\mbox{Vol}_{H'[W]}(W - S')\} > \mbox{Vol}_{H'[W]}(D)$, $D$ is updated to the side of the cut $(S',W-S')$ of smaller volume in $H'[W]$. Once all iterations have been executed, the algorithm outputs $D$ and then halts.

Clearly, the set $D$ output by this algorithm respects $\mathcal C$. Since $H'$ is nowhere dense, we have $|H'[W]| = O(|W|)$. Excluding the time for obtaining $D'$, each iteration takes $O(|H'[W]|) = O(|W|)$ time by the third part of Lemma~\ref{Lem:LowConductanceMultigraph}. Each call to \texttt{Partition} takes $\tilde O(|E(H'[W])|/\theta^5) = \tilde O(|W|/\theta^5)$ time and computing the union $D'$ of sets can clearly be done within this time bound as well. Hence, the entire algorithm above runs in $\tilde O(|W|/\theta^5)$ time.

Let $(S,W-S)$ be a cut with $\mbox{Vol}_{H'[W]}(S)\le\mbox{Vol}_{H'[W]}(W-S)$ and $\Phi_{H'[W]}(S)\le\theta_+$. We need to show that the two conditions of the lemma are satisfied with probability at least $1 - 1/n^c$. Note that $\mbox{Vol}_{H'[W]}(D)$ can only increase over time and can never be larger than $\frac 1 2 \mbox{Vol}_{H'[W]}(W)$. Consider any iteration and let $S_{\min}'$ be the side of $(S',W-S')$ of smaller volume in $H'[W]$. It suffices to show that for sufficiently large constant $C$ and sufficiently small constant $d$, the two conditions of the corollary, with $D$ replaced by $S_{\min}'$, are satisfied with probability at least $1/2$ in any given iteration.

Consider an arbitrary iteration and let $D'$ be the union of set output by \texttt{Partition}$(H'[W],\theta)$. By Lemma~\ref{Lem:Partition}, $\mbox{Vol}_{H'[W]}(D')\le\frac{65}{72}\mbox{Vol}_{H'[W]}(W)$ and the following two properties hold with a certain probability that we show is at least $1/2$:
\begin{enumerate}
\item either $\mbox{Vol}_{H'[W]}(D')\ge\frac{35}{144}\mbox{Vol}_{H'[W]}(W)$ or $\mbox{Vol}_{H'[W - D']}(S\cap(W - D'))\le\frac 1 2 \mbox{Vol}_{H'[W]}(S)$,
\item $\Phi_{H'[W]}(D') = \tilde O(\theta)$.
\end{enumerate}
The probability that these properties hold is $1 - \tilde O(1/(\theta^5|W|^3))$ due to our modifification to \texttt{Partition} described above. Since $\theta\ge n^{-\eps/2}$ and since $|W| = \Omega(n^{\eps})$, we get a lower bound on the probability of $1 - \tilde O(n^{(5/2 - 3)\eps}) = 1 - \tilde O(n^{-\eps/2})$ which is at least $1/2$ for $n$ larger than some constant.

We assume in the following that both of these conditions hold. Since $\Phi_{H'[W]}(D') = \tilde O(\theta)$ and since $\theta < \log^{-c'}n$, picking $c'$ large enough ensures that $\Phi_{H'[W]}(D')$ is less than a small constant factor such that the condition in the third part of Lemma~\ref{Lem:LowConductanceMultigraph} is satisfied, with $D'$ playing the role of $S$. Hence, $\mbox{Vol}_{H'[W]}(S') = \Theta(\mbox{Vol}_{H'[W]}(D'))$, $\mbox{Vol}_{H'[W]}(W - S') = \Theta(\mbox{Vol}_{H'[W]}(W - D'))$, and $\Phi_{H'[W]}(S_{\min}') = \Theta(\Phi_{H'[W]}(D')) = \tilde O(\theta)$.

It remains to show that for suitable constant $d\in(0,1)$ either $\mbox{Vol}_{H'[W]}(S_{\min}')\ge d\mbox{Vol}_{H'[W]}(W)$ or $\mbox{Vol}_{H'[W]}(S_{\min}') = \Omega(\mbox{Vol}_{H'[W]}(S))$. Assume first that $\mbox{Vol}_{H'[W]}(D')\ge\frac{35}{144}\mbox{Vol}_{H'[W]}(W)$. Since also $\mbox{Vol}_{H'[W]}(D')\le\frac{65}{72}\mbox{Vol}_{H'[W]}(W)$, it follows from the above that $\mbox{Vol}_{H'[W]}(S') = \Theta(\mbox{Vol}_{H'[W]}(W))$ and that
\[
  \mbox{Vol}_{H'[W]}(W - S') = \Theta(\mbox{Vol}_{H'[W]}(W - D')) = \Theta(\mbox{Vol}_{H'[W]}(W) - \mbox{Vol}_{H'[W]}(D')) = \Theta(\mbox{Vol}_{H'[W]}(W)).
\]
Picking $d$ sufficiently small gives $\mbox{Vol}_{H'[W]}(S_{\min}') = \min\{\mbox{Vol}_{H'[W]}(S'),\mbox{Vol}_{H'[W]}(W-S')\}\ge d\mbox{Vol}_{H'[W]}(W)$. This shows the desired since $d\mbox{Vol}_{H'[W]}(W) = \Omega(\mbox{Vol}_{H'[W]}(S)$.

Finally, assume that $\mbox{Vol}_{H'[W]}(D') < \frac{35}{144}\mbox{Vol}_{H'[W]}(W)$ and $\mbox{Vol}_{H'[W - D']}(S\cap(W - D'))\le\frac 1 2 \mbox{Vol}_{H'[W]}(S)$. The latter implies that $\mbox{Vol}_{H'[W]}(D') = \Omega(\mbox{Vol}_{H'[W]}(S))$ so by the above, $\mbox{Vol}_{H'[W]}(S') = \Omega(\mbox{Vol}_{H'[W]}(S))$. Hence, if $\mbox{Vol}_{H'[W]}(S')\le \mbox{Vol}_{H'[W]}(W - S')$, we get $\mbox{Vol}_{H'[W]}(S_{\min}') = \Omega(\mbox{Vol}_{H'[W]}(S))$ as desired. If $\mbox{Vol}_{H'[W]}(S') > \mbox{Vol}_{H'[W]}(W - S')$ then since $\mbox{Vol}_{H'[W]}(D')\le\frac{65}{72}\mbox{Vol}_{H'[W]}(W)$, we have
\[
  \mbox{Vol}_{H'[W]}(S_{\min}') = \mbox{Vol}_{H'[W]}(W - S') = \Theta(\mbox{Vol}_{H'[W]}(W - D')) = \Theta(\mbox{Vol}_{H'[W]}(W)) = \Omega(\mbox{Vol}_{H'[W]}(S)),
\]
again showing the desired.
\end{proof}

We will give a recursive version of \texttt{CPartition} called \texttt{RCPartition} which w.h.p.~outputs our desired partition $\mathcal X$. Let $\xi > 0$ be a given constant and let $\tau > 0$ be a constant to be specified later. Let $f_{+}:\mathbb R_+\rightarrow\mathbb R_+$ be a function mapping a value $\theta$ to a value which is $\tilde\Theta(\theta_+) = \tilde\Theta(\theta^3)$ so that the cut output by a call \texttt{CPartition}$(H,f_{+}(\theta),c)$ has conductance at most $\theta_+/\log n$, assuming the two conditions in Lemma~\ref{Lem:CPartition} are satisfied. Let $f_- = f_+^{-1}$ be its inverse.

In the following, let $H$, $H'$, and $\mathcal C$ be as in Theorem~\ref{Thm:RCPartition} and Lemma~\ref{Lem:LowConductanceMultigraph}. Pseudocode for \texttt{RCPartition} can be seen in Figure~\ref{fig:RCPartition}; we assume it has access to $H'$ and $\mathcal C$. Define $\theta_{\mathit{init}} = f_-^{(\lceil 1/\xi\rceil)}((n^{-\tau})_-)$. The initial call has parameters $W = V$, $\theta = \theta_{\mathit{init}}$, and $d = 1$.

%The probability of success of the implementation of \texttt{Partition} by Spielman and Teng is $1 - m^{-3}$. Our implementation is a slight modification which calls their algorithm twice and outputs the larger of the two sets $D$ for which the conductance is at most $(260/7)\theta\lg^2m$; if none of the sets have this property, the algorithm restarts to find another two sets.

\begin{figure}[!ht]
\begin{tabbing}
\rule{\linewidth}{\arrayrulewidth}\\
d\=dd\=\quad\=\quad\=\quad\=\quad\=\quad\=\quad\=\quad\=\quad\=\quad\=\quad\=\quad\=\kill
\>\texttt{RCPartition}$(W, \theta, d, c)$:\\\\
\>1. \>\>let $D$ be the output of \texttt{CPartition}$(H'[W], \theta, c+1)$\\
\>2. \>\>if $\mbox{Vol}_{H'[W]}(D) < n^{2\tau}$ then return $\{W\}$\\
\>3. \>\>if $\mbox{Vol}_{H'[W]}(D) > n^{1 - d\xi}$ then return $\mbox{\texttt{RCPartition}}(D, \theta, d, c)\cup
                 \mbox{\texttt{RCPartition}}(W - D, \theta, d, c)$\\
\>4. \>\>else return $\mbox{\texttt{RCPartition}}(W, f_+(\theta), d+1, c)$\\

\rule{\linewidth}{\arrayrulewidth}
\end{tabbing}
\caption{Pseudocode for the recursive algorithm \texttt{RCPartition} which outputs a partition of $V$ with the properties stated in Lemma~\ref{Lem:RCPartition}.}\label{fig:RCPartition}
\end{figure}

\begin{Lem}\label{Lem:RCPartition}
Let $c > 0$, $\xi > 0$, and $\tau > 0$ be constants where $\tau\le\frac 3 2 \eps$. Algorithm \texttt{RCPartition}$(V, \theta_{\mathit{init}}, 1, c)$ outputs a partition $\mathcal X$ of $V$ respecting $\mathcal C$ such that with probability at least $1 - 1/n^c$, the following three conditions hold:
\begin{enumerate}
\item for every $X\in\mathcal X$ and for every cut $(S, X - S)$ of $X$ where $\mbox{Vol}_{H[X]}(X - S)\geq \mbox{Vol}_{H[X]}(S) = \Omega(n^{2\tau})$, we have $\delta_{H[X]}(S) = \Omega(|S|n^{-\tau-\eps})$,
\item the number of edges of $H$ between distinct sets of $\mathcal X$ is $n^{-\tau/2^{O(1/\xi)}}\sum_{X\in\mathcal X}|X|\log(n/|X|)$, and
\item the worst-case running time of \texttt{RCPartition}$(V, \theta, 1, c)$ is $\tilde O(n^{1 + 5\tau + \xi})$.
\end{enumerate}
\end{Lem}
\begin{proof}
By Lemma~\ref{Lem:CPartition}, the output is a partition $\mathcal X$ respecting $\mathcal C$.

We observe that the number of leaves of the recursion tree is $O(n^{1 - 2\tau})$ since each leaf corresponds to a subset of $V$ of size $\Omega(n^{2\tau})$ and the subsets corresponding to all leaves form a partition of $V$. Also, the maximum possible value of $d$ in any recursive step is at most $1/\xi = O(1)$, implying that the total number of nodes of the recursion tree is $O(n^{1 - 2\tau})$. Note that for all values of $\theta$ in the recursive calls, $\theta > f_+(\theta)$. Hence, the upper bound on $d$ implies that in any recursive step, $\theta \ge f_+^{(\lceil 1/\xi\rceil)}(\theta_{\mathit{init}}) = (n^{-\tau})_-$. Since $x\mapsto x_+$ is a monotonically increasing function, we thus have $\theta_+ \ge ((n^{-\tau})_-)_+ = n^{-\tau}$. Since $\tau\le\frac 3 2 \eps$, we also have $\theta\ge (n^{-\tau})_- = n^{-\tau/3}\ge n^{-\eps/2}$, as required by Lemma~\ref{Lem:CPartition}.

The probability that a single call to \texttt{CPartition} in line $1$ succeeds to satisfy the conditions in Lemma~\ref{Lem:CPartition} is at least $1 - 1/n^{c+1}$. By a union bound over all nodes in the recursion tree, all calls to \texttt{CPartition} succeed with probability at least $1 - 1/n^c$. When we show the three conditions below, we assume that this indeed is the case for every call to \texttt{CPartition}.

\paragraph{Condition $1$:} Let $X\in\mathcal X$ be given and let $(S, X - S)$ be a cut of $X$ with $\Phi_{H[X]}(S)\le c'n^{-\tau - \eps}$ in $H[X]$ where $\mbox{Vol}_X(S)\le\mbox{Vol}_X(X - S)$ and $c' > 0$ is a constant specified below. We choose $S$ so that $\mbox{Vol}_X(S)$ is maximized over all such cuts $(S, X - S)$; if $S$ does not exist, condition $1$ cannot be violated for set $X$.

Consider the recursive call where $X$ is output in line $2$ and let $D$ be the set computed in line $1$ in that recursive call. We have $\mbox{Vol}_{H'[X]}(D) < n^{2\tau}$.  By Lemma~\ref{Lem:CPartition}, $X$ respects $\mathcal C$ so by the second part of Lemma~\ref{Lem:LowConductanceMultigraph}, $\Phi_{H'[X]}(S) = O(\Phi_{H[X]}(S)n^{\eps}) = O(c'n^{-\tau})$. We choose $c'$ sufficiently small so that $\Phi_{H'[X]}(S)\le n^{-\tau}\le \theta_+$. Applying Lemma~\ref{Lem:CPartition} gives $n^{2\tau} > \mbox{Vol}_{H'[X]}(D) = \Omega(\mbox{Vol}_{H'[X]}(S))$, implying that $\mbox{Vol}_{H'[X]}(S) = O(n^{2\tau})$. By the choice of $S$, we have shown that for any cut in $H[X]$ of conductance at most $c'n^{-\tau - \eps}$, one of the two sides of the cut has volume $O(n^{2\tau})$. This implies that for any cut $(S,X - S)$ where $\mbox{Vol}_{H[X]}(X - S)\ge\mbox{Vol}_{H[X]}(S) = \Omega(n^{2\tau})$, we have $\delta_{H[X]}(S) = \Omega(\mbox{Vol}_{H[X]}(S)n^{-\tau-\eps}) = \Omega(|S|n^{-\tau-\eps})$, showing the first condition.

\paragraph{Condition $2$:} For each call \texttt{CPartition}$(H'[W],\theta,c+1)$, let $D$ be the set output and consider charging the number of edges of $H'[W]$ crossing the cut $(D,W-D)$ evenly to the vertices on the smaller side of the cut. Then for each $X\in\mathcal X$, each vertex of $X$ is charged at most $\log(n/|X|)$ times. Since $\theta$ decreases with $d$, the second part of Lemma~\ref{Lem:CPartition} implies that the amount charged to the vertex each time is $\tilde O(\theta_{\mathit{init}}) = n^{-\tau/2^{O(1/\xi)}}$. This shows the second condition.

%For each fixed $d$, the calls to \texttt{CPartition} with that value of $d$ is to pairwise disjoint vertex sets $W$. Since there are less than $1/\xi = O(1)$ possible values of $d$ and since $H$ has $O(n)$ edges, the number of edges between distinct clusters of $\mathcal X$ is $\tilde O(n\theta_{\mathit{init}}) = $. 

\paragraph{Condition $3$:} Consider a fixed node of the recursion tree corresponding to a call \texttt{RCPartition}$(W,\theta,d,c)$ for which line $4$ is executed. Let $D_1,\ldots,D_k$ be the maximal subsets $D\subset W$ for which the two recursive calls in line $3$ are made with the first having input $(D,f_+(\theta),d+1,c)$. We order the sets such that $D_i$ is obtained before $D_{i+1}$ by the algorithm for $i = 1,\ldots,k-1$. Note that these are pairwise disjoint subsets of $W$. For $i = 0,\ldots,k$, let $S_i = \cup_{j = 1}^i D_i$.

Before showing condition $3$, we first show that $\mbox{Vol}_{H'[W]}(S_k) = O(n^{1 - d\xi})$. For some constant $C > 1$ to be specified below, we may assume that $\mbox{Vol}_{H'[W]}(W) > 2Cn^{1 - d\xi}$ and $\mbox{Vol}_{H'[W]}(S_k) > Cn^{1 - d\xi}$. We will show how to derive a contradiction when $C$ is sufficiently large.

Consider the largest index $k'\in\{0,\ldots,k-1\}$ for which $\mbox{Vol}_{H'[W]}(S_{k'}) \le Cn^{1 - d\xi} < \frac 1 2 \mbox{Vol}_{H'[W]}(W)$. For $i = 0,\ldots,k'$, since $D_{i+1}$ was obtained by a call \texttt{CPartition}$(H'[W-S_i],f_+(\theta),c+1)$, we have $\Phi_{H'[W-S_i]}(D_{i+1})\le \theta_+/\log n$. Hence
\[
  \delta_{H'[W]}(S_{k'+1}) \le \sum_{i = 0}^{k'}\delta_{H'[W - S_i]}(D_{i+1})
                      \le \frac{\theta_+}{\log n}\sum_{i = 0}^{k'}\mbox{Vol}_{H'[W - S_i]}(D_{i+1})
                      \le \frac{\theta_+}{\log n}\mbox{Vol}_{H'[W]}(S_{k'+1}).
\]
By the choice of $k'$ and by Lemma~\ref{Lem:CPartition}, it follows that 
$\mbox{Vol}_{H'[W]}(W - S_{k'+1}) = \Theta(\mbox{Vol}_{H'[W]}(W)) = \Omega(\mbox{Vol}_{H'[W]}(S_{k'+1}))$. Thus, $\Phi_{H'[W]}(S_{k'+1}) = O(\theta_+/\log n)$ so for $n$ bigger than some constant, we have $\Phi_{H'[W]}(S_{k'+1})\le\theta_+$. By the choice of $W$, \texttt{CPartition}$(H'[W],\theta,c+1)$ gave a set $D$ with $\mbox{Vol}_{H'[W]}(D)\le n^{1 - d\xi}$. Applying Lemma~\ref{Lem:CPartition} with $S_{k'+1}$ playing the role of $S$, we can choose constant $C$ large enough so that $\mbox{Vol}_{H'[W]}(S_{k'+1}) < C\mbox{Vol}_{H'[W]}(D)\le Cn^{1 - d\xi}$, contradicting the choice of $k'$.

It follows from the above that $|S_k| = O(\mbox{Vol}_{H'[W]}(S_k)) = O(n^{1 - d\xi})$. We can use this to bound the number of sets $D\subset W$ for which the test in line $3$ succeeds with parameter $d+1$. For each such $D$, we have $D\subseteq S_k$ and $\mbox{Vol}_{H'[W']}(D) > n^{1 - (d+1)\xi}$ for some $W'\subseteq W$. Since $H'$ is nowhere dense, since $H$ has constant degree, and since $D$ respects $\mathcal C$, we have $|D| = \Omega(|E(H'[D])| + |\delta_{H'[W']}(D)|) = \Omega(\mbox{Vol}_{H'[W']}(D)) = \Omega(n^{1 - (d+1)\xi}) = \Omega(|S_k|/n^{\xi})$ so the number of choices for $D$ is $O(n^{\xi})$.

We have shown that the total number of recursion nodes of the form $(W',f_+(\theta),d+1,c)$ with $W'\subseteq W$ is $O(n^{\xi})$. The time spent in each of them is dominated by a single call to \texttt{CPartition}$(W',f_+(\theta),c+1)$ which by Lemma~\ref{Lem:CPartition} takes worst-case time $\tilde O(|W'|/(f_+(\theta))^5) = \tilde O(|W|n^{5\tau})$. Summing over all choices of $W'$, this is $\tilde O(|W|n^{5\tau+\xi})$. Over all $W$, this is $\tilde O(n^{1+5\tau+\xi})$. Finally, summing over all $O(1/\xi) = O(1)$ choices of $d$ gives a total worst-case running time for \texttt{RCPartition}$(V,\theta_{\mathit{init}},1,c)$ of $\tilde O(n^{1+5\tau+\xi})$, showing the third condition.
\end{proof}
We are now ready to prove Theorem~\ref{Thm:RCPartition} from Section~\ref{subsec:PreprocRestrictedDecMSF} in the case where $W = V$. First, we construct $H'[W]$ which by Lemma~\ref{Lem:SimpleExpander} takes $O(|W|)$ time. We then apply Lemma~\ref{Lem:RCPartition} with $\tau = \eps$, $\gamma$ of the form $Cn^{-2\eps}$ for suitable constant $C$, $\lambda = \gamma^{1/2^{O(1/\xi)}}$ for suitable hidden constants. Denote by $\mathcal X'$ the output partition of $V$. We may assume that the conditions in the lemma are satisfied since this holds with probability at least $1 - 1/n^c$. Let $\mathcal X$ be the set of components in $H[X']$ over all $X'\in\mathcal X'$. Clearly, each set of $\mathcal X$ respects $\mathcal C$ and the second and third conditions of Theorem~\ref{Thm:RCPartition} hold with the above substitutions.

To show the first condition, let $X\in\mathcal X$ be given and consider a cut $(S,X - S)$ in $X$ where $\mbox{Vol}_X(S)\leq \mbox{Vol}_X(X - S)$. We will show that $\delta_{H[X]}(S) = \Omega(|S|n^{-2\eps})$. Assume first that $\mbox{Vol}_{H[X]}(S) = \Omega(n^{2\tau})$ and consider the cut $(S,X' - S)$ where $X\subseteq X'\in \mathcal X'$. By picking the hidden constant in $\Omega(n^{2\tau})$ sufficiently big, the first condition of Lemma~\ref{Lem:RCPartition} implies that $\delta_{H[X]}(S) = \delta_{H[X']}(S) = \Omega(|S|n^{-\tau-\eps}) = \Omega(|S|n^{-2\eps}))$, as desired.

Now assume that $\mbox{Vol}_{H[X]}(S) = O(n^{2\tau}) = O(n^{2\eps})$. Since $H[X]$ is connected, $\delta_{H[X]}(S)\ge 1 = \Omega(\mbox{Vol}_{H[X]}(S)n^{-2\eps}) = \Omega(|S|n^{-2\eps})$, again showing the desired. Hence, $H[X]$ is a $\gamma$-expander graph for suitable choice of constant $C$, completing the proof of Theorem~\ref{Thm:RCPartition} in the special case where $W = V$.

In the general case, $W = \Omega(n^{1 - \eps})$ and replacing $n$ by $|W|$ above, we get $\gamma = \Omega(n^{-2\eps})$ and $\lambda = \gamma^{1/2^{O(1/\xi)}}$. This shows Theorem~\ref{Thm:RCPartition}.

\section{Decremental Maintenance of Expander Graphs}\label{sec:DecDSLowCondCuts}
In this section, we present the data structure of Theorem~\ref{Thm:DisconnectExpanderGraph}. We shall refer to the dynamic graph as $G = (V,E)$ here rather than $H$ and to simplify notation, we assume it to have $n$ vertices; it is easy to see that the problem only becomes easier if there are fewer than $n$ vertices. We require that $G$ has max degree at most $3$, that it is initially a $\gamma$-expander graph w.h.p., where $\gamma = \Theta_{n^{\eps}}(1)$, and that the total number of edge deletions in $G$ is at most $\Delta = \Omega(\sqrt n)$. We also require that the sequence of updates is independent of the random bits used by the data structure. We regard $G$ as an unweighted graph since its edge weights are not relevant in this section. However, since we will apply Lemma~\ref{Lem:FewNonTreeEdges} which assumes an edge weight function, we pick some arbitrary lexicographical ordering of the edges of $G$. We shall refer to an FFE structure here as an instance of the data structure in this lemma and denote it by $\FFE{H}$ for a graph $H$.

\subsection{Preprocessing}\label{subsec:PreprocDecSF}
We start by describing the preprocessing step of our data structure. We may restrict our attention to the case where the initial graph $G$ is connected as follows. During preprocessing, the data structure checks if $G$ is connected. If $G$ is not, it cannot be a $\gamma$-expander graph and the data structure simply lets the first set $V_1$ output be equal to $V$ as this will satisfy the conditions in Theorem~\ref{Thm:DisconnectExpanderGraph}.

First, for some parameter $\kappa$ between $1$ and $n$, we apply Frederickson's \texttt{FINDCLUSTERS} procedure~\cite{Frederickson} to an arbitrary spanning tree of of $G$, giving a partition of $V$ into a set $\mathcal C$ of \emph{clusters} where for each $C\in\mathcal C$, $\kappa\le |C|\le 3\kappa$, and $G[C]$ is connected; we choose $\kappa = n^{1/2 - \Theta(\eps)}$ where the hidden constant will be picked sufficiently big (but independent of $\eps$ which we regard as a variable here) to make our arguments in this and the next two sections carry through. We compute a spanning tree $T(C)$ of $G[C]$ for each cluster $C$. The set $\mathcal C$ will be dynamic and our data structure maintains this clustering as well as spanning tree $T(C)$ of each $C\in\mathcal C$.

Next, we obtain a subset $E'$ of $E$ by sampling each edge independently with some probability $p$ to be specified later. We form a subgraph $H$ of $G$ consisting of edge set $E'$ and of $T(C)$ for each $C\in\mathcal C$. We apply Lemma~\ref{Lem:FewNonTreeEdges} to set up an FFE structure $\FFE{H}$ for $H$ where the initial forest $F(H)$ is a spanning forest of $H$ containing $T(C)$ for each $C\in\mathcal C$. The purpose of $F(H)$ will be to certify connectivity of $G[V - W_k]$ where $W_k$ is a subset satisfying the first requirement in Theorem~\ref{Thm:DisconnectExpanderGraph}.

%We check if $F(H)$ is a spanning tree of $H$ and if it is not, we let the first set $V_1$ output be equal to $V$. As we show later, for a suitable choice of $p$, the initial graph $H$ is connected w.h.p.~so that even if we output the entire set $V$ in a single update, the requirements in Theorem~\ref{Thm:DisconnectExpanderGraph} will still be satisfied.

Our data structure will maintain a subset $X$ of $V$ respecting $\mathcal C$. We require that vertices can never be added to $X$, only removed. With high probability, at all times, $G[X]$ is an expander graph for some later specified expansion factor and $H[X]$ is connected. We initialize $X = V$.

% such that after each update, For each edge deletion in $G$, our data structure outputs a subset of $V$ as required by Theorem~\ref{Thm:DisconnectExpanderGraph}. We let $X$ denote the subset of $V$ not output in any previous update. We require that at all times, $X$ respects $\mathcal C$ and we let $\mathcal C(X)$ be the subset of clusters of $\mathcal C$ contained in $X$. In the preprocessing step, we thus set $X = V$ and $\mathcal C(X) = \mathcal C$.

Finally, we do some additional preprocessing for a procedure called \texttt{XPrune}. Its purpose will be to ``prune'' $X$ in each update by removing some clusters from this set such that w.h.p., $G[X]$ remains an expander graph. We describe this procedure in detail in Section~\ref{sec:XPrune}; in this section, we shall regard it as a black box. When an edge $e$ is deleted from $G$, \texttt{XPrune}$(e)$ will update $X$ and output the clusters of $\mathcal C$ that are removed from $X$. For some later specified value $\gamma'\le\gamma$, we require the following two properties to hold w.h.p.~when \texttt{XPrune}$(e)$ returns:
\begin{enumerate}
\item for each $\mathcal C$-respecting cut $(K,X-K)$ of $X$, the number of edges of $G[X]$ crossing $(K,X-K)$ is at least $\gamma'\min\{|K|,|X-K|\}$, and
\item the total size of all clusters output by \texttt{XPrune} over all updates is $O(\Delta/\gamma)$.
\end{enumerate}
For now, we require that $\gamma' = \Theta_{n^{\eps}}(1)$.

\texttt{XPrune} will have access to $G$ and $\mathcal C$ but not to $E'$ so the updates to $X$ will be independent of the random bits used to form $E'$.

\subsection{Updates}
Given the above preprocessing, we now describe how to handle updates. The following invariants will be maintained:
\begin{Inv}\label{Inv:Clustering}
For each $C\in\mathcal C$, $G[C]$ contains a spanning tree $T(C)$, $|C|\leq 3\kappa$, and either $|C|\geq\kappa$ or no edge of $G[X]$ leaves $C$.
\end{Inv}
%\begin{Inv}\label{Inv:ClusterTreesInForests}
%$H$ consists of $T(C)$ for each $C\in\mathcal C$ and a subset of $E(G[X])$ obtained by %sampling each edge independently with probability $p$.
%\end{Inv}
\begin{Inv}\label{Inv:FH}
$F(H)$ is a spanning forest of $H$ such that for each $C\in\mathcal C$, $F(H)$ contains $T(C)$.
\end{Inv}
%In fact, the edges of $H$ not belonging to any spanning tree $T(C)$ of a cluster $C$ will be the remaining edges of $E'$ that belong to $G[X]$. Since $X$ is updated by \texttt{XPrune} and since these updates are independent of the random bits used to form $H$, we thus ensure that the first part of Invariant~\ref{Inv:ClusterTreesInForests} is maintained.

\subsubsection{Maintaining clusters}\label{subsec:UpdateClusters}
Now we describe how to handle an update to $G$ consisting of the deletion of an edge $e = (u,v)$ such that the invariants above are maintained. For now, we only focus on maintaining clusters and ignore updates to $X$. The approach we use is similar to how regions are maintained in Section~\ref{sec:FewNonTreeEdges}. We first delete $e$ from $G$, $H$, and $\FFE{H}$. We then consider two cases:

\paragraph{Edge $e$ is an intra-cluster edge:}\label{subsubsec:IntraClusterEdge}
In this case, $e\in G[C]$ for some cluster $C\in\mathcal C$. If $e\notin T(C)$, no further updates are needed and the invariants are maintained, so consider the case when $e\in T(C)$. The deletion of $e$ splits $T(C)$ into two subtrees $T(C_u)\ni u$ and $T(C_v)\ni v$ spanning subsets $C_u$ and $C_v$ of $C$, respectively. We visit the edges of $G[X]$ incident to $T(C_u)$ to look for a replacement edge for $T(C)$. If such an edge $f$ is found, $C$ remains a cluster, $T(C)$ is updated to $T(C_u)\cup T(C_v)\cup\{f\}$, and $\FFE{H}$ is updated by adding $f$ as a tree edge to $F(H)$.

The remaining case is when no replacement edge for $T(C)$ was found among the edges in $G[X]$. First, $C$ is removed from $\mathcal C$. Next, $C_u$ and $C_v$ are updated; we only describe the update for $C_u$ as $C_v$ is handled similarly. If $|C_u|\geq\kappa$, $C_u$ becomes a new cluster and is added to $\mathcal C$ with spanning tree $T(C_u)$. Otherwise, we look for an edge of $G[X]$ leaving $C_u$. If no such edge is found, $C_u$ is added to $\mathcal C$ with spanning tree $T(C_u)$. Otherwise, let $e' = (u',v')$ be the lexicographically smallest such edge\footnote{This particular choice of $e'$ is not important and is mainly made to emphasize that at any step, the clusters of $\mathcal C$ do not depend on the random bits used to form $E'$.} where $u'\in C_u$ and $v'$ belongs to some other cluster $C'$. We form $C'' = C_u\cup C'$ and let $T(C'')$ be the spanning tree $T(C_u)\cup T(C')\cup\{e'\}$ of $C''$. Note that $|C''|\geq\kappa$. If also $|C''|\leq 3\kappa$, we add $C''$ to $\mathcal C$. Otherwise, we apply Frederickson's \texttt{FINDCLUSTERS} to $T(C'')$ to partition $C''$ into $O(1)$ sub-clusters each inducing a subtree of $T(C'')$ and each of size between $\kappa$ and $3\kappa$. We replace $C_u$ and $C'$ with these sub-clusters in $\mathcal C$.

It is easy to see that the above satisfies Invariant~\ref{Inv:Clustering}. To satisfy Invariant~\ref{Inv:FH}, we do as follows. If a replacement edge $f$ was found for $T(C)$ in the above procedure, we add it to $F(H)$. Otherwise, if $C''$ could be formed when processing $C_u$ above, we add $e'$ to $F(H)$. This may create a cycle in $F(H)$ in which case we delete an inter-cluster edge incident to $C''$ belonging to this cycle. A similar update is done when processing $C_v$.

At this point, it may happen that $F(H)$ is no longer a spanning forest of $H[X]$. We apply \texttt{connect}$(u,v)$ to $\FFE{H}$ and if a reconnecting edge is found, it is added to $F(H)$.

\paragraph{Edge $e$ is an inter-cluster edge:}
This case is handled in the same way as above except that $\mathcal C$ remains unchanged.

It is easy to see that the above satisfies the invariants. Note that since the procedure above only looks for reconnecting edges in $G[X]$, a cluster not in $\mathcal C(X)$ can never be merged with another cluster, it can only be split into smaller clusters and these will never intersect $X$. In particular, vertices will never be added to $X$, satisfying our requirement above.

\subsubsection{Updating $X$}
We now present the entire data structure for handling updates which, in addition to maintaining clusters, also supports updates to $X$ with the procedure \texttt{XPrune}. At all times, $X$ respects $\mathcal C$ and this procedure implicitly maintains $X$ by maintaining the set $\mathcal C(X)$ of clusters of $\mathcal C$ contained in $X$.

The data structure maintains a subset $\mathcal C'$ of $\mathcal C$ which is initialized to be empty during preprocessing. This set can be regarded as a buffer of clusters whose vertex sets are waiting to be output in subsets $V_k$ in Theorem~\ref{Thm:DisconnectExpanderGraph}. At all times, $\mathcal C'\subseteq\mathcal C - \mathcal C(X)$.

Now consider an update consisting of the deletion of an edge $e = (u,v)$ from $G$. We split the update into two phases where the first phase takes place prior to $e$ being deleted and the second phase starts with the deletion of $e$.

\paragraph{Phase $1$:}
We check if $\mathcal C'$ is empty. If not, we continue with Phase $2$. Otherwise, we first remove from $H$ and $\FFE{H}$ every edge $(u',v')\in E'$ incident to a cluster of $\mathcal C - \mathcal C(X)$ and apply \texttt{connect}$(u',v')$ in $\FFE{H}$ to maintain Invariant~\ref{Inv:FH}. Then we check if $F(H)$ contains a tree spanning $X$. If not, we output $V$ and halt, skipping Phase $2$.

\paragraph{Phase $2$:}
At the beginning of Phase $2$, either $\mathcal C'$ is non-empty or $F(H)$ contains a tree spanning $X$. We first apply the procedure described above for updating clusters. If a cluster $C\in\mathcal C'$ is split into two sub-clusters in this procedure, they replace $C$ in $\mathcal C'$; note that clusters in $\mathcal C'$ can never be merged since $\mathcal C'\subseteq\mathcal C - \mathcal C(X)$. Let $T_u$ and $T_v$ be the trees of $F(H)$ containing $u$ and $v$, respectively, after this update. If $T_u\ne T_v$, we set $W_e$ to be the smaller of the two sets $V(T_u)$ and $V(T_v)$; otherwise, $W_e = \emptyset$. We then execute \texttt{XPrune}$(e)$ which updates $\mathcal C(X)$ and outputs a subset $\mathcal C_e$ of clusters; we update $\mathcal C'\leftarrow\mathcal C'\cup\mathcal C_e$. Next, we remove a subset $\mathcal C_e'$ of clusters from $\mathcal C'$ whose total vertex size is between $n^{1/2-4\eps}$ and $n^{1/2-4\eps}+3\kappa$; if the total size of clusters in $\mathcal C'$ is less than $n^{1/2-4\eps}$, we set $\mathcal C_e' = \mathcal C'$, thereby emptying $\mathcal C'$. Note that by Invariant~\ref{Inv:Clustering}, $C_e'$ is well-defined. Finally, we output $V_e = \cup_{C\in\mathcal C_e'}C\cup W_e$.

\subsubsection{Correctness}
We now show that the update procedure described above satisfies the requirements of Theorem~\ref{Thm:DisconnectExpanderGraph} except that we delay the bound on running time until later.
\begin{Lem}
For suitable sets $W_k$, the procedure above satisfies the first requirement of Theorem~\ref{Thm:DisconnectExpanderGraph}.
\end{Lem}
\begin{proof}
We may assume that $V$ is not output in Phase $1$ of any update since if such an update $k$ exists, we can simply pick $W_{k'} = V$ for every $k'\ge k$ and we only need to focus on updates $k' < k$. We shall denote the set $X$ resp.~$\mathcal C'$ at the beginning of an update $k$ by $X(k)$ resp.~$\mathcal C'(k)$ and denote the set $W_e$ formed in update $k$ by $W_e(k)$.

Now, consider an update $k$ and let $k'\le k$ be the latest update for which $\mathcal C'(k') = \emptyset$; note that $k'$ exists since $\mathcal C'$ is empty at the beginning of the first update. We show that $W_k = (V - X(k'))\cup\bigcup_{k'\le k''\le k}W_e(k'')$ satisfies the first requirement of Theorem~\ref{Thm:DisconnectExpanderGraph}. 

Since $\mathcal C'(k') = \emptyset$, every vertex of $V - X(k')$ has been output in updates prior to $k'$. Furthermore, $W_e(k'')$ is output in update $k''$ for $k'\le k''\le k$. Hence, $W_k$ is contained in the union of sets output during the $k$ first updates.

It remains to show that $G[V - W_k]$ is connected at the end of update $k$. At the end of Phase $1$ of update $k'$, $F(H)$ contains a tree spanning $X(k')$. Hence, at the end of update $k$, $F(H)$ contains a tree spanning $X(k') - \bigcup_{k'\le k''\le k}W_e(k'') = V - W_k$. Since $F(H)$ is contained in $G$, it follows that $G[V - W_k]$ is connected.
\end{proof}

We now consider the second requirement of Theorem~\ref{Thm:DisconnectExpanderGraph}. We delay the analysis of the running time until Section~\ref{subsec:DecExpImpl} below and show that w.h.p., each set output has size $O(n^{1/2 - 4\eps})$.

In the following, let $N$ be an integer which w.h.p.~is an upper bound on the maximum number of consecutive updates for which $\mathcal C'$ fails to be emptied. Since w.h.p., the total size of all clusters output by \texttt{XPrune} over all updates is $O(\Delta/\gamma)$ and since we output $\Omega(n^{1/2 - 4\eps})$ vertices in each update that does not empty $\mathcal C'$, we can pick $N = \Theta(\Delta/(\gamma n^{1/2 - 4\eps})) = \Theta_{n^{\eps}}(\Delta/\sqrt n)$.

Fix sampling probability $p = 8c_p(\ln n)/(\gamma'\kappa) = \Theta_{n^\eps}(1/\kappa)$ for a sufficiently large constant $c_p > 0$. We get the following lemma, showing that the data structure is unlikely to output the entire vertex set $V$ at the end of Phase $1$.
\begin{Lem}\label{Lem:SkeletonCertificate}
W.h.p., at the beginning of each update, $H[X]$ is connected.
\end{Lem}
\begin{proof}
We may assume that $|\mathcal C(X)| > 1$ since otherwise, $H[X]$ is connected as every cluster in $\mathcal C(X)$ is spanned by a tree belonging to $H[X]$. Given this assumption and since w.h.p., $G[X]$ is connected by the first property of \texttt{XPrune}, it follows from Invariant~\ref{Inv:Clustering} that w.h.p., each cluster in $\mathcal C(X)$ contains at least $\kappa$ vertices.

Assume in the following that $G$ initially is a $\gamma$-expander graph and that the first property of \texttt{XPrune} holds after each call to this procedure. We may make these assumptions since they hold with high probability.

Consider the beginning of some update. If it is the first update then since $G[X] = G$ is a $\gamma$-expander graph, we have in particular that for any $\mathcal C$-respecting cut $(K,X-K)$, the number of edges of $E(G[X])$ crossing $(K,X - K)$ is at least $\gamma\min\{|K|,|X - K|\}\ge\gamma'\min\{|K|,|X-K|\}$. If it is not the first update then since \texttt{XPrune} was executed at the end of the previous update, the number of edges crossing each such cut is at least $\gamma'\min\{|K|,|X-K|\}$.

Updates to $X$ are independent of the sampled edges of $H$ so for any $\mathcal C$-respecting cut $(K,X-K)$, the expected number of edges of $H$ crossing $(K,X - K)$ is at least $p\gamma'\min\{|K|,|X - K|\}$.

Consider a $\mathcal C$-respecting cut in $X$ where the smaller side contains $k$ clusters. The expected number of edges of $H$ crossing the cut is at least $p\gamma'\kappa k = 8c_pk\ln n$. By a Chernoff bound, the probability that the number of edges of $H$ crossing the cut is less than $4c_p k\ln n$ is at most $n^{-c_p k}$. The number of $\mathcal C$-respecting cuts of $X$ where the smaller side contains $k\geq 1$ clusters is less than $n^k$. A union bound over all such cuts and over all $k$ shows that with probability at least $1 - n^{2 - c_p}$, the number of edges of $H$ crossing any $\mathcal C$-respecting cut of $X$ is at least $4c_p\ln n$. In particular, $H[X]$ is connected with probability at least $1 - n^{2 - c_p}$. Picking $c_p$ sufficiently large shows the lemma.
\end{proof}
We pick $\kappa = n^{1/2 - \Theta(\eps)}$ sufficiently small such that $\kappa N = O(n^{1/2 - 4\eps})$.
\begin{Lem}\label{Lem:HPrimeSplitBound}
W.h.p., at the end of each update, every tree of $F(H)$ except one has size $O(n^{1/2 - 4\eps})$.% $O(\kappa M(\log^{5/2}n)/(\gamma\sqrt{n\lambda}\gamma'))$.
\end{Lem}
\begin{proof}
Assign numbers $1,2,\ldots$ to the updates in the order they occur. For $i\ge 1$, define $t_i$ such that in the beginning of update $t_i$, $\mathcal C'$ is empty and such that this happened exactly $i-1$ times in previous updates. Note that $t_1 < t_2 <\cdots$ and since $\mathcal C'$ is empty initially, we have $t_1 = 1$. We denote by $X_i$ the set $X$ at the start of update $t_i$.

Note that from the start of Phase $2$ of update $t_i$ until the end of the last update, all inter-cluster edges of $H$ are contained in $G[X_i]$ so by Invariant~\ref{Inv:Clustering}, all trees of $F(H)$ not in $G[X_i]$ have size $O(\kappa)$. Hence, from the end of update $t_i$ until the end of update $t_{i+1}-1$, we only need to show the lemma for trees of $F(H)$ contained in $X_i$. By Invariant~\ref{Inv:Clustering} and Lemma~\ref{Lem:SkeletonCertificate}, we may assume that at the beginning of update $t_i$, each cluster has size $\Theta(\kappa)$.

Consider an update $j$ and pick $i$ such that $t_i\le j < t_{i+1}$. With the same arguments as in the proof of Lemma~\ref{Lem:SkeletonCertificate}, it follows that at the start of update $t_i$, w.h.p., for every $\mathcal C$-respecting cut $(K,X_i-K)$, the number of edges of $G[X_i]$ crossing $(K,X_i-K)$ is at least $\gamma'\min\{|K|,|X_i-K|\}$. Let $\mathcal C_{t_i}$ resp.~$\mathcal C_{j+1}$ be the set $\mathcal C$ at the start of update $t_i$ resp.~$j+1$. Note that $\mathcal C_{j+1}$ is also the set $\mathcal C$ at the end of update $j$.

Now, consider the end of update $j$ and let $(K,X_i-K)$ be a $\mathcal C_{j+1}$-respecting cut where the smaller side contains $k$ clusters. We have $j - t_i + 1 \le t_{i+1} - t_i \le N$. Hence, only $O(N)$ clusters of $\mathcal C_{t_i}$ intersect both sides of $(K,X_i-K)$ since each update changes only $O(1)$ clusters. Let $K'\subseteq K$ be the union of clusters of $\mathcal C_{t_i}\cap\mathcal C_{j+1}$ contained in $K$. Note that $K'$ contains $k - O(N)$ clusters of $\mathcal C_{t_i}$. By the first property of \texttt{XPrune}, the number of edges of $G[X_i]$ that crossed $(K',X_i - K')$ at the start of update $t_i$ was $\gamma'(k - O(N))\Theta(\kappa)$. The number of such edges which have one endpoint in $K'$ and one endpoint in $K$ is $O(|K - K'|) = O(N\kappa)$ so the number of edges of $G[X_i]$ that crossed $(K,X_i - K)$ at the start of update $t_i$ was $\gamma'(k - O(N))\Theta(\kappa)$. Since no more than $N$ edges have been deleted since then, there are  $\gamma'(k - O(N))\Theta(\kappa)$ edges of $G[X_i]$ and hence $\Theta(\ln n)(k - O(N))$ expected number of edges of $H[X_i]$ crossing $(K,X_i-K)$ at the end of update $j$.

Using Chernoff bounds as in the proof of Lemma~\ref{Lem:SkeletonCertificate}, it follows that at the end of update $j$, w.h.p., for every $\mathcal C_{j+1}$-respecting cut where the smaller side contains $\Omega(N)$ clusters, there is at least one edge of $H$ crossing this cut.

If at the end of update $j$ there were two trees in $F(H)$ of size $\omega(\kappa N)$, there would be a $\mathcal C_{j+1}$-respecting cut where the smaller side contains $\omega(N)$ clusters and where no edge of $H$ crosses this cut which by the above only occurs with low probability. The lemma now follows since $\kappa N = O(n^{1/2 - 4\eps})$.
\end{proof}
We can now show that the size bound in the second requirement of Theorem~\ref{Thm:DisconnectExpanderGraph} holds.
\begin{Lem}
W.h.p., for each update, the set output in Phase $2$ has size $O(n^{1/2 - 4\eps})$.
\end{Lem}
\begin{proof}
By Lemma~\ref{Lem:SkeletonCertificate}, w.h.p., sets are only output in Phase $2$. Consider an execution of this phase when an edge $e$ is deleted. By Invariant~\ref{Inv:Clustering}, the size of each subset $\cup_{C\in\mathcal C_e'}C$ is $O(n^{1/2 - 4\eps})$. Note that when $W_e$ is formed, $F(H)$ does not change for the rest of the update. Hence, at the end of the update, if $W_e$ is not empty, it must be the vertex set of some tree of $F(H)$ and because of the way we choose $W_e$, this cannot be the tree with the most vertices. Lemma~\ref{Lem:HPrimeSplitBound} then implies that $|W_e| = O(n^{1/2 - 4\eps})$.
\end{proof}

\subsection{Implementation and performance}\label{subsec:DecExpImpl}
We now give the implementation details for the data structure of this section and analyze its preprocessing and update time. The implementation and analysis of the performance of $\texttt{XPrune}$ is delayed until Section~\ref{sec:XPrune}.

By Lemma~\ref{Lem:FewNonTreeEdges}, the preprocessing can be done in $\tilde O(n)$ worst-case time. We shall maintain $\mathcal C'$ as a linked list so that each insertion/deletion of a cluster in this list takes $O(1)$ time. In the following, we focus on an update consisting of the deletion of an edge $e$.

\paragraph{Phase $1$:}
Observe that at all times, the edges of $H$ that do not belong to $F(H)$ must all belong to sampled set $E'$ and w.h.p., $|E'| = O(np) = O_{n^\eps}(n/\kappa)$. Thus, by Lemma~\ref{Lem:FewNonTreeEdges}, w.h.p.~each update to $\FFE{H}$ can be done in $O_{n^{\eps}}(\sqrt{n/\kappa}) = O_{n^\eps}(n^{1/4})$ worst-case time.

By the second property of \texttt{XPrune}, w.h.p., the total number of vertices in clusters of $\mathcal C - \mathcal C(X)$ is $O(\Delta/\gamma)$. Since updates to $X$ are independent of $E'$, w.h.p., the expected number of edges of $E'$ incident to these clusters is $O(p\Delta/\gamma)$. By a Chernoff bound, w.h.p.~the actual number of such edges is $\tilde O(p\Delta/\gamma)$. For each cluster $C\in\mathcal C$, we shall maintain a linked list of the edges of $E'$ incident to $C$; this can easily be done in $O(\kappa)$ time per update since a single update only affects $O(1)$ clusters.

By the first property of \texttt{XPrune}, w.h.p., in every update, each cluster of $\mathcal C(X)$ has size $\Theta(\kappa)$. Hence, w.h.p., for any execution of Phase $1$ where $\mathcal C'$ is empty, the number of clusters of $\mathcal C - \mathcal C(X)$ that have not been processed in a previous such execution is $O(\Delta/(\gamma\kappa) + N)$. If we use the edge-lists associated with clusters, we can identify the edges of $E'$ incident to clusters of $\mathcal C - \mathcal C(X)$ in worst-case time $\tilde O(\Delta/(\gamma\kappa) + N + p\Delta/\gamma) = O_{n^\eps}(\Delta/\sqrt n)$ with high probability. By the above, w.h.p.~the total worst-case time for updating $F(H)$ is $\tilde O((p\Delta/\gamma)\sqrt{np}) = O_{n^\eps}(\Delta/n^{1/4})$.

In order to detect if $F(H)$ contains a tree spanning $X$, we shall maintain $|\mathcal C - \mathcal C(X)|$ as well as the number of trees in $F(H)$. This can easily be done within the above time bounds. We observe that after the update of $F(H)$, the number of trees in $F(H)$ is equal to $|\mathcal C - \mathcal C(X)| + 1$ iff $F(H)$ contains a tree spanning $X$. Hence, detecting whether the latter holds takes constant time. By Lemma~\ref{Lem:SkeletonCertificate}, we can afford to spend linear time to output $V$ since this case occurs with low probability.

Combining all of the above, it follows that w.h.p., Phase $1$ can be executed in $O_{n^{\eps}}(\Delta/n^{1/4})$ worst-case time.

\paragraph{Phase $2$:}
Each execution of the procedure in Section~\ref{subsec:UpdateClusters} can easily be done in $O(\kappa)$ time plus the time to execute an operation in $F(H)$ where an operation is either detecting a cycle when inserting an edge in $F(H)$ or the operation \texttt{connect} in $\FFE{H}$. The latter takes $\tilde O(\sqrt{np}) = O_{n^\eps}(n^{1/4})$ worst-case time. With the notation in Section~\ref{subsec:UpdateClusters}, if adding $e'$ creates a cycle in $F(H)$, we can identify the inter-cluster edge on the cycle incident to $C''$ in $O(\log n)$ time by maintaining a top tree for $F(H)$ which supports the operation of finding the first inter-cluster edge on a path between two query vertices; inter-cluster edges of $F(H)$ are marked in the top tree and finding the nearest marked node is an operation that such a data structure supports. A top tree can also be used to maintain the vertex size of each tree in $F(H)$ in $O(\log n)$ time. By Lemma~\ref{Lem:HPrimeSplitBound}, w.h.p.~the set $W_e$ can thus be formed in $O(n^{1/2 - 4\eps})$ time.

We can easily maintain the size of each cluster within the time bounds above so extracting set $\mathcal C_e'$ from $\mathcal C'$ can be done in $O(n^{1/2 - 4\eps} + \kappa) = O(n^{1/2-4\eps})$ worst-case time. It now follows that w.h.p., Phase $2$ can be executed in $O(n^{1/2 - 4\eps}) + O_{n^{\eps}}(n^{1/4})$ worst-case time, excluding the time for \texttt{XPrune}. Total time for both phases is thus $O(n^{1/2 - 4\eps}) + O_{n^{\eps}}(\Delta/n^{1/4})$ which is within the time bound of Theorem~\ref{Thm:DisconnectExpanderGraph}.

%\[
%  O(\sqrt{n\lambda}\log^{3/2}n + \kappa M(\log^{5/2} n)/(\gamma\sqrt{n\lambda}\gamma') + (M\sqrt n\log^{9/2}n)/(\gamma(\gamma')^{3/2}\kappa^{3/2})).
%\]

\section{Low-conductance Cuts and Sparsification}\label{sec:LowCondCutsSparsification}
In this section, we show Corollary~\ref{Cor:LowConductance} which will be needed in the next section. It shows a result somewhat similar to Karger~\cite{Karger} but for conductance instead of cut values. The corollary is a bit technical but it roughly implies that in order to find low-conductance cuts in a graph, it suffices to look for them in a sparse sampled representative of this graph. First we need the following lemma.
\begin{Lem}\label{Lem:LowConductance}
Given $c > 0$, $\kappa\geq 1$, and $\rho\leq 1$, let $G_{\mathcal C} = (\mathcal C, E_{\mathcal C})$ be an $n$-vertex multigraph with a finite number of edges and degree at least $\kappa\rho$. Let $G_{\mathcal C}' = (\mathcal C, E_{\mathcal C}')$ be the multigraph obtained from $G_{\mathcal C}$ by sampling each edge independently with probability $p = \min\{1, (12c + 24)(1/(\rho^2\kappa))\ln n\}$. Then with probability $1 - O(1/n^c)$, for every cut $(S,\mathcal C - S)$ in $\mathcal C$,
\begin{enumerate}
\item if $\Phi_{G_{\mathcal C}}(S)\geq \rho$ then $\Phi_{G_{\mathcal C}'}(S)$ deviates from $\Phi_{G_{\mathcal C}}(S)$ by a factor of at most $4$, and
\item if $\Phi_{G_{\mathcal C}}(S) < \rho$ then $\Phi_{G_{\mathcal C}'}(S) \leq 6\rho$.
\end{enumerate}
\end{Lem}
\begin{proof}
We may assume that $p < 1$. Let $c' = 12c + 24$ so that $p = c'(1/(\rho^2\kappa))\ln n$. Let positive integer $k\leq \lfloor n/2\rfloor$ be given. Consider a cut $(S,\mathcal C - S)$ in $\mathcal C$ where the smaller side has size $k$. Let $S_v\in\{S,\mathcal C - S\}$ resp.~$S_v'\in\{S,\mathcal C - S\}$ be a side of the cut with minimum volume in $G_{\mathcal C}$ resp.~$G_{\mathcal C}'$. Let $\mu_{\delta} = E[\delta_{G_{\mathcal C}'}(S_v)] = p\delta_{G_{\mathcal C}}(S_v)$, $\mu_v = E[\mbox{Vol}_{G_{\mathcal C}'}(S_v)] = p\mbox{Vol}_{G_{\mathcal C}}(S_v)$, and $\mu_v' = E[\mbox{Vol}_{G_{\mathcal C}'}(S_v')] = p\mbox{Vol}_{G_{\mathcal C}}(S_v')$. By the degree lower bound,
\[
  \mu_v \geq pk\kappa\rho = (c'k/\rho)\ln n.
\]

Assume first that $\Phi_{G_{\mathcal C}}(S)\geq \rho$. Since $\mu_v \geq c'k\ln n$, a Chernoff bound implies that the probability that $\mbox{Vol}_{G_{\mathcal C}'}(S_v)$ deviates by at most a factor of $2$ from $\mu_v$ is at least $1 - 2e^{-\mu_v/{12}}\geq 1 - 2n^{-c'k/12}$. Similarly, $\mbox{Vol}_{G_{\mathcal C}'}(S_v')$ deviates by at most a factor of $2$ from $\mu_v'$ with probability at least $1 - 2n^{-c'k/12}$.

We have
\[
  \mu_{\delta}  = p(\mbox{Vol}_{G_{\mathcal C}}(S_v)\Phi_{G_{\mathcal C}}(S_v)) = \mu_v\Phi_{G_{\mathcal C}}(S_v) = \mu_v\Phi_{G_{\mathcal C}}(S)\geq c'k\ln n.
\]
By a Chernoff bound, the probability that $\delta_{G_{\mathcal C}'}(S_v)$ deviates from $\mu_{\delta}$ by at most a factor of $2$ is at least $1 - 2n^{-c'k/12}$. A union bound then implies that all three Chernoff bounds hold with probability at least $1 - 6n^{-c'k/12}$, in which case
\begin{align*}
  \frac 1 4 \Phi_{G_{\mathcal C}'}(S) & = \frac{\delta_{G_{\mathcal C}'}(S_v)}{4\mbox{Vol}_{G_{\mathcal C}'}(S_v')}
                      \leq \frac{2p\delta_{G_{\mathcal C}}(S_v)}{4\frac p 2 \mbox{Vol}_{G_{\mathcal C}}(S_v')}
                      \leq \frac{2p\delta_{G_{\mathcal C}}(S_v)}{4\frac p 2 \mbox{Vol}_{G_{\mathcal C}}(S_v)}
                      = \Phi_{G_{\mathcal C}}(S)
                      \leq \frac{\frac 2 p\cdot\delta_{G_{\mathcal C}'}(S_v)}{\frac 1{2p} \mbox{Vol}_{G_{\mathcal C}'}(S_v)}\\
                      & \leq \frac{4\delta_{G_{\mathcal C}'}(S_v')}{\mbox{Vol}_{G_{\mathcal C}'}(S_v')}
                      = 4\Phi_{G_{\mathcal C}'}(S).
\end{align*}
Hence, if $\Phi_{G_{\mathcal C}}(S) \geq \rho$ then the first condition of the lemma holds for cut $(S, \mathcal C - S)$ with probability at least $1 - 6n^{-c'k/12}$.

Now assume that $\Phi_{G_{\mathcal C}}(S) < \rho$. Using the observations above and the fact that $\mu_v\leq \mu_v'$, we can bound the probability that $\Phi_{G_{\mathcal C}'}(S)$ is greater than $6\rho$ by
\begin{align*}
\Pr(\Phi_{G_{\mathcal C}'}(S) > 6\rho) & = \Pr(\delta_{G_{\mathcal C}'}(S_v) > 6\mbox{Vol}_{G_{\mathcal C}'}(S_v')\rho)\\
                                 & = \Pr(\delta_{G_{\mathcal C}'}(S_v) > 6\mbox{Vol}_{G_{\mathcal C}'}(S_v')\rho\land \mu_v'\leq 2\mbox{Vol}_{G_{\mathcal C}'}(S_v')) +{}\\
                                 & \phantom{{} = {}} \Pr(\delta_{G_{\mathcal C}'}(S_v) > 6\mbox{Vol}_{G_{\mathcal C}'}(S_v')\rho\land \mu_v' > 2\mbox{Vol}_{G_{\mathcal C}'}(S_v'))\\
                                 & \leq \Pr(\delta_{G_{\mathcal C}'}(S_v) > 6\mbox{Vol}_{G_{\mathcal C}'}(S_v')\rho\land \mu_v\leq 2\mbox{Vol}_{G_{\mathcal C}'}(S_v')) + \Pr(\mu_v' > 2\mbox{Vol}_{G_{\mathcal C}'}(S_v'))\\
                                 & \leq \Pr(\delta_{G_{\mathcal C}'}(S_v) > 3\mu_v\rho\land \mu_v\leq 2\mbox{Vol}_{G_{\mathcal C}'}(S_v')) + 2n^{-c'k/12}\\
                                 & \leq \Pr(\delta_{G_{\mathcal C}'}(S_v) > 3\mu_v\rho) + 2n^{-c'k/12}.
\end{align*}

We will use a Chernoff bound to show that $\Pr(\delta_{G_{\mathcal C}'}(S_v) > 3\mu_v\rho) = O(n^{-c'k/12})$. Pick real number $\delta$ such that $3\mu_v\rho = (1+\delta)\mu_{\delta}$. Since $\mu_{\delta}/\mu_v = \Phi_{G_{\mathcal C}}(S) < \rho$, we have $1 + \delta > 3$ and hence $\delta > 2$. Furthermore, it follows from the above that $(1+\delta)\mu_{\delta} = 3\mu_v\rho\geq 3c'k\ln n$ and a Chernoff bound now shows that
\begin{align*}
\Pr(\delta_{G_{\mathcal C}'}(S) > 3\mu_v\rho) & = \Pr(\delta_{G_{\mathcal C}'}(S) > (1+\delta)\mu_{\delta})
                                          < \left(\frac{e^{\delta}}{(1+\delta)^{1+\delta}}\right)^{\mu_{\delta}}
                                          < \left(\frac{e}{3}\right)^{(1+\delta)\mu_{\delta}}
                                          \leq \left(\frac{e}{3}\right)^{3c'k\ln n}\\
                                        & < n^{-c'k/12},
\end{align*}
as desired. We conclude that $\Pr(\Phi_{G_C'}(S) > 6\rho)\leq 3n^{-c'k/12}$.

Combining all of the above, it follows that with probability at least $1 - 6n^{-c'k/12}$, $(S,\mathcal C - S)$ satisfies the first condition of the lemma when $\Phi_{G_{\mathcal C}}(S)\geq \rho$ and the second condition when $\Phi_{G_{\mathcal C}}(S) < \rho$. The number of cuts of $\mathcal C$ where the smaller side has size $k$ is at most $n^k$. By a union bound, the probability that the conditions of the lemma hold for all such cuts is at least $1 - 6n^{k - c'k/12} = 1 - 6n^{k(1 - (12c+24)/12)} = 1 - 6n^{k(-c-1)}\geq 1 - 6n^{-c-1}$. The lemma now follows by a union bound over all $k\leq\lfloor n/2\rfloor$.
\end{proof}

\begin{Cor}\label{Cor:LowConductance}
Let $G = (V,E)$ be an $n$-vertex graph of max degree $d$, let $c > 0$ be a constant and let $\kappa\geq 1$, and $\rho\leq 1$ be given. Let $\mathcal C$ be a clustering of $V$ such that for each $C\in\mathcal C$, $\kappa\leq |C| \leq 3\kappa$ and $G[C]$ is connected. Let $G_{\mathcal C}$ be the multigraph $(\mathcal C, E_{\mathcal C})$ where $E_{\mathcal C}$ is the set of edges of $E$ between distinct clusters of $\mathcal C$ and assume that $G_{\mathcal C}$ has min degree at least $\kappa\rho$. Let $G_{\mathcal C}' = (\mathcal C, E_{\mathcal C}')$ be the multigraph obtained from $G_{\mathcal C}$ by sampling each edge independently with probability $p = \min\{1, (12c + 24)(1/(\rho^2\kappa))\ln n\}$. Then with probability $1 - O(1/|\mathcal C|^c)$, the following holds for every cut $(S,\mathcal C - S)$ in $\mathcal C$:
\begin{enumerate}
\item if $\Phi_{G_{\mathcal C}'}(S) \leq \rho'$ then $\Phi_G(\cup_{C\in S}C) \leq 4\rho'$ for any $\rho'\geq \rho/4$, and
\item if $\Phi_{G_{\mathcal C}'}(S) > 6\rho$ then $\Phi_G(\cup_{C\in S}C) \geq \rho^3/(9d^3)$.
\end{enumerate}
\end{Cor}
\begin{proof}
Let $\rho'\geq \rho/4$ be given. By Lemma~\ref{Lem:LowConductance}, with probability $1 - O(1/|\mathcal C|^c)$, for every cut $(S,\mathcal C - S)$ in $\mathcal C$, if $\Phi_{G_\mathcal C'}(S) \leq \rho'$ then $\Phi_{G_{\mathcal C}}(S) \leq 4\rho'$ and if $\Phi_{G_{\mathcal C}'}(S) > 6\rho$ then $\Phi_{G_{\mathcal C}}(S)\geq \rho$. Assume that this property holds in the following.

Let $(S,\mathcal C - S)$ be a cut in $\mathcal C$ such that $\Phi_{G_{\mathcal C}'}(S) \leq \rho'$ and let $S_G = \cup_{C\in S}C$. Assume w.l.o.g.~that $\mbox{Vol}_G(S_G)\leq\mbox{Vol}_G(V - S_G)$. The first part of the corollary follows from
\[
  \Phi_G(S_G) = \frac{\delta_G(S_G)}{\mbox{Vol}_G(S_G)}
               = \frac{\delta_{G_{\mathcal C}}(S)}{\mbox{Vol}_G(S_G)}
               \leq \frac{\delta_{G_{\mathcal C}}(S)}{\mbox{Vol}_{G_{\mathcal C}}(S)}
               \leq \Phi_{G_{\mathcal C}}(S)
               \leq 4\rho'.
\]

For the second part, assume instead that $\Phi_{G_{\mathcal C}'}(S) > 6\rho$. Since each cluster is incident to no more than $3d\kappa$ edges of $E$ and since it has degree at least $\kappa\rho$ when viewed as a vertex in $G_{\mathcal C}$, we get
\[
  \mbox{Vol}_G(S_G)\leq |S|\cdot 3d\kappa\leq\frac{\mbox{Vol}_{G_{\mathcal C}}(S)}{\kappa\rho}\cdot 3d\kappa = \frac{3d\mbox{Vol}_{G_{\mathcal C}}(S)}{\rho},
\]
and similarly $\mbox{Vol}_G(V - S_G)\le 3d\mbox{Vol}_{G_{\mathcal C}}(\mathcal C - S)/\rho$. Since $G[C]$ is connected for each $C\in\mathcal C$, each vertex of $G$ has degree at least $1$ so
\[
  \mbox{Vol}_{G_\mathcal C}(S)\le d|S_G|\le d\mbox{Vol}_G(S_G)\le d\mbox{Vol}_G(V - S_G)\le \frac{3d^2\mbox{Vol}_{G_{\mathcal C}}(\mathcal C - S)}{\rho}.
\]
It follows that $\min\{\mbox{Vol}_{G_\mathcal C}(S),\mbox{Vol}_{G_{\mathcal C}}(\mathcal C - S)\}\ge \rho\mbox{Vol}_{G_{\mathcal C}}(S)/(3d^2)$ and hence,
\[
  \Phi_G(S_G) = \frac{\delta_{G_{\mathcal C}}(S)}{\mbox{Vol}_G(S_G)}
              \geq\frac{\rho\delta_{G_{\mathcal C}}(S)}{3d\mbox{Vol}_{G_{\mathcal C}}(S)}
              \geq \frac{\rho^2\Phi_{G_{\mathcal C}}(S)}{9d^3}
              \geq \frac{\rho^3}{9d^3},
\]
as desired.
\end{proof}

\section{The \texttt{XPrune} Procedure}\label{sec:XPrune}
In this section, we present the procedure \texttt{XPrune} which we used as a black box in Section~\ref{sec:DecDSLowCondCuts}. It makes use of a new dynamic version of the procedure \texttt{Nibble} of Spielman and Teng~\cite{SpielmanTeng} so before moving on, we will introduce some notation used in their paper as well as the procedure \texttt{Nibble}.

When we refer to vectors in the following, we assume that each of them has an entry for each vertex in a graph that should be clear from context. We denote by $d(S)$ the sum of degrees of vertices in a subset $S$ and we write $d(v)$ instead of $d(\{v\})$ for a vertex $v$. Let $A$ be the adjacency matrix for the graph, let $D$ be the diagonal matrix where entry $(i,i)$ is the degree of the $i$th vertex, and let $I$ be the identity matrix of the same dimensions as $A$ and $D$. We define the matrix $P$ by $P = (AD^{-1} + I)/2$.

For a graph $H$ and for a vertex $s\in V(H)$, let $\chi_s$ be the vector with an entry for each vertex in $V(H)$ where $\chi_s(s) = 1$ and $\chi_s(v) = 0$ for all $v\neq s$. For a vector $p$ and for $\varepsilon > 0$, define the truncation operation $[p]_\varepsilon$ by
\[
  [p]_\varepsilon(v) = \left\{\begin{array}{ll} p(v) & \mbox{if } p(v)\geq 2\varepsilon d(v),\\
                                           0    & \mbox{otherwise}.
                         \end{array}\right.
\]

\subsection{The \texttt{Nibble} procedure}
Pseudocode for \texttt{Nibble}$(H,s,\theta,b)$ can be seen in Figure~\ref{fig:Nibble}.
\begin{figure}
\begin{tabbing}
d\=dd\=\quad\=\quad\=\quad\=\quad\=\quad\=\quad\=\quad\=\quad\=\quad\=\quad\=\quad\=\kill
\>\texttt{Procedure} \texttt{Nibble}$(H,s,\theta,b)$\\\\
\>1. \>\>set $\tilde p_0 \leftarrow \chi_s$\\
\>2. \>\>set $t_0 \leftarrow 49\ln(|E(H)|e^4)/\theta^2$ and $\epsilon_b \leftarrow \theta/(56\ln(|E(H)|e^4)t_02^b$)\\
\>3. \>\>for $t\leftarrow 1$ to $t_0$\\
\>4. \>\>\>set $\tilde p_t \leftarrow [P\tilde p_{t-1}]_{\epsilon_b}$\\
\>5. \>\>\>compute a permutation $\tilde\pi_t$ of $V(H)$ such that for all $i$, $\tilde p_t(\tilde\pi_t(i))\geq\tilde p_t(\tilde\pi_t(i+1))$\\
\>6. \>\>\>if there exists a $\tilde j$ such that\\
\>7. \>\>\>\>\textbullet\> $\Phi_{V(H)}(\tilde\pi_t(\{1,\ldots,\tilde j\}))\leq\theta$,\\
\>8. \>\>\>\>\textbullet\> $\tilde p_t(\tilde\pi_t(\tilde j))\geq 5\theta/(392(\ln(|E(H)|e^4))\mbox{Vol}_{V(H)}(\tilde\pi_t(\{1,\ldots,\tilde j\})))$, and\\
\>9. \>\>\>\>\textbullet\> $\frac 5 6\mbox{Vol}_{V(H)}\geq\mbox{Vol}_{V(H)}(\tilde\pi_t(\{1,\ldots,\tilde j\}))\geq\frac 5 7 2^{b-1}$,\\
\>10.\>\>\>then output $C \leftarrow \tilde\pi_t(\{1,\ldots,\tilde j\})$ and halt\\
\>11.\>\>return \emph{failed}
\end{tabbing}
\caption{Pseudocode for procedure \texttt{Nibble}. It is assumed that vertices in $V(H)$ are indexed from $1$ to $|V(H)|$.}\label{fig:Nibble}
\end{figure}
It calculates truncated probability distributions for $t_0$ steps of a random walk in $H$ starting in vertex $s$ where in each step, the walk stays in the current vertex with probability $1/2$ and otherwise goes to one of the adjacent vertices with equal probability. It then derives from one of these truncated probability distributions a low-conductance cut, assuming a suitable starting vertex $s$ is chosen.

In this section, we define $\theta_+ = \theta^3/(14^4\ln^2(3ne^4))$ and we shall implicitly assume that each graph contains at most $3n$ edges, as is the case for $G$. Spielman and Teng~\cite{SpielmanTeng} showed the following property of \texttt{Nibble}.
\begin{Lem}[\cite{SpielmanTeng}]\label{Lem:Nibble}
Let $H$ be a graph. For each $\theta\leq 1$ and for each $S\subseteq V(H)$ satisfying
\[
  \mbox{Vol}_{V(H)}(S)\leq\frac 2 3\mbox{Vol}_{V(H)}(V(H))\mbox{ and }
  \Phi_{V(H)}(S)\leq 2\theta_+,
\]
there is a subset $S^g\subseteq S$ such that $\mbox{Vol}_{V(H)}(S^g)\geq\mbox{Vol}_{V(H)}(S)/2$ and this subset can be decomposed into sets $S_b^g$ for $b = 1,\ldots,\lceil\lg(E(H))\rceil$ such that for each $b$ and any $s\in S_b^g$, \texttt{Nibble}$(H,s,\theta,b)$ outputs a vertex set $C$ such that
\begin{enumerate}
\item $\Phi_{V(H)}(C)\leq\theta$,
\item $\frac 4 7 2^{b-1}\leq\mbox{Vol}_{V(H)}(C\cap S)$, and
\item $\mbox{Vol}_{V(H)}(C)\leq\frac 5 6 \mbox{Vol}_{V(H)}(V(H))$.
\end{enumerate}
For all $b$, \texttt{Nibble} can be implemented to run in worst-case time $O(2^b\ln^4(E(H))/\theta^5)$.
\end{Lem}
We will not need the full strength of this result but only the following simpler corollary.
\begin{Cor}\label{Cor:SimpleNibble}
Let $H$ be a graph of max degree at most $d_{\max}$. For each $\theta\leq 1$ and for each $S\subseteq V(H)$ with $\Phi_{V(H)}(S)\leq 2\theta_+$, there is an $s\in V(H)$ and an integer $b\in\{1,\ldots,\lceil\lg(E(H))\rceil\}$ such that in worst-case time $\tilde O(2^b/\theta^5)$, \texttt{Nibble}$(H,s,\theta,b)$ outputs a vertex set $C$ of size between $\Omega(2^b/d_{\max})$ and $\tilde O(2^b/\theta^3)$ with $\Phi_{V(H)}(C)\leq\theta$. Furthermore, for any $s$ and $b$, if \texttt{Nibble}$(H,s,\theta,b)$ outputs a set, this set has size between $\Omega(2^b/d_{\max})$ and $\tilde O(2^b/\theta^3)$ and has conductance at most $\theta$ in $H$.
\end{Cor}
\begin{proof}
We may assume that $\mbox{Vol}_{V(H)}(S)\le \frac 1 2 \mbox{Vol}_{V(H)}(V(H))$ since if this does not hold, we can redefine $S$ to be $V(H) - S$. Then $S$ satisfies the requirements of Lemma~\ref{Lem:Nibble} which for suitable $s$ and $b$ gives a set $C$ with $\Phi_{V(H)}(C)\le\theta$ and $|C| \ge \mbox{Vol}_{V(H)}(C\cap S)/d_{\max}\ge\frac 4 7 2^{b-1}/d_{\max}$. It follows from the pseudocode in Figure~\ref{fig:Nibble} that each vertex of $C$ has a positive $\tilde p_t$-value when \texttt{Nibble} halts and hence by the truncation operation, this value is at least $2\epsilon_b$ (we may assume that each such vertex has degree at least $1$). Since the total truncated probability mass is at most $1$, it follows that $|C| \le 1/(2\epsilon_b) = \tilde O(2^b/\theta^3)$. The last part of the corollary also follows from Lemma~\ref{Lem:Nibble} and from analyzing the pseudocode in Figure~\ref{fig:Nibble}.
\end{proof}

We need the following result which was shown in~\cite{SpielmanTeng}.
\begin{Lem}[\cite{SpielmanTeng}]\label{Lem:NibbleTrunc}
Let $s$ be a vertex of a connected graph $H$, let $b > 0$ and $t\geq 0$ be integers and let $0 <\theta\leq 1$. Let $\tilde p_t$ be the probability distribution found by iteration $t$ of \texttt{Nibble}$(H, s, \theta, b)$. Let $p_t$ be the probability distribution found by the variant of \texttt{Nibble}$(H, s, \theta, b)$ which does not truncate probabilities, i.e., line $4$ is replaced by $\tilde p_t \leftarrow P\tilde p_{t-1}$. Then $\tilde p_t \leq p_t$.
\end{Lem}

We say that \texttt{Nibble} \emph{visits} an edge $e = (u,v)$ if in some step, it sends a non-zero amount of probability mass along $e$, i.e., if $\tilde p_{t-1}$ has a non-zero entry for either $u$ or $v$ (or both) in some execution of line $4$ in Figure~\ref{fig:Nibble}. The next lemma bounds the number of edges visited by \texttt{Nibble}. This is key to making \texttt{Nibble} work efficiently in our dynamic setting.
\begin{Lem}\label{Lem:RandomWalkOverlap}
Let $e$ be an edge of an $n$-vertex connected graph $H$. Then the number of vertices $s$ for which \texttt{Nibble}$(H,s,\theta,b)$ visits $e$ is $O(2^b(\log^3n)/\theta^5)$.
\end{Lem}
\begin{proof}
We first consider probability distributions for random walks defined by matrix $P$ where no truncation occurs. Let $\Pr(u,w,t)$ denote the probability of reaching vertex $w$ in $t$ steps in a random walk in $H$ from vertex $u$ where in each step, the walk remains in the current vertex with probability $1/2$ and otherwise goes to one of the incident vertices with equal probability. Given $s_0,w\in V$ and integers $T\geq t\geq 0$, $\Pr(s_0,w,T) = \sum_{s\in V} \Pr(s_0,s,T - t)\Pr(s,w,t)$. It is well-known that in a connected graph $H'$, when the number of steps in a random walk from any starting vertex approaches infinity, the probability distribution for this walk converges to the stationary distribution in which the probability mass at each vertex $x$ is $d(x)/d(V(H'))$. Hence,
\[
  \frac{d(w)}{d(V(H))} = \lim_{T\rightarrow\infty}\Pr(s_0,w,T) = \lim_{T\rightarrow\infty}\sum_{s\in V} \Pr(s_0,s,T - t)\Pr(s,w,t) = \sum_{s\in V}\frac{d(s)}{d(V(H))}\Pr(s,w,t),
\]
implying that $d(w) = \sum_{s\in V(H)}d(s)\Pr(s,w,t)$.

Let $u$ and $v$ be the endpoints of $e$. \texttt{Nibble} visits $e$ if at some point it sends probability mass along $e$ either from $u$ to $v$ or from $v$ to $u$; we shall only bound the number of starting vertices for which the former happens since the same argument applies for the latter. By Lemma~\ref{Lem:NibbleTrunc}, \texttt{Nibble}$(H,s,\theta,b)$ only sends probability mass from $u$ to $v$ along $e$ if there is a $t$ such that $\Pr(s,u,t) \geq 2\epsilon_b d(u)$. Let $S_t = \{s\in V(H) | \Pr(s,u,t) \geq 2\epsilon_b d(u)\}$. Since $H$ is connected and contains at least two vertices, it has min degree at least $1$. The above then implies that the number of starting vertices $s$ for which \texttt{Nibble}$(H,s,\theta,b)$ sends probability mass from $u$ to $v$ along $e$ is at most
\[
  \left|\bigcup_{t = 1}^{t_0}S_t\right|\le \sum_{t = 1}^{t_0}|S_t|\le\sum_{t = 1}^{t_0}\sum_{s\in S_t}d(s)\le\sum_{t = 1}^{t_0}\sum_{s\in S_t}\frac{d(s)\Pr(s,u,t)}{2\epsilon_b d(u)}\le
  \frac{t_0}{2\epsilon_b} = O(2^b(\log^3n)/\theta^5).
\]
\end{proof}

\subsection{Preprocessing}
We are now ready to present \texttt{XPrune}$(e)$. Pseudocode can be seen in Figure~\ref{fig:XPrune}. In this subsection, we describe the preprocessing needed by this procedure.

In the following, we pick $\theta = \gamma/96 = \Theta_{n^\eps}(1)$ and $\gamma' = \theta_{+}^3/3^8 = \Theta_{n^{\eps}}(1)$. Note that our previous constraints in Section~\ref{subsec:PreprocDecSF} that $\gamma'\le \gamma$ and $\gamma' = \Theta_{n^\eps}(1)$ are satisfied.

Next, let $\rho = \frac 1 3 \theta_+ = \Theta_{n^\eps}(1)$ and let $p = \Theta((\ln n)/(\theta_+^2\kappa)) = \Theta_{n^\eps}(1/\kappa)$ be the probability from Corollary~\ref{Cor:LowConductance}. Furthermore, let $d_{\max} = 6p\kappa = \Theta_{n^\eps}(1)$ and let $b_{\max}$ be the largest integer $b$ such that the size lower bound in Corollary~\ref{Cor:SimpleNibble} is at most $64\Delta/(\gamma\kappa)$. Note that $b_{\max} = \lg(\Theta_{n^{\eps}}(\Delta/\kappa))$. Finally, let $h_{\max} = \lceil 64\Delta/(\gamma\kappa)\rceil = \Theta_{n^\eps}(\Delta/\kappa)$.

In the following, let multigraph $G_{\mathcal C}$ be defined as in Corollary~\ref{Cor:LowConductance}. For $i = 1,\ldots,h_{\max}$, we form a multigraph with vertex set $\mathcal C$ by sampling each edge of $G_{\mathcal C}$ independently with probability $p$. Let $\mathcal H$ denote a list of the $h_{\max}$ graphs obtained. By Invariant~\ref{Inv:Clustering} in Section~\ref{subsec:PreprocDecSF}, each vertex of each graph of $\mathcal H$ has expected degree at most $3p\kappa$ so by a Chernoff and a union bound, w.h.p.~each graph in $\mathcal H$ has max degree at most $d_{\max}$.

%Let $\mathcal X_{\mathcal C}$ be the collection of subsets $X_{\mathcal C} = \mathcal C(X)$ over all $X\in\mathcal X$; we shall refer to these subsets as expander clusters. For each graph $H_{\mathcal C}'\in\mathcal H$, associate with each edge $e$ in $H_{\mathcal C}'$ the expander cluster $X_{\mathcal C}(e)\in\mathcal X_{\mathcal C}$ containing it, if any. Similarly, associate with each vertex $C$ in $H_{\mathcal C}'$ the expander cluster $X_{\mathcal C}(C)$ containing it. Values $X_{\mathcal C}(e)$ and $X_{\mathcal C}(C)$ will change dynamically as expander clusters and clusters are updated.

We shall assume that at any point during the sequence of updates, each graph in $\mathcal H$ is simple so that \texttt{Nibble} can be applied to it. Furthermore, we assume that no edge deletion disconnects a spanning tree $T(C)$ (see Section~\ref{subsec:PreprocDecSF}) of a cluster $C\in\mathcal C$. We later show how to get rid of these assumptions.

The following preprocessing is done for each graph $H_{\mathcal C}\in\mathcal H$ and for $b = 1,\ldots,b_{\max}$. For each $s\in V(H_{\mathcal C})$, we run \texttt{Nibble}$(H_{\mathcal C}, s, \theta, b)$ and store the set $\mathcal E_b(s, H_{\mathcal C})$ of edges of $H_{\mathcal C}$ visited by this call. Having executed these calls, we then obtain and store dual sets $\mathcal S_b(e, H_{\mathcal C})$ consisting of all $s$ such that $e\in\mathcal E_b(s, H_{\mathcal C})$. For each $s\in V(H_{\mathcal C})$, we store a bit indicating whether $s$ is \emph{$b$-active} or \emph{$b$-passive} (in $H_{\mathcal C}$); we say that $s$ is $b$-active if \texttt{Nibble}$(H_{\mathcal C}, s, \theta, b)$ outputs a set. Otherwise, $s$ is $b$-passive.

Next, we check the condition in line $4$ of \texttt{XPrune} for each $H_{\mathcal C}\in\mathcal H$ and each $C\in\mathcal C$. If the condition is satisfied, we mark $C$ as a low-degree cluster in graph $H_{\mathcal C}$. We keep these low-degree clusters in a linked list $\mathcal L(H_{\mathcal C})$ which will be maintained during updates.

The $\mathcal E_b$-sets, their duals $\mathcal S_b$, and the list $\mathcal L(H_{\mathcal C})$ will only become relevant later on when we focus on the implementation and show how the tests in lines $4$ and $7$ can be done efficiently.

\begin{figure}[!ht]
\begin{tabbing}
\rule{\linewidth}{\arrayrulewidth}\\
d\=dd\=\quad\=\quad\=\quad\=\quad\=\quad\=\quad\=\quad\=\quad\=\quad\=\quad\=\quad\=\kill
\>\texttt{XPrune}$(e)$:\\\\
\>1. \>\>for each graph $H_{\mathcal C}\in\mathcal H$, $H_{\mathcal C}\leftarrow H_{\mathcal C} - \{e\}$\\
\>2. \>\>while $\mathcal H\neq\emptyset$\\
\>3. \>\>\>let $H_{\mathcal C}$ be the first graph in $\mathcal H$\\
\>4. \>\>\>if $\exists C\in\mathcal C$ s.t.~the number of edges of $H_{\mathcal C}$ leaving $C$ is less than $\kappa\theta_{+} p$\\
\>5. \>\>\>\>update $\mathcal C(X)\leftarrow \mathcal C(X) - \{C\}$\\
\>6. \>\>\>\>$\mathcal H\leftarrow\mathcal H - \{H_{\mathcal C}\}$\\
\>7. \>\>\>else if $\exists b\in\{1,\ldots,b_{\max}\}\exists s\in V(H_{\mathcal C})$ s.t.~\texttt{Nibble}$(H_{\mathcal C}[\mathcal C(X)], s, \theta, b)$ outputs a set $K_{\mathcal C}$\\
\>8. \>\>\>\>let $K_{\mathcal C}'$ be a set of smaller size among $K_{\mathcal C}$ and $\mathcal C(X) - K_{\mathcal C}$\\
\>9. \>\>\>\>update $\mathcal C(X)\leftarrow \mathcal C(X) - K_{\mathcal C}'$\\
\>10. \>\>\>\>$\mathcal H\leftarrow\mathcal H - \{H_{\mathcal C}\}$\\
\>11. \>\>\>else output the set of clusters removed from $\mathcal C(X)$ in lines $2$--$10$\\
\rule{\linewidth}{\arrayrulewidth}
\end{tabbing}
\caption{Pseudocode for procedure \texttt{XPrune} which keeps track of low-conductance cuts of small size. It has access to the list of graphs $\mathcal H$ as well as to $\mathcal C$ and $\mathcal C(X)$. It is assumed that each graph $H_{\mathcal C}'$ is simple and that no edge deletion splits a cluster in $\mathcal C$. Set $X$ is implicitly updated when $\mathcal C(X)$ is updated in lines $5$ and $9$.}\label{fig:XPrune}
\end{figure}

\subsection{Correctness}\label{subsec:XPruneCorrectness}
We now show that the two properties of \texttt{XPrune} in Section~\ref{subsec:PreprocDecSF} hold. Lemma~\ref{Lem:TotSizeSmallCuts} below implies the second property of \texttt{XPrune}. To show this lemma, we first need the following result.
\begin{Lem}\label{Lem:XPruneDeg}
W.h.p., in each execution of the while-loop of \texttt{XPrune}$(e)$, $G_{\mathcal C}[\mathcal C(X)]$ has min degree at least $\frac 1 3 \kappa\theta_{+} = \kappa\rho$ in line $7$, and at the beginning of line $5$ $C$ has degree less than $3\kappa\theta_{+}$ in $G_{\mathcal C}[\mathcal C(X)]$.
\end{Lem}
\begin{proof}
Procedure \texttt{XPrune} maintains the invariant that in every execution of line $3$, the edges of $H_{\mathcal C}$ are sampled independently of the updates to $X$ thus far. This follows since every time we update $X$, we remove $H_{\mathcal C}$ from $\mathcal H$ while the updates to $X$ have been done independently of the remaining graphs in $\mathcal H$.

Consider a single iteration of the while-loop and consider some cluster $C\in \mathcal C(X)$ in line $3$. Let $\delta$ denote its degree in $G_{\mathcal C}[\mathcal C(X)]$ and assume first that $\delta < \frac 1 3 \kappa\theta_{+}$. We will show that w.h.p., line $7$ is not reached in the current iteration of the while-loop. By the above, the expected degree of $C$ in $H_{\mathcal C}[\mathcal C(X)]$ is $\delta p < \frac 1 3 \kappa\theta_{+}p$ and a Chernoff bound implies that w.h.p., the actual degree in $H_{\mathcal C}[\mathcal C(X)]$ is less than $\kappa\theta_{+} p$. Thus, w.h.p., line $7$ is not reached in this iteration of the while-loop.

Now, assume that $\delta \ge 3\kappa\theta_{+}$. Then the expected degree of $C$ in $H_{\mathcal C}[\mathcal C(X)]$ is at least $3\kappa\theta_{+}p$ and a Chernoff bound shows that w.h.p., its actual degree in $H_{\mathcal C}[\mathcal C(X)]$ is at least $\kappa\theta_{+}p$. Hence, w.h.p., $C$ is not removed in line $5$ of the current iteration of the while-loop. A union bound over all choices for $C$ shows the second part of the lemma.
\end{proof}
The following lemma easily implies the second property of \texttt{XPrune} in Section~\ref{subsec:PreprocDecSF}.
\begin{Lem}\label{Lem:TotSizeSmallCuts}
Let $S\subseteq V$ be the union of clusters output over all calls to a variant of \texttt{XPrune} which does not require the upper bound $b_{\max}$ on $b$ in line $7$. Then w.h.p., $|S| \leq 64\Delta/\gamma\leq h_{\max}\kappa$.
\end{Lem}
\begin{proof}
First observe that w.h.p., $G[X_{\mathit{init}}]$ is a $\gamma$-expander graph where $X_{\mathit{init}}$ is the initial set $X$.

Now, consider the start of an execution of line $9$ and let $K = \cup_{C\in K_{\mathcal C}'}C$. By Corollary~\ref{Cor:SimpleNibble}, w.h.p., $\Phi_{H_{\mathcal C}}(K_{\mathcal C}') = \Phi_{H_{\mathcal C}}(K_{\mathcal C})\le \theta$ and $|K_{\mathcal C}'|\le |\mathcal C(X) - K_{\mathcal C}'|$. By Lemma~\ref{Lem:XPruneDeg}, w.h.p., $G_{\mathcal C}[\mathcal C(X)]$ has min degree at least $\kappa\rho$. With $\rho' = \theta > \rho$, it follows from the first part of Corollary~\ref{Cor:LowConductance} that w.h.p., $\Phi_{G[X]}(K)\le 4\theta = \gamma/24$ so by Invariant~\ref{Inv:Clustering}, the number of edges of $G[X]$ crossing $(K,X-K)$ is at most $(\gamma/24)\min\{\mbox{Vol}_{G[X]}(K),\mbox{Vol}_{G[X]}(X-K)\}\le(\gamma/8)|K|$. Since $|K_{\mathcal C}'|\le |\mathcal C(X) - K_{\mathcal C}'|$ and since $G$ has max degree $3$, Invariant~\ref{Inv:Clustering} implies that $|K|\le 3\kappa|K_{\mathcal C}'|\le 3\kappa|\mathcal C(X) - K_{\mathcal C}'|\le 3|X - K|$ so $4|K|\le 3|X|$ and hence $|K|\le\frac 3 4 |X|\le\frac 3 4 |X_{\mathit{init}}|$.

Next, consider the start of an execution of line $5$. By Lemma~\ref{Lem:XPruneDeg} and Invariant~\ref{Inv:Clustering}, w.h.p., the number of edges of $G[X]$ crossing $(C, X - C)$ is less than $3|C|\theta_{+} < 3|C|\theta < (\gamma/8)|C|$. We may assume that $|X_{\mathit{init}}| > 64\Delta/\gamma = \omega(\kappa)$ so that $|C|\le\frac 3 4 |X_{\mathit{init}}|$.

Now, let $\mathcal K'$ be the family of all sets removed from $X$, either in an execution of line $5$ or of line $9$. Note that $S = \cup_{K\in\mathcal K'}K$. Let $X$ denote $X_{\mathit{init}}$ and let $X' = X - S$. Note that $\mathcal K = \mathcal K'\cup\{X'\}$ is a partition of $X_{\mathit{init}}$. We consider two cases: $|S| < \frac 1 2 |X|$ and $|S|\ge\frac 1 2 |X|$.

If $|S| < \frac 1 2 |X|$ then w.h.p., the initial number of edges of $G[X]$ crossing $(X',X - X') = (X',S)$ is at least $\gamma|S|$. By the above, w.h.p.~the number of edges of $G[X]$ crossing $(X',X - X')$ at termination is at most $(\gamma/8)|S|$. Hence, w.h.p.~at least $\frac 7 8 \gamma|S|$ edges crossing this cut have been deleted over all updates in which case $\Delta\ge\frac 7 8\gamma|S|$, as desired.

Next, assume that $|S|\ge \frac 1 2 |X|$. Order the sets of $\mathcal K'$ by when they were removed from $X$. There is a an integer $k$ such that if $\mathcal K_1$ resp.~$\mathcal K_2$ is the subset of the $k$ resp.~$k+1$ first sets in this ordering then $X_1 = \cup_{K\in\mathcal K_1}K$ has size less than $\frac 1 2 |X|$ and $X_2 = \cup_{K\in\mathcal K_2}$ has size at least $\frac 1 2 |X|$. If $|X_1|\ge\frac 1 8 |X|$ then w.h.p., the total number of deleted edges crossing $(X_1,X-X_1)$ is at least $\frac 7 8 \gamma |X_1|\ge\frac 7{64}\gamma|X|\ge\frac 7{64}\gamma|S|$. If $|X_1| < \frac 1 8 |X|$ then $|X_2| \le |X_1| + \frac 3 4 |X| < \frac 7 8 |X|$ so w.h.p., the total number of deleted edges crossing $(X_2,X-X_2)$ is at least $\gamma |X - X_2| - \frac\gamma 8 |X_2| > \frac \gamma 8 |X| - \frac 7{64}\gamma |X| \ge \frac \gamma{64}|S|$.

In both cases, w.h.p.~$\Delta\ge\frac \gamma{64}|S|$, showing the desired.
\end{proof}
The second property of \texttt{XPrune} follows from this lemma since by the choice of $b_{\max}$ and by Invariant~\ref{Inv:Clustering}, w.h.p.~\texttt{XPrune} and the variant in Lemma~\ref{Lem:TotSizeSmallCuts} behave in exactly the same manner.

The next lemma shows that the first property of \texttt{XPrune} is maintained over all edge deletions, assuming $\mathcal H$ never becomes empty.
\begin{Lem}\label{Lem:DeleteEdgeProp}
Suppose a call to \texttt{XPrune} has just returned where in each execution of line $2$, $\mathcal H$ was non-empty. Then w.h.p., for every $\mathcal C(X)$-respecting cut $(K, X - K)$, the number of edges of $G[X]$ crossing $(K,X-K)$ is at least $\gamma'\min\{|K|,|X - K|\}$.
\end{Lem}
\begin{proof}
We prove the lemma for the variant of \texttt{XPrune} in Lemma~\ref{Lem:TotSizeSmallCuts}. This suffices as argued above.

Consider a moment where \texttt{XPrune} has just returned and let $H_{\mathcal C}$ be the first graph in $\mathcal H$. As argued in the proof of Lemma~\ref{Lem:XPruneDeg}, the edges of $H_{\mathcal C}$ are sampled independently of the updates to $X$ done so far. By Corollary~\ref{Cor:SimpleNibble}, for every cut $(K_{\mathcal C},\mathcal C(X) - K_\mathcal C)$, $\Phi_{H_{\mathcal C}}(K_{\mathcal C}) > 2\theta_{+} = 6\rho$. Lemma~\ref{Lem:XPruneDeg} implies that w.h.p., $G_{\mathcal C}[\mathcal C(X)]$ has min degree at least $\rho\kappa$. By the second part of Corollary~\ref{Cor:LowConductance}, w.h.p.~for every cut $(K_{\mathcal C},\mathcal C(X) - K_\mathcal C)$, $\Phi_{G[X]}(K)\ge \rho^3/3^5 = \gamma'$ where $K = \cup_{C\in K_{\mathcal C}}C$; hence, the number of edges of $G[X]$ crossing $(K,X-K)$ is at least $\gamma'\min\{\mbox{Vol}_{G[X]}(K),\mbox{Vol}_{G[X]}(X - K)\}\ge \gamma'\min\{|K|,|X-K|\}$, as desired.
\end{proof}

The final lemma of this subsection shows that the requirement in Lemma~\ref{Lem:DeleteEdgeProp} of $\mathcal H$ being non-empty can be dropped. This shows the correctness of \texttt{XPrune}.
\begin{Lem}\label{Lem:DeleteEdgeT}
W.h.p., $\mathcal H$ is non-empty in all executions of line $3$ in calls to \texttt{XPrune}.
\end{Lem}
\begin{proof}
Since we assume that clusters never become disconnected it follows from Invariant~\ref{Inv:Clustering} that every time \texttt{XPrune} removes a graph from $\mathcal H$, the size of $X$ is reduced by at least $\kappa$. By Lemma~\ref{Lem:TotSizeSmallCuts}, w.h.p.~this happens no more than $h_{\max}$ times which is the initial size of $\mathcal H$.
\end{proof}

\subsection{Implementation}
We now show how to implement the preprocessing and the update step of \texttt{XPrune} and analyze the performance of this implementation.

\subsubsection{Preprocessing}
We first describe how to obtain the list $\mathcal H$ of graphs. We will use an adjacency list representation for each graph in $\mathcal H$ and we use a linked list representation of $\mathcal H$ itself. Initially, all $h_{\max} = \Theta_{n^{\eps}}(\Delta/\kappa)$ graphs in $\mathcal H$ are empty, containing only vertices. To obtain the edges, the trivial way of scanning through all the graphs in $\mathcal H$ and including each $e$ of $G_{\mathcal C}$ in each of them independently with probability $p$ will be too slow so we need to do something more clever.

For each edge $e\in G_{\mathcal C}$, we apply a procedure that we describe in the following.

We keep a list $L$ of the graphs from $\mathcal H$. This list will shrink during the course of the algorithm. We represent the initial  $L$ as an array and keep an index to the start of $L$. We implicitly shrink $L$ by increasing this index, letting the new $L$ be the suffix of the array starting from this index.

Let $s$ be the current length of $L$. For $k = 1,\ldots,s$, let $\mathcal E_{k,s}$ be the event that the $k$th graph in $L$ is the first to include $e$ among all the graphs in $L$. Let $\mathcal E_{s+1,s}$ be the event that $e$ is not added to any graph in $L$. Then $p_{k,s} = \Pr(\mathcal E_{k,s}) = (1-p)^{k-1}p$ for $k = 1,\ldots,s$ and $p_{k,s} = \Pr(\mathcal E_{k,s}) = (1-p)^s$ for $k = s+1$. We pick $k$ randomly according to this probability distribution, add $e$ to the $k$th graph in $L$ (assuming $k \leq s$), update $L$ to its suffix of length $\max\{0, s - k\}$, and then repeat the procedure on the new list $L$. The procedure stops when $L$ is empty. Running this procedure is equivalent to including $e$ in each graph of $\mathcal H$ independently with probability $p$.

We need to describe how to pick $k$ from this distribution. We first precompute $p_{i_1,i_2,s} = \sum_{i = i_1}^{i_2} p_{i,s}$ for all $1\le i_1\le i_2\le s+1\le h+1$. We shall not do this explicitly. Instead, observing that $p_{i_1,i_2,s} = p_{i_1,i_2,i_2+1}$ when $i_2\le s$, we only compute $p_{i_1,i_2,i_2}$ and $p_{i_1,i_2,i_2+1}$ for all choices of $i_1$ and $i_2$, using a simple bottom-up dynamic programming procedure. From these values, we can obtain any $p_{i_1,i_2,s}$ in constant time.

Given these precomputed values and letting $s$ be the current length of $L$, we find the next $k$ with the following recursive procedure which takes the pair $(i_1,i_2)$ as input which is initially $(1,s+1)$. If $i_1 < i_2$, let $j = \lceil (i_1 + i_2)/2\rceil$. We recurse with the pair $(i_1,j-1)$ with probability $p_{i_1,j-1,s}/p_{i_1,i_2,s}$ and recurse with the pair $(j,i_2)$ otherwise, i.e., with probability $p_{j,i_2,s}/p_{i_1,i_2,s}$. The recursion stops once $i_1 = i_2$ in which case we pick $k = i_1$. It is easy to see that this recursive procedure picks $k$ according to the distribution above. This completes the description of the implementation of the procedure for forming $\mathcal H$.

Each set $\mathcal E_b(s,H_{\mathcal C})$ resp.~$\mathcal S_b(e,H_{\mathcal C})$ is stored as a linked list with a pointer from $s$ resp.~$e$ to the start of this list. This completes the description of the implementation of the preprocessing step.

The following lemma shows the performance of the procedure just described.
\begin{Lem}\label{Lem:ConstructionHGraphs}
With high probability, the initial graphs in $\mathcal H$ can be constructed in worst-case time $\tilde O(n + ph_{\max}n + h_{\max}^2) = \tilde O(n) + O_{n^\eps}(\Delta)$.
\end{Lem}
\begin{proof}
In the proof, we use the same notation as in the description of the implementation above. For each edge $e\in G_{\mathcal C}$, we let $s(e)$ be the number of graphs in $\mathcal H$ that $e$ is included in at the end of the preprocessing step.

Consider an edge $e\in G_{\mathcal C}$. We get $E[\sum_{e\in E(G_{\mathcal C})}s(e)] = \sum_{e\in E(G_{\mathcal C})}E[s(e)] = \sum_{e\in E(G_{\mathcal C})}ph_{\max} = O(ph_{\max}n)$. A Chernoff bound shows that w.h.p., $\sum_{e\in E(G_{\mathcal C})}s(e) =  O(ph_{\max}n) = O_{n^{\eps}}(\Delta)$. Excluding the time to precompute $p_{i_1,i_2,s}$, the lemma will thus follow if we can show that $e$ can be processed in $O((s(e)+1)\log n)$ worst-case time.

The number of times we need to pick a $k$ from the distribution is either $s(e)$ or $s(e)+1$. It is easy to see that, given precomputed values $p_{i_1,i_2,s}$, the recursive procedure runs in $O(\log n)$ time, as desired. Values $p_{i_1,i_2,s}$ only need to be computed once, not for each edge $e$. The time to compute these values is $O(h_{\max}^2) = O((\Delta/\kappa)^2) = O_{n^{\eps}}(\Delta^2/n) = O_{n^{\eps}}(\Delta)$.
\end{proof}

The rest of the preprocessing is dominated by the time for calls to \texttt{Nibble}. With high probability, the number of calls to \texttt{Nibble} with parameter $b$ is $O(h_{\max}(n/\kappa))$. By Corollary~\ref{Cor:SimpleNibble}, each such call takes $\tilde O(2^b/\theta^5)$; hence w.h.p., the total time for calls to \texttt{Nibble} over all $b$ is $\tilde O(h_{\max}(n/\kappa)2^{b_{\max}}/\theta^5) = O_{n^\eps}(\Delta^2n/\kappa^3) = O_{n^{\eps}}(\Delta^2/\sqrt n)$.

%For each $H_{\mathcal C}'\in\mathcal H$, obtaining $\mathcal L(H_{\mathcal C}')$  can be done in $O(|H_{\mathcal C}'|)$ time which w.h.p.~is $O((n/\kappa)p) = O((n/\kappa^2)(\ln n)/\theta_{+}^2)$ time for a total of $O((Mn\ln n)/(\gamma\kappa^3\theta_{+}^2))$ time.

We conclude that w.h.p., the total preprocessing time is $\tilde O(n) + O_{n^\eps}(\Delta + \Delta^2/\sqrt n) = \tilde O(n) + O_{n^{\eps}}(\Delta^2/\sqrt n)$ which is within the bound of Theorem~\ref{Thm:DisconnectExpanderGraph}.

\subsubsection{Updates}
We now describe how to implement \texttt{XPrune}. Note that in line $7$, \texttt{Nibble} is applied to subgraphs of graphs $H_{\mathcal C}\in\mathcal H$ induced by $\mathcal C(X)$. In our implementation, we shall maintain these subgraphs explicitly by removing edges of $H_{\mathcal C}$ incident to clusters removed from $\mathcal C(X)$. This way, $H_{\mathcal C}[\mathcal C(X)]$ is a component of $H_{\mathcal C}$ so when we run \texttt{Nibble} on $H_{\mathcal C}$ with a start vertex in $\mathcal C(X)$, we do not need to worry about edges not in $H_{\mathcal C}[\mathcal C(X)]$ being visited.

To implement line $1$, we do as follows for each $H_{\mathcal C}\in\mathcal H$. Assume that $e\in H_{\mathcal C}$ as otherwise nothing needs to be done. After deleting $e$ from $H_{\mathcal C}$, for each $b = 1,\ldots,b_{\max}$, the only calls \texttt{Nibble}$(H_{\mathcal C}[\mathcal C(X)], s, \theta, b)$ that are affected by the deletion are those from vertices $s\in \mathcal S_b(e, H_{\mathcal C})$. For each such $s$, we run \texttt{Nibble}$(H_{\mathcal C}[\mathcal C(X)], s, \theta, b)$; $s$ is made $b$-active if a set is returned and $b$-passive otherwise. Let $\mathcal E_b'(s,H_{\mathcal C})$ be the set of edges visited by this call. For each $e'\in\mathcal E_b(s,H_{\mathcal C}) - \mathcal E_b'(s,H_{\mathcal C})$, we remove $s$ from $\mathcal S_b(e',H_{\mathcal C})$ and for each $e'\in\mathcal E_b'(s,H_{\mathcal C}) - \mathcal E_b(s,H_{\mathcal C})$, we add $s$ to $\mathcal S_b(e',H_{\mathcal C})$. Finally, we update $\mathcal E_b(s,H_{\mathcal C})$ to $\mathcal E_b'(s,H_{\mathcal C})$. This correctly updates all $\mathcal S$- and $\mathcal E$-sets.

To maintain lists $\mathcal L(H_{\mathcal C})$ in line $1$, we only need to check the low-degree condition of line $4$ for the clusters containing $u$ and $v$ and for the clusters that have been merged or split during the current update (the latter will only be relevant when we later allow clusters to change over time).

Next, we describe how lines $3$ to $6$ are implemented. For the condition in line $4$, checking each $C\in\mathcal C$ will be too slow. Instead, we make use of the $\mathcal L$-lists. Consider any execution of line $4$. If $\mathcal L(H_{\mathcal C})$ is empty, no $C$ exists satisfying the condition. Otherwise, we obtain $C$ by extracting the first element of $\mathcal L(H_{\mathcal C})$.

Handling the update in line $5$ is done as follows. For each graph $H_{\mathcal C}'\in\mathcal H - \{H_{\mathcal C}\}$, we delete from $H_{\mathcal C}'$ every edge incident to $C$. For each deleted edge $e'$, we run \texttt{Nibble} from all vertices in $\mathcal S_b(e',H_{\mathcal C}')$ for all $b$ and update the $b$-active/$b$-passive bits and the $\mathcal S$- and $\mathcal E$-sets as above. For each cluster $C'$ incident to $C$ in $H_{\mathcal C}'$, the removal of $C$ may have caused $C'$ to now have low degree. We update the adjacency list of each such $C'$ and add it to $\mathcal L(H_{\mathcal C}')$ if it has low degree. We maintain $\mathcal C(X)$ implicitly by associating a bit with each cluster indicating whether it belongs to $\mathcal C(X)$. Clusters removed from $\mathcal C(X)$ in lines $5$ and $8$ are stored in a linked list which is output in line $10$.

To implement line $7$, we check for each $b$ if there are any $b$-active vertices in $H_{\mathcal C}$. If not, the condition in line $7$ cannot be satisfied and we execute line $10$. Otherwise, we pick a $b$ and a $b$-active vertex $s$ in $H_{\mathcal C}$ and run \texttt{Nibble}$(H_{\mathcal C}[\mathcal C(X)], s, \theta, b)$.

The update in line $9$ is handled similarly to line $5$ the only modification being that we process every cluster $C$ on the $K_{\mathcal C}$-side of the cut rather than just a single cluster.

\paragraph{Performance:}
The update time is dominated by the time spent in the while-loop. Consider a single execution of lines $3$ to $6$. Since we maintain the $\mathcal L$-lists, we can obtain a cluster $C$ satisfying the condition in line $4$ in $O(1)$ time, assuming such a $C$ exists. If it does then since w.h.p., each vertex of $H_{\mathcal C}$ has degree at most $d_{\max} = O_{n^{\eps}}(1)$ and since $C$ has degree less than $\kappa\theta_{+}p = O_{n^\eps}(1)$, w.h.p., updating $\mathcal L$-lists takes $O_{n^\eps}(1)$ time per graph $H_{\mathcal C}'\in\mathcal H - \{H_{\mathcal C}\}$ since we only need to update adjacency lists for clusters adjacent to $C$. Since we delete from $H_{\mathcal C}'$ the edges incident to $C$, it follows from Lemma~\ref{Lem:RandomWalkOverlap} that for each $b$, w.h.p.~only $\tilde O(d_{\max}2^b/\theta^5) = O_{n^\eps}(2^b)$ calls to \texttt{Nibble} in $H_{\mathcal C}'$ with parameter $b$ need to be updated which by Corollary~\ref{Cor:SimpleNibble} takes a total of $O_{n^\eps}(2^{2b})$ time. Over all $b$ and $H_{\mathcal C}'$, this is $O_{n^\eps}(h_{\max}2^{2b_{\max}}) = O_{n^\eps}(\Delta^3/\kappa^3) = O_{n^\eps}(\Delta^3/n^{3/2})$ time.

We have bounded the time for a single execution of lines $3$ to $6$. By Lemma~\ref{Lem:TotSizeSmallCuts}, the number of executions of these lines in an update is $O(\Delta/(\gamma\kappa)) = O_{n^\eps}(\Delta/\sqrt n)$ which sums up to a total time for these executions of $O_{n^\eps}(\Delta^4/n^2)$.

It remains to bound the total time spent in lines $7$ to $10$ during an update. Consider a single execution of these lines. For each $H_{\mathcal C}'\in\mathcal H - \{H_{\mathcal C}\}$, the number of edges deleted from $H_{\mathcal C}'$ is $O(|K|p)$ (w.h.p.) where $K = \cup_{C\in K_{\mathcal C}'}C$. These edges are incident in $H_{\mathcal C}'$ to at most $O(|K|p)$ clusters so updating the $\mathcal L$-lists in $H_{\mathcal C}'$ takes $O(|K|p^2\kappa) = O_{n^\eps}(|K|/\kappa) = O_{n^\eps}(|K|/\sqrt n)$ time. For each $H_{\mathcal C}\in\mathcal H$ and each $b$, we keep the set of $b$-active vertices of $H_{\mathcal C}$ in a linked list so that we can identify such a vertex in constant time if it exists. By Lemma~\ref{Lem:RandomWalkOverlap}, the number of calls to \texttt{Nibble} in $H_{\mathcal C}'$ with parameter $b$ that are updated is $\tilde O((|K|p)(2^b/\theta^5)) = O_{n^\eps}(|K|2^b/\sqrt n)$ and by Corollary~\ref{Cor:SimpleNibble}, the total time for these calls is $O_{n^\eps}(|K|2^{2b}/\sqrt n)$. Over all $b$ and $H_{\mathcal C}'$, this is $O_{n^\eps}(|K|2^{2b_{\max}}h_{\max}/\sqrt n) = O_{n^\eps}(|K|\Delta^3/n^2)$. By Lemma~\ref{Lem:TotSizeSmallCuts}, the total size of sets $K$ over all executions of line $9$ in an update is $O(\Delta/\gamma) = O_{n^\eps}(\Delta)$. Hence, total time for lines $7$ to $10$ in a single update is $O_{n^\eps}(\Delta^4/n^2)$.

It follows that w.h.p., we get an update time within the bound of Theorem~\ref{Thm:DisconnectExpanderGraph}.

\subsection{From multigraphs to simple graphs}
Above we made two simplifying assumptions, namely that each graph $H_{\mathcal C}\in\mathcal H$ is simple and that no edge deletion disconnects a cluster in $\mathcal C$. In this subsection, we focus on getting rid of the former assumption.

Associate with each graph $H_{\mathcal C}\in\mathcal H$ a simple graph $\overline H_{\mathcal C}$ as follows. For each vertex $C$ of $H_{\mathcal C}$, if $d_C$ denotes its degree, then we have vertex set $\overline C = \{v_1(C),\ldots,v_{d_C}(C)\}$ in $\overline H_{\mathcal C}$ where $\{v_1(C),\ldots,v_{d_C}(C)\}$ is the subset of $V$ of endpoints in $C$ of edges of $H_{\mathcal C}$ incident to $C$. For each edge $e = (C,D)$ of $H_{\mathcal C}$, if $e$ is the $i$th edge incident to $C$ and the $j$th edge incident to $D$ in their adjacency list orderings, we add to $\overline H_{\mathcal C}$ the edge $(v_i(C), v_j(D))$ and identify this edge with $e$. To complete the construction of $\overline H_{\mathcal C}$, we apply for each $C\in V(H_{\mathcal C})$ the algorithm of Lemma~\ref{Lem:SimpleExpander}, giving w.h.p.~a $1$-expander graph of $\overline C$ with $O(|\overline C|) = O(d_C)$ edges and max degree $O(\log d_C)$. Note that $V(\overline H_{\mathcal C})\subseteq X$.

In the following, let $\overline X$ be $X$ restricted to the union of $V(\overline C)$ over all $C\in\mathcal C(X)$.

\paragraph{Preprocessing:}
We now describe the modifications to \texttt{XPrune}. In the preprocessing step, we form both graphs $H_{\mathcal C}$ as well as the graphs $\overline H_{\mathcal C}$. Instead of applying \texttt{Nibble} to all vertices of $H_{\mathcal C}$, we now apply it to all vertices of $\overline H_{\mathcal C}$ and the $\mathcal E$- and $\mathcal S$-sets are formed w.r.t.~$\overline H_{\mathcal C}$ but only for inter-cluster edges. This suffices since edges of $1$-expander graphs are unchanged over all updates. The range of $b$-values is changed since $b_{\max}$ is adjusted, as we describe later.

\paragraph{Updates:}
Now consider the update step where we no longer assume that graphs in $\mathcal H$ are simple. Except for line $7$ in \texttt{XPrune}, we use these graphs as before. In line $7$, \texttt{Nibble} requires a simple graph as input. We instead give as input to this procedure the graph $\overline H_{\mathcal C}[\overline X]$.

Suppose \texttt{Nibble} outputs a set $K$. Note that $K$ may not be $\mathcal C(X)$-respecting. To form $K_{\mathcal C}$ in line $7$ of \texttt{XPrune}, we apply an algorithm which is essentially the same as the one in the third part of Lemma~\ref{Lem:LowConductanceMultigraph}. More precisely, let $\mathcal C_1$ be the collection of vertex sets $\overline C$ intersecting both sides of $(K,\overline X - K)$ and $|\overline C - K|\leq |\overline C\cap K|$ and let $\mathcal C_2$ be the collection of the remaining vertex sets $\overline C$ intersecting both sides of $(K,\overline X - K)$. We let $K_{\mathcal C}$ be the set of clusters $C$ such that $\overline C\subseteq (K\cup\cup_{C\in\mathcal C_1}\overline C) - (\cup_{C\in\mathcal C_2}\overline C)$.

The $\mathcal E$- and $\mathcal S$-sets are maintained as before but for the graphs $\overline H_{\mathcal C}$. This completes the description of the modifications needed in \texttt{XPrune}.

\subsubsection{Correctness}
We now show that the modified version of \texttt{XPrune} is correct for suitable new choices of the parameters of this section.

First we claim that for any $H_{\mathcal C}\in\mathcal H$ and any $\mathcal C(X)$-respecting cut, the conductance of this cut in $H_{\mathcal C}$ and in $\overline H_{\mathcal C}$ differ by only a constant factor. To see this, observe that the number of edges crossing the cut is the same in the two graphs. Since the expander graphss inserted when forming $\overline H_{\mathcal C}$ are sparse and since there are $\Theta(|\overline C|)$ edges of $H_{\mathcal C}$ incident to each cluster $C$, the volume of each side of the cut differs by only a constant factor in the two graphs. Hence, the conductance of the cut differs by only a constant factor in the two graphs.

It follows from what we have just shown that if every cut $(K,\overline X - K)$ has conductance at least $\theta$ in a graph $\overline H_{\mathcal C}[\overline X]$ then in particular every $\mathcal C(X)$-respecting cut in $H_{\mathcal C}[\mathcal C(X)]$ has conductance $\Omega(\theta)$.

Next, we claim that if a cut $(K,\overline X - K)$ has conductance at most $\theta$ in a graph $\overline H_{\mathcal C}[\overline X]$ then $K_{\mathcal C}$, obtained as described above, has conductance $O(\theta)$ in $H_{\mathcal C}[\mathcal C(X)]$. To see this, note that each vertex of $\overline H_{\mathcal C}$ has only a constant number of incident inter-cluster edges. The proof now follows using the same arguments as in the proof of the third part of Lemma~\ref{Lem:LowConductanceMultigraph}. These arguments also show that the number of vertices of $\overline X$ in $K_{\mathcal C}$
and in $K$ differ by only a constant factor.

We now go through the lemmas in Section~\ref{subsec:XPruneCorrectness} that need to be adjusted to the new version of \texttt{XPrune}. Previously, we set $\theta = \gamma/96$. By the above observations, Lemma~\ref{Lem:TotSizeSmallCuts} remains correct if we make $\theta$ smaller by a sufficiently big constant factor.

To ensure that Lemma~\ref{Lem:DeleteEdgeProp} remains correct, first note that w.h.p., $d_{\max}$ in Corollary~\ref{Cor:SimpleNibble} is now $\tilde O(1)$. We will determine the new value of $b_{\max}$. By Lemma~\ref{Lem:TotSizeSmallCuts}, \texttt{Nibble} only needs to identify sets $K\subseteq\overline X$ such that $|K_{\mathcal C}| = O(\Delta/(\gamma\kappa))$. Since w.h.p.~graphs in $\mathcal H$ have max degree $O(\kappa p)$, we get w.h.p.~that $|K|$ is upper bounded by $O(\kappa p|K_{\mathcal C}|) = O(p\Delta/\gamma) = \tilde O(\Delta/(\gamma\kappa\theta_{+}^2))$. Defining $b_{\max}$ as before to be the largest integer $b$ such that the size lower bound in Corollary~\ref{Cor:SimpleNibble} is at most this upper bound, we get $b_{\max} = \lg(\tilde\Theta(\Delta/(\gamma\kappa\theta_{+}^2))) = \lg(\Theta_{n^{\eps}}(\Delta/\kappa))$. Recall that we previously chose $\gamma' = \theta_{+}^3/3^8$. It follows from the above that Lemma~\ref{Lem:DeleteEdgeProp} remains correct if we make $\gamma'$ smaller by a sufficiently large constant factor.

It is easy to see that the remaining lemmas in Section~\ref{subsec:XPruneCorrectness} remain correct for the modified version of \texttt{XPrune}.

\subsubsection{Performance}
It remains to show the performance of the modified version of \texttt{XPrune}.

\paragraph{Preprocessing:}
By Lemma~\ref{Lem:SimpleExpander}, forming graphs $\overline H_{\mathcal C}$ does not increase the asymptotic preprocessing time. The remaining time spent is dominated by the calls to \texttt{Nibble}. Since w.h.p.~the number of vertices of each graph $\overline H_{\mathcal C}$ is a factor of $\Theta(\kappa p) = \Theta_{n^{\eps}}(1)$ larger than the number of vertices in $H_{\mathcal C}$, the number of calls to \texttt{Nibble} increases by this factor as well. Since $2^{b_{\max}}$ is a factor of $O_{n^\eps}(1)$ larger than before and since $d_{\max}$ in Corollary~\ref{Cor:SimpleNibble} is a factor of $O_{n^\eps}(1)$ smaller, each call to \texttt{Nibble} takes the same time as before up to a constant number of $n^{\eps}$-factors. Hence, the overall time for this part of the preprocessing increases by a factor of $O_{n^\eps}(1)$. We conclude that the preprocessing time bound in the previous subsection still holds.

\paragraph{Updates:}
As observed above, each call to \texttt{Nibble} takes the same amount of time as before up to $\log$-factors. By Lemma~\ref{Lem:RandomWalkOverlap} and the above, the number of calls to \texttt{Nibble} per edge deletion increases by a factor of $\tilde\Theta(1/\theta_{+}^2) = O_{n^\eps}(1)$. Since we only delete inter-cluster edges, our previous bound on the number of edges deleted per update remains valid. Hence, the update time bound in the previous subsection still holds.

\subsection{Handling cluster splits and merges}
We now remove the remaining simplifying assumption and allow clusters to become disconnected. Clusters can now both split and merge as described in Section~\ref{subsec:UpdateClusters}. The preprocessing step remains the same so we only focus on updates.

Recall that only $O(1)$ clusters become split or merged per update. Assume for now that at all times, each cluster has size between $\kappa$ and $3\kappa$. We modify line $1$ of \texttt{XPrune}$(e)$ so that it does the following for each $H_{\mathcal C}\in\mathcal H$ when a cluster is split by the deletion of $e$. For each cluster $C$ destroyed by the updates to $\mathcal C$, we delete $C$ and temporarily delete its incident inter-cluster edges from $H_{\mathcal C}$. In $\overline H_{\mathcal C}$, we delete $\overline C$ and its $1$-expander graph and temporarily delete its incident inter-cluster edges. For every new cluster $C$, we add $C$ to $H_{\mathcal C}$ along with its incident inter-cluster edges that were temporarily deleted. In $\overline H_{\mathcal C}$, we add $\overline C$ and its incident inter-cluster edges together with a new $1$-expander graph of $\overline C$.

Next, for every inter-cluster edge $e'$ of $\overline H_{\mathcal C}$ that was temporarily deleted and for each $b$, we run \texttt{Nibble} from every vertex in $\mathcal S(s,\overline H_{\mathcal C})$ and $\mathcal S$- and $\mathcal E$-sets are updated accordingly as described earlier.

\paragraph{Correctness:}
In the previous subsection, it sufficed to define the $\mathcal S$- and $\mathcal E$-sets only w.r.t.~inter-cluster edges since clusters remained fixed over all updates. We claim that this still suffices in this subsection. To see this, note that a call to \texttt{Nibble} visits an edge $e'$ of a $1$-expander graph iff it visits an inter-cluster edge incident to this edge. Hence, if $e'$ is deleted or if $e'$ is a new edge, there is an inter-cluster edge incident to $e'$ which is temporarily deleted in the above procedure. Therefore, \texttt{Nibble} is rerun from every starting vertex that is affected by the deletion or insertion of $e'$. It follows that $\mathcal S$- and $\mathcal E$-sets are correctly maintained.

Since updates to clusters happen independently of the random bits used to form the graphs $H_{\mathcal C}$ and $\overline H_{\mathcal C}$, it follows that \texttt{XPrune} remains correct.

So far, we have assumed that at all times, clusters have size between $\kappa$ and $3\kappa$. Our correctness and performance analysis in this section rely crucially on this property. By Invariant~\ref{Inv:Clustering}, it may happen that a cluster $C$ has size less than $\kappa$. We modify \texttt{XPrune}$(e)$ so that for each such $C$ formed in line $1$, we remove it from $\mathcal C(X)$ and add it as part of the set of clusters output in line $10$. By Invariant~\ref{Inv:Clustering}, $C$ is disconnected from the rest of $G[X]$ so the cut $(C,X - C)$ is independent of any random bits used to form the graphs in $\mathcal H$; hence, unlike in lines $5$ and $6$, we do not need to remove the current graph from $\mathcal H$ so Lemma~\ref{Lem:DeleteEdgeT} still holds. Also note that no $\mathcal L$-list and no $\mathcal S$- or $\mathcal E$-set need to be updated when $C$ is removed in line $1$.

The modification to line $1$ ensures that in lines $2$ to $9$, every cluster has size between $\kappa$ and $3\kappa$ as desired and correctness of \texttt{XPrune} follows.

\paragraph{Performance:}
Since the above modification to \texttt{XPrune} makes no changes to the preprocessing step, it suffices to bound the update time.

The only change to \texttt{XPrune}$(e)$ is in line $1$. Since only $O(1)$ clusters are changed, updating the $\mathcal L$-lists accordingly does not take asymptotically more time than before. The upper bound on the time spent by \texttt{Nibble} to maintain the $\mathcal S$- and $\mathcal E$-sets in the while-loop clearly is also an upper bound on the time spent by \texttt{Nibble} in line $1$, again since only $O(1)$ clusters are affected. By Lemma~\ref{Lem:SimpleExpander}, it takes $O(\kappa)$ time to compute $1$-expander graphs for the new clusters which is also within our previous update time bound.

We conclude that the new update time is asymptotically the same as in the previous subsection.

\section{Concluding Remarks}\label{sec:ConclRem}
We have given a Las Vegas data structure for fully-dynamic MSF which w.h.p.~handles an update in $O(n^{1/2 - c})$ worst-case time for some constant $c > 0$ where $n$ is the number of vertices of the graph. This is the first improvement over the $O(\sqrt n)$ worst-case bound of Eppstein et al.~\cite{Sparsification}. Previously, such an improvement was not even known for the problem of maintaining a spanning forest of an unweighted fully-dynamic graph. We also obtain the first Las Vegas data structure for fully-dynamic connectivity with worst-case update time polynomially better than $O(\sqrt n)$; this data structure has $O(1)$ worst-case query time.

By breaking this important barrier for fully-dynamic MSF, our hope is that further progress can be made for this problem as well as for fully-dynamic connectivity. We also hope that our techniques are applicable to other dynamic graph problems. Dynamic global minimum cut may be one such problem, especially given that the recent deterministic near-linear time algorithm for the static version of the problem~\cite{StaticMinCut} exploits properties related to low-conductance cuts.

We leave two open problems for dynamic MSF, namely can polylogarithmic update time be achieved w.h.p., thereby matching the best known amortized update time bounds, and can the worst-case update time of the type in this paper be matched deterministically?

\end{document}